\documentclass[12 pt, a4paper]{article}
\RequirePackage[OT1]{fontenc}
\RequirePackage{amsthm,amsmath}
\usepackage{framed,enumitem} 
\usepackage{natbib}
\usepackage{xr-hyper}
\RequirePackage[colorlinks,citecolor=blue,urlcolor=blue]{hyperref}
\usepackage{amsthm}
\usepackage{mathtools}
\usepackage{bbm}
\usepackage{natbib}
\usepackage{textpos, tikz}
\usepackage{bm}
\usepackage{graphicx}
\usepackage{subfig}
\usepackage{epsfig}
\usepackage{float}
\usepackage{bbm}
\usepackage{amssymb}
\usepackage{parskip}
\usepackage{multirow}
\usepackage{pdflscape}
\usepackage{pdfpages }
\usepackage{rotating}
\usepackage{booktabs}
\usepackage{lmodern}
\usepackage{multirow}
\usepackage{verbatim}

\usepackage[titletoc,title]{appendix}
\usepackage{caption}
\usepackage{parskip}
\usepackage{siunitx}
\usepackage{tikz}
\usepackage{minibox}
\usepackage{titlesec}
\usepackage{algorithm}
\usepackage{algpseudocode}
\usepackage{array}
\newcolumntype{H}{>{\setbox0=\hbox\bgroup}c<{\egroup}@{}}
\usepackage{longtable}
\usepackage{enumitem}
\usepackage[margin=1in]{geometry}
\usetikzlibrary{matrix}

\newcommand{\comm}[1]{}
\newtheorem{thm}{Theorem}

\newtheorem{lem}[thm]{Lemma}
\theoremstyle{remark}
\newtheorem*{remark}{Remark}
\theoremstyle{definition}

\newcommand{\norm}[1]{\left\lVert#1\right\rVert}

\newtheorem{assumption}{Assumption}
\newtheorem{assumpB}{Assumption}
\newtheorem{assumpA}{Assumption}

\newcommand{\X}{\textbf{X}}

\newcommand{\pro}{\mathbb{P}}

\title{A generalized likelihood based Bayesian approach for scalable joint regression
and covariance selection in high dimensions}
\author{Srijata Samanta, Kshitij Khare and George Michailidis}

\begin{document}
	\maketitle
	
	\begin{abstract}
The paper addresses joint sparsity selection in the regression coefficient matrix and the error precision 
(inverse covariance) matrix for high-dimensional multivariate regression models in the Bayesian paradigm. The 
selected sparsity patterns are crucial to help understand the network of relationships between the 
predictor and response variables, as well as the conditional relationships among the latter. While Bayesian
methods have the advantage of providing natural uncertainty quantification through posterior inclusion 
probabilities and credible intervals, current Bayesian approaches  
either restrict to specific sub-classes of sparsity patterns and/or are not scalable to settings with 
hundreds of responses and predictors. Bayesian approaches which only focus on estimating the posterior mode
are scalable, but do not generate samples from the posterior distribution for uncertainty quantification. Using a 
bi-convex regression based generalized likelihood and spike-and-slab priors, we develop an algorithm called
Joint Regression Network Selector (JRNS) for joint regression and covariance selection which (a) can 
accommodate general sparsity patterns, (b) provides posterior samples for uncertainty quantification, and 
(c) is scalable and orders of magnitude faster than the state-of-the-art Bayesian approaches providing 
uncertainty quantification. We demonstrate the statistical and computational efficacy of the proposed 
approach on synthetic data and through the analysis of selected cancer data sets. We also establish high-dimensional 
posterior consistency for one of the developed algorithms.

\end{abstract}

\section{Introduction}

\noindent
We consider joint variable and precision matrix selection in 
high-dimensional multivariate regression models with multiple responses. In particular, we consider two sets of 
variables: the $n 
\times p$ matrix $X$ whose rows ${\bf x}_1, \cdots, {\bf x}_n$ $ \in \mathbb{R}^p$ comprise of $n$ samples on $p$ predictor variables 
 and the $n \times q$ matrix $Y$ whose rows  ${\bf y}_1, \cdots, {\bf y}_n \in \mathbb{R}^q$ 
comprise of $n$ matched samples (same set of entities) on $q$ response variables. We are interested in inferring a graphical model on the variables from the $Y$ data, while {\em accounting for} the effect of the $X$ data. The corresponding  multivariate regression model is given by
\begin{equation} \label{multregmdl}
Y = XB + \boldsymbol{\varepsilon}
\end{equation}
\noindent
where $\boldsymbol{\varepsilon}$ is an $n \times q$ matrix whose rows $\boldsymbol{\varepsilon}_1, \cdots, \boldsymbol{\varepsilon}_n \in \mathbb{R}^q$ comprise of the $n$ noise vectors and $B$ is a $ p \times q$ matrix of regression coefficients. To make 
things concrete, assume that ${\bf y}_1, {\bf y}_2, \cdots, {\bf y}_n$ are independent, and 
$$
{\bf y}_i \sim \mathcal{N}_q \left( B^T {\bf x}_{i}, \Omega^{-1} \right) \mbox{ for } i = 1,2, \cdots, n. $$

\noindent
This is equivalent to assuming that the noise vectors $\boldsymbol{\varepsilon}_1, \cdots, \boldsymbol{\varepsilon}_n$ are i.i.d. $\mathcal{N}_q ({\bf 0}, \Omega^{-1})$. The $q \times q$ matrix 
$\Omega$ captures the dependence between the response variables conditional on the 
predictor variables, and the $p \times q$ matrix $B$ captures the effect of the 
predictor variables on the response variables. 

For example, in molecular biology applications, multiple Omics modalities are profiled on the same set of samples. Then, following the central dogma of biology, the predictor variables $X$ could correspond to the DNA level (e.g., copy number or methylation data), while the response variables $Y$ to the transcriptomic level (mRNA expression). Another possibility is that the predictors correspond to the transcriptomic level and the responses to the proteomic level.
Thus, the regression coefficients in $B$ encode transcriptional or translational dependencies, while the entries of $\Omega$ reflect statistical associations within a molecular compartment.

We focus on the problem under a high-dimensional setting, wherein $p$ and/or $q$ is larger than or comparable to the sample size $n$. In such sample starved settings, imposing sparsity in $B$ and $\Omega$ offers a simple and effective approach for reducing the effective number of parameters. The sparsity patterns in $B$ and $\Omega$ often have specific scientific 
interpretations and can help researchers understand the underlying relationships between 
variables in the data set. To summarize, our goal is simultaneous sparse estimation 
of $B$ and $\Omega$, and the use of estimated sparsity patterns to understand 
relevant dependence structures. 

The above problem has been studied in the literature. On the frequentist side,  
various penalized likelihood based methods have been proposed. Many of these methods 
use an $\ell_1$ penalty that encourages sparsity both in $B$ and $\Omega$. Various 
optimization algorithms have been employed; see 
 \cite{friedman2008sparse}, \cite{rothman2010sparse}, \cite{lee2012simultaneous}, \cite{cai2013covariate}, \cite{lin2016penalized} and references therein. Note that the conditional log-likelihood for response $Y$ given $X$ can be written as
\begin{equation} \label{ceq1}
\begin{split}
    &\quad \ell (Y|X, B, \Omega) \\
    &= \text{constant}+\frac{n}{2}\log\det\Omega\\
    &\quad \; -\frac{1}{2} \text{tr} \left( \Omega \sum_{i=1}^n ({\bf y}_i - B^T {\bf x}_i) ({\bf y}_i - 
	B^T {\bf x}_i)^\top \right),
\end{split}
 \end{equation}

\noindent
and is not \textit{jointly convex} in $B$ and $\Omega$. Hence, many popular algorithms (e.g., block coordinate descent) for optimizing a penalized version of this log-likelihood may fail to converge to the global optimum, especially in settings where $p>n$, as pointed out in \cite{lee2012simultaneous}. However, the log-likelihood is \textit{bi-convex}, i.e., it is convex in $B$ for fixed $\Omega$ and in $\Omega$ for fixed $B$. The bi-convexity is leveraged in \cite{lin2016penalized} to develop a two-block coordinate descent algorithm which converges to a stationary point of the objective function assuming that all iterates need to be within a ball of certain radius $R(p,q,n)$ (a function of the model dimensions $p, q$ and the sample size $n$) that in addition contains the true data generating parameters. This condition is then shown to hold with high probability. 

Some recent papers such as \cite{Sohn:Kim:2012}, \cite{Yuan:Zhang:2014}, \cite{McCarter:Kim:2014} 
consider an alternate parameterization $(\tilde{B}, \Omega)$ where $\tilde{B} = 
-B \Omega$. The likelihood can be shown to be jointly convex in $(\tilde{B}, \Omega)$, and the respective algorithms in these papers provide sparse estimates of $\tilde{B}$ and $\Omega$ by using appropriate $\ell_1$ penalties. However, sparsity in $\tilde{B}$ does not in general correspond to sparsity in $B$. In many applications, the linear model with the $(B, \Omega)$ parameterization is the natural modeling tool, and sparsity in the regression coefficient matrix $B$ has a specific scientific interpretation. This interpretability may be lost if sparsity is instead imposed on  $\tilde{B}$. See \cite[Section 5]{lin2016penalized} for a detailed discussion. 

Bayesian methods offer a natural framework for addressing \textit{uncertainty 
quantification} of model parameters through the posterior distribution, and several 
Bayesian approaches have also been proposed in the literature. 
\cite{BVF:1998} propose a Bayesian approach for the joint estimation of $B$ and $\Omega$, but restrict the sparsity pattern in $B$ to be such that each row 
of $B$ is completely sparse or completely dense. \cite{RBR:2010} 
allow for a general sparsity pattern in $B$, but restrict $\Omega$ to be a diagonal 
matrix. \cite{bhadra2013joint} use spike-and-slab prior distributions to 
induce sparsity in $B$ (conditional on $\Omega$), and $G$-Wishart prior distributions coupled with
independent Bernoulli priors on the sparsity pattern in $\Omega$. Similar to 
\cite{BVF:1998}, they restrict the rows of $B$ to be completely sparse or completely 
dense, and also restrict the sparsity pattern in $\Omega$ to correspond to a 
\textit{decomposable graph}. In a related work \cite{CRP:2017}, the authors develop an 
objective Bayesian approach for Directed Acyclic Graph estimation in the presence of covariates. This approach induces sparsity in the Cholesky factor of $\Omega$ and corresponds to directly inducing sparsity in $\Omega$, when the underlying sparsity pattern is decomposable. In a recent work, \cite{deshpande2019simultaneous} propose a scalable Bayesian approach using spike-and-slab Laplace prior distributions to induce sparsity in $B$ and $\Omega$. The work employs an Expectation Conditional Maximization algorithm to find the (sparse) posterior mode, and thereby obtain sparse estimates of $B$ and $\Omega$. However, methods to generate samples from the posterior distribution are not explored, and hence uncertainty quantification in the form of posterior credible regions/intervals is not available. In \cite{li2021joint} the authors propose a Gaussian likelihood based fully Bayesian procedure for the simultaneous estimation of the mean vector and the inverse covariance matrix which provides measures of uncertainty. However, in moderate/high dimensional settings it might run into scalability issues as pointed out in Section \ref{Sim:results}.

Note that in most of the above cited literature, and in this paper, two layers of 
variables are considered. The matrix $B$ captures the effect of the top layer 
(predictors) on the bottom layer (responses), while the matrix $\Omega$ captures the 
conditional covariance structure of the bottom layer. It is possible to consider a 
scenario where we have a chain of multiple layers of variables, each layer affecting 
the layer below it (see the setting in \cite{lin2016}), thereby giving rise to multiple pairs of $B$ and $\Omega$ matrices. Due to the factorization of the likelihood based on the Markov property induced by the chain structure, any method developed for the two-layer setting can be extended in a reasonably straightforward way to the multiple layer setting, subject to some model parameter identifiability restrictions (see Section 3.4 in \cite{lin2016}).

In \cite{ha2020bayesian}, the authors consider the multiple layer setting, and develop a generalized 
likelihood based Bayesian approach for simultaneous sparse estimation of the $B$ and $\Omega$ pair. In 
\cite{PWZZ:2009} and \cite{KOR:2015}, the use of a regression based generalized likelihood has been shown 
to significantly improve the computational efficiency compared to Gaussian likelihood based methods for 
standard graphical models (no presence of predictors). In \cite{ha2020bayesian} sparsity inducing 
spike-and-slab prior distributions are used for the entries of $B$ and $\Omega$, and an MCMC algorithm 
(called BANS) based on add/delete/swap moves in the space of sparsity patterns is 
developed to generate approximate samples from the posterior distribution. In 
simulation experiments, scenarios with up to ten layers of variables, and up to 20 
variables in each layers are considered. However, the algorithm starts to run into 
serious computational issues when the number of variables in each layer is bumped up 
to 200 (with two layers). One reason for this is the need for several matrix 
inversions to compute Metropolis-Hastings based rejection probabilities in each 
iteration of the BANS algorithm (see Sections \ref{comparison} and \ref{Sim:results} 
for more details). A faster algorithm called BANS-parallel, wherein computations 
corresponding to each response variable can be parallelized, has also been developed 
in \cite{ha2020bayesian}. However, this approach ignores the symmetry in $\Omega$ 
which negatively affects the quality of the estimates, and again requires matrix 
inversions for various Metropolis-Hastings steps (see Remark \ref{rem:bans:parallel} at the end of Section \ref{jointmodel}). In short, existing Bayesian approaches suffer from at least one of the following drawbacks:
(i) restrict to a subclass of sparsity patterns; (ii) focus on estimating the posterior mode and not on sampling from the posterior distribution; (iii) are not computationally scalable due to excessive use of matrix inversions.

The goal of this paper is to develop a \textit{computationally scalable} generalized likelihood based 
\textit{Bayesian} procedure for joint regression and precision matrix selection, which  can account for 
\textit{arbitrary sparsity patterns} in $B$ and $\Omega$, and provide uncertainty quantification. First, we leverage ideas in \cite{KOR:2015} for standard graphical models (no predictors) to the current setting, and construct a regression based generalized likelihood that is bi-convex in $\Omega$ and $B$. The generalized likelihood in \cite{ha2020bayesian}, which corresponds to the predictor adjusted version of the generalized likelihood in \cite{PWZZ:2009}, is \textit{neither jointly convex, nor bi-convex}. In the standard graphical model setting, it has been demonstrated in \cite{KOR:2015} that convexity plays an important role in improved algorithmic and empirical performance of the generalized likelihood as compared to the one in \cite{PWZZ:2009}. Next, we develop a Gibbs sampling algorithm (referred to as the \textit{joint algorithm}) to sample from the corresponding posterior (using spike-and-slab priors for entries of $B$ and $\Omega$). With entry-wise updates of $B$ and 
$\Omega$ involving standard distributions, we completely avoid the matrix inversions 
needed for the Metropolis-Hastings steps in BANS and the resulting algorithm is significantly computationally faster. As an illustrative example, with $200$ responses and $200$ predictors, the proposed MCMC algorithm, coded in \texttt{R/Rcpp}, completes 3000 iterations (each iteration cycles through all the entries of $B$ and $\Omega$) in less than 5 minutes. In the same setting, BANS (implemented using \texttt{R/Rcpp} code available on \texttt{Github}) was only able to finish less than 100 iterations in 4 days. 

Several frequentist methods in the literature, such as those in 
\cite{lee2012simultaneous} and \cite{cai2013covariate}, consider a step-wise approach for sparse estimation of $B$ and $\Omega$. In this approach, $q$ 
regressions corresponding to each of the responses are used to obtain sparse 
estimates of columns of $B$. The resulting estimate of $B$ is used to compute plug-in 
covariate adjusted responses, subsequently provided to a standard graphical model 
estimation procedure to obtain a sparse estimate of $\Omega$. As a Bayesian analog of
this approach, and for faster computation, we develop a \textit{step-wise algorithm} for joint regression/covariance selection. In this 
approach, we first focus on estimation/selection for $B$ by treating $\Omega$ as a 
diagonal matrix and combining the resulting Gaussian likelihood with spike-and-slab 
priors on entries of $B$. An appropriate posterior estimate $\hat{B}$ of $B$ is then 
used to compute covariate-adjusted responses (or pseudo-errors) $\hat{\boldsymbol{\varepsilon}}_i = 
{\bf y}_i - \hat{B}^T {\bf x}_i$. In the second step of the algorithm, the generalized 
likelihood of \cite{KOR:2015}, along with spike-and-slab priors on the entries of 
$\Omega$ is used for selecting the sparsity pattern in $\Omega$. The computational 
advantage obtained by ignoring the cross-correlation in the responses for the $B$ 
estimation step clearly comes with the cost of some loss of statistical efficiency. 
However, we rigorously establish high-dimensional posterior model selection and 
estimation consistency of the resulting estimates in 
Section \ref{high:dimensional:consistency}.
As expected, the simulation experiments in Section \ref{Sim:results} in general 
demonstrate a loss in statistical accuracy and roughly two times improvement in 
computational performance as compared to the joint algorithm. 

The remainder of the paper is organized as follows. The joint and the step-wise algorithms are developed in Sections \ref{jointmodel} and \ref{stepwisemodel}, respectively. High-dimensional posterior consistency results for the step-wise algorithm are provided in Section \ref{high:dimensional:consistency}. An extensive simulation study evaluating the empirical performance of the proposed algorithms is presented in Section \ref{Sim:results} and an analysis of a cancer data set is presented in Section \ref{data:analysis}. Proofs of the technical results along with additional simulation details are provided in a Supplementary document. 

\section{Joint sparsity selection for $B$ and $\Omega$ using a bi-convex generalized 
likelihood} \label{jointmodel}

\noindent
We develop a generalized likelihood based Bayesian approach for 
jointly estimating the sparsity patterns in $B$ and $\Omega$. Consider the 
log-likelihood denoted by $\ell (Y|X,B,\Omega)$ in  (\ref{ceq1}). One reason why block updates 
corresponding to optimization/MCMC algorithms for the corresponding penalized objective 
functions/posteriors run into computational issues, even in moderate dimensional 
settings, is the presence of the $\log \det \Omega$ term, which leads to expensive 
matrix inverse computations. Let ${\bf y}_{.j}$ denote the 
$j^{th}$ column of the $n \times q$ data matrix $Y$. Hence, ${\bf y}_{.j}$ is the  
collection of all the $n$ observations corresponding to the $j^{th}$ response. 
Then, the conditional density of ${\bf y}_{.j}$ given all the other responses (and of 
course, conditional on $X$) is given by 

\begin{align} \label{condj}
&\left(\frac{\omega_{jj}}{2 \pi}\right)^{n/2} \exp\Big\{ - \frac{\omega_{jj}}{2} \| ({\bf y}_{.j} - X B_{\cdot j}) \nonumber \\ &\quad+ \sum_{k \neq j} \frac{\omega_{kj}}{\omega_{jj}} ({\bf y}_{.k} - X B_{\cdot k})\|_2^2 \Big\} \nonumber\\
&= \left(\frac{\omega_{jj}}{2 \pi}\right)^{n/2} \exp\left\{- \frac{1}{2\omega_{jj}} \left\| (Y - XB) \Omega_{\cdot j} \right\|_2^2 \right\}, \end{align}

\noindent
where $B_{\cdot j}$ and $\Omega_{.j}$ denote the $j^{th}$ columns of $B=((b_{jk}))$ and $\Omega=((\omega_{kl}))$ respectively, and $\|{\bf a}\|_2^2 = 
{\bf a}^T {\bf a}$. The above follows by noting 
that when regressing the $j^{th}$ variable against the other variables, the 
regression coefficient of the $k^{th}$ variable is given by 
$-\omega_{jk}/\omega_{kk}$. Using ideas in \cite{besag1975statistical}, a 
generalized likelihood for $(B, \Omega)$ can be defined by taking the product of 
these conditional densities. This is the regression based generalized likelihood 
used for the BANS algorithm in \cite{ha2020bayesian}, and its form is given by 
\begin{equation} \label{gBANS}
\begin{split}
   & L_{g, BANS} (Y \mid X, B, \Omega) \\ &= \left( \prod_{j=1}^q 
\frac{(\omega_{jj})^{n/2}}{(2 \pi)^{n/2}} \right)\\ &\quad \quad \times \exp\left\{-\sum_{j=1}^q 
\frac{1}{2\omega_{jj}} \left\| (Y - XB) \Omega_{\cdot j} \right\|_2^2 \right\}.
\end{split}
\end{equation}

\noindent
Taking logarithm of the expression in (\ref{gBANS}), shows that the problematic $\log 
\det \Omega$ term in the log-Gaussian likelihood (\ref{ceq1}) is now replaced by a 
much simpler $\sum_{j=1}^q \log \omega_{jj}$ term in the generalized log-likelihood. 
However, the bi-convexity is lost, i.e., given $B$, the function $\log L_{g, BANS}$ 
is not convex in $\Omega$. In the simpler setting of Gaussian graphical models with 
no predictors (i.e., no $B$), it was shown in \cite{KOR:2015} that this lack of 
convexity can lead to severe convergence issues (in a penalized optimization context)
and a convex version of the generalized likelihood was constructed. We adapt this 
idea in the more general setting of joint regression and precision matrix estimation in a Bayesian context. 

In particular, by `weighting' each observation with $\omega_{jj}^{-\frac{1}{2}}$ for 
the expression in (\ref{condj}), i.e., using the conditional density of 
$\omega_{jj}^{-\frac{1}{2}} {\bf y}_{.j}$, and then taking the product over every 
$1 \leq j \leq p$, we get the generalized likelihood 
\begin{equation} \label{joint}
\begin{split}
&\tilde{L}_{g, joint} (Y|X, B, \Omega) \\ &=  \prod_{j=1}^q \frac{(\omega_{jj})^n}{{(2 \pi)}^{n/2}} \\ &\quad \quad \times \exp\left\{-\sum_{j=1}^q 
\frac{1}{2} \left\| (Y - XB) \Omega_{\cdot j} \right\|_2^2 \right\}. 
\end{split}
\end{equation}

\noindent
Next, we discuss some important features related to $\tilde{L}_{g, joint}$ and its use for 
Bayesian inference. 
\begin{itemize}
    \item The exponent in $\tilde{L}_{g, joint}$ now becomes a quadratic form in $\Omega$, 
    and the power of $\omega_{jj}$ is now $n$ instead of $n/2$ (as compared to 
    $L_{g, BANS}$). Hence, $\log \tilde{L}_{g, joint}$ is bi-convex (convex in $\Omega$ given
    $B$, convex in $B$ given $\Omega$) and in general analytically more tractable 
    than $\log L_{g, BANS}$.  
    \item Note that \textit{our primary goal, as far as $\Omega$ is concerned, is 
    sparsity selection}. Hence, following 
    \cite{meinshausen2006high}, \cite{PWZZ:2009}, \cite{KOR:2015}, we relax the \textit{constraint of positive definiteness} for $\Omega$ to the simpler constraint of just having \textit{positive diagonal entries}. This relaxation leads to significant improvement in computational scalability. If a positive definite estimate of $\Omega$ is needed for a downstream application, it can be obtained by a quick refitting step restricting to the selected sparsity pattern. The same relaxation of the positive definiteness constraint is also used for the BANS algorithm in \cite{ha2020bayesian}. Such a relaxation is not possible for the Gaussian likelihood because of the presence of the $\det \Omega$ term.  
    \item Although a generalized likelihood is not a probability density anymore, it 
    can still be regarded as a data based weight function, and as long as the product
    of the generalized likelihood and the specified prior density is integrable over 
    the parameter space, one can construct a posterior distribution and carry out Bayesian 
    inference (see \cite{BHW:2016, Alquier:2020} and the references therein). 
\end{itemize}

To induce sparsity, we use spike-and-slab prior distributions (mixture of point mass at zero and a normal density) for the entries of $B$ and the off-diagonal entries of $\Omega$, and
exponential priors for the diagonal entries of $\Omega$. Specifically, for $B= 
((b_{rs}))$ and $\Omega = ((\omega_{st}))$, we use the following priors: 
\begin{align*}
   b_{rs} \sim (1-q_1)  \delta_0 + q_1 N(0, \tau_1^2), &\;1\leq r \leq p,\\
   &\; 1\leq s \leq q,
   \end{align*}
   \begin{align*}
   w_{st} &\sim (1-q_2)  \delta_0 + q_2 N(0, \tau_2^2), \quad \; 1\leq s < t \leq q,\\
   \omega_{ss} &\sim \lambda \exp(-\lambda \omega_{ss}) ,\; 1\leq s \leq q,
\end{align*}

\noindent
where $b_{rs}$'s and $\omega_{st}$'s are independently distributed and $\delta_0$ 
denotes the distribution with its entire mass at 0. Further, the hyperparameters $q_1, q_2 
\in (0,1)$ denote the respective mixing probabilities for entries of $B$ and 
$\Omega$, and the hyperparameters $\tau_1^2, \tau_2^2$ are the respective prior 
slab variances. 

The resulting generalized posterior distribution \\$\pi_{g,joint}$ is 
intractable in the sense that closed form computation or direct sampling
is not feasible. However, straightforward calculations show that:
\begin{itemize}
\item the full conditional posterior distribution of each entry of $B$ (given all the
other parameters and the data) is a mixture of a point mass at zero and an 
appropriate normal density. For $1\leq r\leq p \;, 1 \leq s\leq q,$
\begin{align*}
    (b_{rs} | Y, B_{-(rs)}, \Omega) \; \sim (1-q_1^{*}) \delta_0 + q_1^{*} N\left(\frac{C_2}{C_1}, \frac{1}{C_1}\right) 
\end{align*}
 where
$$1-q_1^{*} = C_0 (1-q_1), $$
$$ C_0 = \left[ (1-q_1) + \frac{q_1}{\tau_1 \sqrt{C_1}} \exp \left(\frac{C_2^2}{2C_1}\right) \right]^{-1},$$
$$ C_1 = \sum_{k=1}^q\sum_{i=1}^n \omega_{sk}^2x_{ir}^2+\frac{1}{\tau^2_1},$$ 
\begin{align*}
    C_2 &= \sum_{k=1}^q\sum_{i=1}^n \omega_{sk}\left(\sum_{l=1}^q \omega_{lk}y_{il}\right)x_{ir}\\
    &-\sum_{k=1}^q\sum_{i=1}^n \omega_{sk}x_{ir}\left[\sum_{l\neq s}B_{.l}^{T}{\bf{x}}_{i}\omega_{lk}+\omega_{sk}\sum_{j\neq r}b_{js}x_{ij}\right].
\end{align*}
\item the full conditional posterior distribution of each off-diagonal entry of $\Omega$ (given all the other parameters 
and the data) is a mixture of a point mass at zero and an appropriate normal density. For $1\leq s<t\leq q,$
\begin{equation} \label{condomegao}
    (\omega_{st} | Y, \Omega_{-(st)}, B )\sim (1-q_2^{*}) \delta_0 + q_2^{*} N\left(\frac{-D_2}{D_1}, \frac{1}{D_1}\right)
\end{equation}
\noindent
where $$1-q_2^{*} = D_0 (1-q_2), $$
$$ D_0 = \left[ (1-q_2) + \frac{q_2}{\tau_2 \sqrt{D_1}} \exp\left(\frac{D_2^2}{2D_1}\right) \right]^{-1},$$
$$ D_1 = S_{ss} + S_{tt} + \frac{1}{\tau^2_2}, $$ 
$$D_2 = \sum_{l\neq s} \omega_{tl}S_{ls} + \sum_{l \neq t} \omega_{sl}S_{lt},$$
$$S = (Y-XB)^T(Y-XB).$$ 
\item The full conditional posterior density of each diagonal entry of $\Omega$ (given all the other parameters and 
the data) is given as 
\begin{equation} \label{condomegad}
\begin{split}
    &\pi_{g,joint}(\omega_{ss} | Y, \Omega_{-(ss)}, B)\\ 
    &\quad \propto \omega_{ss}^n \exp{[-\frac{1}{2}S_{ss}\{\omega_{ss}+ (\lambda+\sum_{i=1}^n\epsilon_{is}f_{is})/S_{ss}\}^2]}
    \end{split}
\end{equation}
where $$ \epsilon_{is} = y_{is}-B_{.s}^{T}{\bf{x}}_{i} \text{ and } f_{is}=\sum_{l\neq s}\omega_{ls}\epsilon_{il}.$$
This is an univariate density with the unique mode at
$$ \frac{\sqrt{(f_s(\lambda))^2 + 4nS_{ss}} - f_s(\lambda)}{2S_{ss}}. $$
where $f_s(\lambda) = \sum_{i=1}^n\epsilon_{is}f_{is}+ \lambda$. 
%

\end{itemize}

\noindent
These properties allow us to construct a Metropolis-within-Gibbs sampler, which we call the Joint 
Regression Network Selector (JRNS),  to sample from the joint 
generalized posterior density of $B$ and $\Omega$. One 
iteration of JRNS, given the current value of $(B, \Omega)$ is described in Algorithm 1 below. Essentially,
all entries of $B$ and off-diagonal entries of $\Omega$ are sampled from the respective full conditional 
distribution. A Metropolis-Hastings approach is used for the diagonal entries $\omega_{ss}$ of $\Omega$. In particular,
a proposal is generated from a normal density centered at the conditional posterior mode for $\omega_{ss}$.
The proposed value is accepted or rejected based on the relevant Metropolis based acceptance probability 
computed using the proposal normal density and the full conditional of $\omega_{ss}$.  A more detailed description of this algorithm is presented in Section \ref{suppsec:detailedAlgo1} of the Supplement.

\begin{algorithm}[H]
	\caption{Joint Regression Network Selector} 
	\begin{algorithmic}
	\Procedure{JRNS($B,\Omega,X,Y$)}{}
	   \For {$r=1$ \textbf{to} $p$}   \Comment{updating matrix $B$}
			\For {$s=1$ \textbf{to} $q$}
			\State Set $C_1,C_2, q_1^{*}$ 
			\State Sample $b_{rs}$ from the mixture distribution, $ (1-q_1^{*}) \delta_0 + q_1^{*} N(\frac{C_2}{C_1}, \frac{1}{C_1})$ 
			\EndFor
			\EndFor
			\State $E = (Y-XB)$
		    \State $S = E^TE$	
		\For {$s=1$ \textbf{to} $q-1$} \Comment{updating off-diagonals of $\Omega$}
		   \For {$t=s+1$ \textbf{to} $q$}
		   \State Set $D_1,D_2,q_2^{*}$
		   \State Sample $\omega_{st}$ from the mixture distribution, $ (1-q_2^{*}) \delta_0 + q_2^{*} N(-\frac{D_2}{D_1}, \frac{1}{D_1})$
	\EndFor
	\EndFor	
	\For {$s=1$ \textbf{to} $q$} \Comment{updating diagonals of $\Omega$}
	\State Set $f_s(\lambda)$ and compute mode 
	\State $v \gets N(\text{mode},0.001)$ \Comment{choosing proposed value}
	\State Calculate acceptance probability, $\rho$ using $\pi_{g,joint}(\omega_{ss} | Y, \Omega_{-(ss)}, B)$ and proposed value $v$
	\State Accept proposed value, $v$ with probability $\rho$
	\EndFor
	\State \textbf{return} $B$
	\State \textbf{return} $\Omega$
     \EndProcedure	
	\end{algorithmic}
	\label{algo1}
\end{algorithm}

\subsection{Sparsity selection and estimation using MCMC output} 
\label{MCMCoutput}

\noindent
The output $(B^{(i)}, \Omega^{(i)})_{i=1}^M$ from JRNS for an appropriate number $M$ of iterations (after the burn-in period), can be used as follows to estimate the sparsity patterns in the corresponding parameters. We follow the majority voting approach to construct such an estimate, wherein we 
include only those variables whose generalized posterior based marginal inclusion probabilities are at least $1/2$ (\cite{barbieri2004optimal}). Let $\gamma_{jk} = I(b_{jk} \neq 0)$ represent the sparsity indicator of $b_{jk}$ for $j = 1,2,\ldots,p $ and $k = 1,2,\ldots,q$. Then, $\gamma = ((\gamma_{jk}))$ represents the sparsity pattern in $B$. For each $(j,k)$ let $\hat{\pi}_{jk} = P(b_{jk} \neq 0|Y)$ which is approximated by
\begin{equation} \label{defn:MCMCaprrox}
    \frac{1}{M} \sum_{i=1}^{M} I(b_{jk}^{(i)} \neq 0).
\end{equation}
The quantity in (\ref{defn:MCMCaprrox}) is the proportion of iterations for which $b_{jk}^{(i)} \neq 0$ out of the $M$ iterations. If $\hat{\pi}_{jk} \geq 1/2$, the $(j,k)$-th entry is considered non-zero in the estimated sparsity pattern of $B$. It is to be noted that by Ergodic theorem the MCMC approximation given above in (\ref{defn:MCMCaprrox}) converges to the generalized posterior probability, $\hat{\pi}_{jk}$ of $b_{jk}$ being non-zero as $M \to \infty$. For large values of $M$, it is very close to $\hat{\pi}_{jk}$. Then, an estimate of $\gamma_{jk}$ is approximately obtained as 
\begin{align}
    \hat{\gamma}_{jk} =  
    \begin{cases}
      1, &    \text{if }\frac{1}{M} \sum_{i=1}^{M} I(b_{jk}^{(i)} \neq 0) \geq \frac{1}{2}\\
      0, & \text{ otherwise. }\\
    \end{cases} \nonumber
\end{align}
\noindent
A similar majority voting approach based on the generalized posterior 
based marginal inclusion probabilities can be used to estimate the 
sparsity pattern in $\Omega$. In particular, if $\eta_{rs} = 
I(\omega_{rs} \neq 0)$ represents the sparsity indicator for 
$\omega_{rs}$ for $1 \leq r < s \leq q$, then an estimate of $\eta_{rs}$
is approximately obtained as
 \begin{align}
    \hat{\eta}_{rs} =  
    \begin{cases}
      1, & \frac{1}{M} \sum_{i=1}^{M} I(\omega_{rs}^{(i)}\neq 0)  \geq \frac{1}{2}\\
      0, & \text{ otherwise. }\\
    \end{cases} \nonumber
\end{align}
Further, an estimate of the magnitudes of the selected non-zero entries of
$B$ can also be obtained as follows. If $\hat{\gamma}_{jk} = 1$, then 
$$
\hat{b}_{jk} = \frac{\sum_{i=1}^M b_{jk}^{(i)} I(b_{jk}^{(i)}\neq0)}
{\sum_{i=1}^{M} I(b_{jk}^{(i)}\neq0)}. 
$$
\noindent
An estimate of the magnitudes of the selected non-zero entries of $\Omega$ can also 
be obtained similarly. As stated earlier, the positive definiteness constraint
on $\Omega$ is relaxed for faster sparsity selection. An examination of the output of 
JRNS for many of our simulation settings in Section \ref{Sim:results} consistently 
revealed positive definite $\Omega$ iterates. However, there is no general 
guarantee that these iterates or the resulting estimate of $\Omega$ will be positive definite. 

If one wants to enforce positive definiteness, it can be achieved through a 
post-processing step (see \cite{lee2020post}) which focuses on the induced 
posterior of $h(\Omega)$, where 
$$
h(\Omega) = \begin{cases} \Omega & \mbox{if } \mbox{eig}_{\min} (\Omega) > \epsilon, \cr
\Omega + (\epsilon - \mbox{eig}_{\min} (\Omega)) I_q & \mbox{if } \mbox{eig}_{\min} (\Omega) \leq \epsilon. 
\end{cases}
$$
for some suitably chosen $\epsilon > 0$.

\noindent
Note that $h(\Omega)$ is guaranteed to be positive definite and has the exact 
same off-diagonal entries as $\Omega$. Hence, if $\{\Omega^{(r)}\}_{r=1}^{M}$ 
are the $\Omega$ components of the iterates produced by the JRNS or step-wise 
algorithm, sparsity selection, inclusion probabilities and credible intervals 
for the off-diagonal entries are unchanged if one uses 
$\{h(\Omega^{(r)})\}_{r=1}^{M}$ instead of $\{\Omega^{(r)}\}_{r=1}^{M}$. The 
transformation to $h(\Omega)$ only affects estimation of the diagonal entries 
$\{\omega_{ss}\}_{s=1}^q$. This additional eigenvalue check for computing $h(\Omega)$ 
takes $O(q^3)$ computations and hence does not change the computational complexity of 
JRNS (see (\ref{JRNS:computations:omega}) below), and marginally increases the 
wall-clock time (less than 5\% in all our simulation settings). 

Another approach to ensure positive definiteness is to use the {\it refitting}
idea from the penalized sparsity selection literature (see for example 
\cite{ma2016joint}). The estimators generated from penalized sparsity 
selection methods often suffer from (magnitude) bias issues, and one way to 
fix this is to obtain a constrained MLE of the desired parameter (by 
restricting to the the estimated sparsity pattern). Using this idea in our 
context, we compute the estimated sparsity pattern $\hat{\eta}$ in $\Omega$ 
and the regression coefficient matrix estimator $\hat{B}$ from the MCMC output
as described above. Now, we use the pseudo-errors (rows of $Y-X\hat{B}$) as 
approximate samples from a $\mathcal{N}_q ({\bf 0}, \Omega^{-1})$ 
distribution, and use the {\it glasso} function in {\it R} to compute the 
constrained MLE of $\Omega$ restricted to the sparsity pattern $\hat{\eta}$. 

\textbf{Hyperparameter selection:} Selecting hyperparameters is an important issue in any Bayesian approach. In sample-starved settings discussed in this paper, the choice of hyperparameters may have a significant impact on the resulting estimates (see the Supplementary section \ref{suppsec:hpselection} for more illustrations or details). A standard approach to choose the hyperparameters will be to use cross validation wherein we consider a grid of values for the hyperparameters and select the set of values based on the minimum prediction error. However, this method can be computationally expensive. If one does not have the computational resources or time to carry out the cross validation technique, another approach is to make some sensible objective choices as discussed below.

For JRNS, we have the prior mixture probabilities $q_1,q_2$, the prior slab variances $\tau_1^2$ and $\tau_2^2$  and $\lambda$ as hyperparameters. For $q_1$ and $q_2$ one can always consider a flat $U(0,1)$ prior in which case we will get a beta-update for $q_1$ and $q_2$ in every iteration of the Gibbs sampler. Another choice of  $q_1$ and $q_2$ which is motivated by the theoretical results in this paper and also in \cite{cao2019posterior}, \cite{Narisetty:He:2014} is to take $q_1 = 1/p$ and $q_2 = 1/q$. We use these choices in the simulation studies and obtain good results.  For $\tau_1^2$ and $\tau_2^2$ one may choose values around 1 or for a more principled choice one  may choose objective Inverse-Gamma priors with shape = $10^{-4}$ and rate = $10^{-8}$ as suggested in \cite{wang2012bayesian}. These will result in straightforward Inverse-Gamma updates for  $\tau_1^2$ and $\tau_2^2$ in each iteration. One may consider the Gamma prior with the same shape and rate values for $\lambda$ as well.

\subsection{JRNS and BANS: A computational cost comparison} \label{comparison}

\noindent
Next, we discuss the computational cost associated with the proposed JRNS 
algorithm, and compare it with the computational cost for the BANS algorithm in \cite{ha2020bayesian}. The structural differences in the generalized likelihoods used
by the two algorithms have been described in the discussion surrounding equations 
(\ref{gBANS}) and (\ref{joint}). As we describe below, there are also crucial 
differences between the two approaches at the computational level that lead to a
significant difference in overall computational costs. 
\begin{itemize}
    \item Algorithm \ref{supp:detailedAlgo1}, as described in Section C of the supplementary document, provides the detailed pseudo-code for one iteration of the JRNS algorithm. The matrix multiplications in Lines 2 and 3 take at most $O(pq^2 + qp^2)$ operations (computation of $X^T X$ and $X^T Y$ needs to be done only once prior to starting the iterations, and hence is not included). For each of 
    the $pq$ repetitions of the dual for loops in Lines 4 and 5, the most expensive 
    steps are the computation of $C_2$ in Line 12, which takes $O(q)$ operations, and
    the update of the $s^{th}$ row of $M_2$ in Line 20 which takes $O(p)$ operations.
    The computational cost of all the other steps does not depend on $n, p, q$ and involves $O(1)$ operations in all. Hence, the overall cost of Lines 4 to 22 is at most
     \begin{equation} \label{JRNS:computations:B}
    pq(O(p) + O(q)) = O(p^2q + pq^2) 
    \end{equation}
    
    \noindent
    operations. The matrix multiplications in Lines 23 and 24 need $O(npq + nq^2)$ operations. For each of the $\binom{q}{2}$ repetitions of the dual 
    for loops in Lines 25 and 26, the most expensive step is the computation of $D_2$ 
    in Line 33, which takes $O(q)$ operations. The computational cost of all the 
    other steps does not depend on $n, p, q$ and takes $O(1)$ operations in all. The 
    update of the diagonal entry in Lines 42 to 46 takes $O(q)$ operations and is 
    only repeated in the outer for loop. Also, the update of $\Omega^2$ in Line 49 requires $O(q^3)$ operations. Hence, the overall cost of Lines 23 to 49 is
    at most 
    \begin{equation} \label{JRNS:computations:omega}
    \binom{q}{2}O(q) + qO(q) + q^3 = O(q^3)
    \end{equation}
    \noindent
    The overall cost of one iteration of the JRNS algorithm can be obtained by 
    adding the values in (\ref{JRNS:computations:B}) and 
    (\ref{JRNS:computations:omega}). \textit{Note that this is an upper bound}, as the sparsity 
    in $B$ and $\Omega$ can reduce the cost of many vector/matrix products 
    in Algorithm~1. 
    
    \item The BANS algorithm
    \citep[Supplemental Section S3]{ha2020bayesian} uses a Metropolis-Hastings based 
    approach to do a neighbourhood exploration in the graph spaces for the sparsity 
    patterns in $B$ and $\Omega$. This algorithm is implemented in the 
    \texttt{ch.chaingraph} function available on \texttt{Github} with the Supplementary material for 
    \cite{ha2020bayesian}  In particular, the $\Omega$ update in each
    iteration of the BANS algorithm cycles through each response variable, and 
    proposes an add-delete or swap operation among its current 
    neighbors or non-neighbors. This proposal is accepted or rejected based on a 
    Metropolis-Hastings based probability. The non-zero entries in the appropriate 
    rows of $\Omega$ are then generated from relevant multivariate normal 
    distributions. To implement this procedure, the authors start by computing the 
    inverse of a $q \times q$ matrix (Line $73$ of \texttt{chaingraph.R} in 
    \cite{ha:supplemental:github}). The inversion requires $q^3$ operations. Since 
    this is done \textit{for all $q$ response variables}, the costs of these inversions 
    add up to $q^4$ operations. There are of course, additional costs to consider for
    the computation of the acceptance probability and multivariate normal sampling 
    described previously, which requires more albeit smaller matrix inversions and 
    matrix multiplications of its own. A similar approach and inversion of $p \times 
    p$ matrices for all $p$ predictor variables is needed in the $B$ update, which 
    leads a computational cost of $p^4$ operations (Lines 169-173 of 
    \texttt{chaingraph.R} in \cite{ha:supplemental:github}). The overall computational cost for one iteration of the BANS algorithm is therefore of the order of 
    $p^4 + q^4$. 
\end{itemize}
\noindent
The above analysis shows that each iteration of the JRNS algorithm is an order of 
magnitude faster than each iteration of the BANS algorithm. The multiple inversions 
in the BANS algorithm are probably the main reason for the computational issues 
encountered when both $p$ and $q$ are in the hundreds (see the simulation study in 
Section \ref{Sim:results} for more details). 

\begin{remark}\label{rem:bans:parallel}
There is a faster version of the BANS algorithm, called BANS-parallel,
which has been constructed by ignoring the symmetry in $\Omega$ to parallelize 
the computations for each row. While a similar parallel version can also 
be constructed for JRNS, we find that even without parallelization, JRNS 
is computationally faster than the BANS-parallel algorithm. 
For instance, in a 
simulation setting with $n=100,p=30,q=60$, 3000 MCMC iterations take 
around 50 seconds for the BANS-parallel as opposed to 5 seconds for the 
regular JRNS algorithm. Also, it is well known from the vanilla graphical models 
literature (see for example \cite{PWZZ:2009},\cite{KOR:2015}) that this 
non-symmetric approach can lead to statistical inefficiencies, and we do 
not pursue it further. 
\end{remark} 

\section{Step-wise estimation of the sparsity patterns of $B$ and 
$\Omega$} 
\label{stepwisemodel}

\noindent
Next, we present a computationally faster alternative to the JRNS procedure. As opposed to jointly estimating the 
sparsity patterns of $B$ and $\Omega$ using the generalized likelihood in (\ref{joint}), we first estimate the sparsity 
pattern in $B$ by looking at the $q$ individual regressions inherent in the multivariate regression model (\ref{multregmdl}), 
and use this to obtain an estimate of the sparsity pattern in $\Omega$. We provide a detailed description below. 

\noindent
{\bf Step 1: Estimating the sparsity pattern in $B$ using individual regressions}. Let ${\bf y}_{.j}$ denote the $j^{th}$ column of 
the response matrix $Y$, $B_{\cdot j}$ denote the $j^{th}$ column of the regression coefficient matrix $B$, and 
$\boldsymbol{\varepsilon}_{\cdot j}$ denote the $j^{th}$ column of the error matrix $\boldsymbol{\varepsilon}$, for $j = 1,2, \cdots, q$. Note that ${\bf y}_{.j}$ is the collection 
of the $n$ observations for the $j^{th}$ response variable, and the entries of $\boldsymbol{\varepsilon}_{\cdot j}$ are i.i.d. $N(0, \sigma_j^2)$, 
where $\sigma_j^2$ is the $j^{th}$ diagonal entry of $\Omega^{-1}$. The multivariate regression model $Y = XB + \boldsymbol{\varepsilon}$ in 
(\ref{multregmdl}) can be equivalently represented as a collection of the $q$ individual regressions 
\begin{equation} \label{individual}
{\bf y}_{.j} = X B_{\cdot j} + \boldsymbol{\varepsilon}_{\cdot j} \quad \mbox{ for } j = 1,2, \cdots, q. 
\end{equation}
\noindent
Clearly, the vectors ${\bf y}_{.1}, {\bf y}_{.2}, \cdots, {\bf y}_{.q}$ are dependent, and this dependence is precisely captured 
by the precision matrix $\Omega$. However, in this section, we will be agnostic to this dependence, and consider a 
generalized likelihood for $B, (\sigma_j^2)_{j=1}^q$ based on the product of the marginal densities of the vectors ${\bf y}_{.1}, 
{\bf y}_{.2}, \cdots, {\bf y}_{.q}$ as follows. 
\begin{equation} \label{genrealized:likelihood:individual}
\begin{split}
&\quad \widetilde{L}_{g, individual} (Y|X, B, (\sigma_j^2)_{j=1}^q) \\
& = \prod_{j=1}^q \left( \frac{1}{(2 \pi \sigma_j^2)^{n/2}} \right)\\
& \; \times \prod_{j=1}^q 
\exp\left\{- \frac{1}{2 \sigma_j^2} ({\bf y}_{.j} - X B_{\cdot j})^T ({\bf y}_{.j} - X B_{\cdot j})\right\}.
\end{split}
\end{equation}
We use spike-and-slab priors (mixture of point mass at zero and a normal density) for the entries of $B = ((b_{rs}))$, and 
Inverse-Gamma priors for $(\sigma_s^2)_{s=1}^q$. In particular for $1 \leq r \leq p, \; 1 \leq s \leq q,$
\begin{align*}
 \label{priors2}
   b_{rs} &\sim (1-q_1) \delta_0 + q_1 N(0, \tau_1^2 \sigma_{s}^2), \\
   \sigma_{s}^2 &\sim \text{ Inv-Gamma}(\alpha,\beta),
\end{align*}
\noindent
where $b_{rs}$'s and $\sigma_{s}^2$'s are independently distributed and $\delta_0$ 
denotes the distribution with a point mass at 0. Again, $q_1 \in (0,1)$ is a 
hyperparameter denoting the mixing probability for the spike-and-slab priors. 
The resulting generalized posterior distribution (denoted by $\pi_{g,individual}$) is again intractable in 
the sense that closed form computation or direct sampling is not feasible. However, straightforward 
calculations show that: 
\begin{itemize}
\item the full conditional posterior distribution of each entry of $B$ (given all the other parameters and 
the data) is a mixture of a point mass at zero and an appropriate normal density:
\begin{align*}
&(b_{rs} | Y, B_{-(rs)},\sigma_1^2,\ldots,\sigma_q^2)\\ &\quad \sim (1-q_1^{*}) \delta_0 + q_1^{*} 
N\left(\frac{C_2}{C_1}, \frac{\sigma^2_s}{C_1}\right) 
  \end{align*}
\noindent
where,
$$1-q_1^{*} = C_0(1- q_1),$$ 
$$ C_0 = \left[ (1-q_1) + \frac{q_1}{\tau_1 \sqrt{C_1}} \exp \left(\frac{C_2^2}{2\sigma^2_s C_1}\right) 
\right]^{-1},$$
$$ C_1 = \sum_{i=1}^{n}x_{ir}^2 + \frac{1}{\tau^2_1}, $$ 
$$ C_2 = \sum_{i=1}^{n}x_{ir}\left(y_{is}-\sum_{j\neq r}x_{ij}b_{js}\right).$$
\item The full conditional posterior distribution of $\sigma_s^2$ (given all 
the other parameters and the data) is again Inverse-Gamma:
\begin{equation}
(\sigma_s^2 | {\bf y}_{.s}, b_{1s},\ldots ,b_{ps}) 
\sim \text{Inv-Gamma} \left(\alpha^{*}, \beta^{*} \right),
\end{equation}
where
$$ \alpha^{*} = \alpha+\frac{n+|B_{.s}|}{2},$$
$$ \beta^{*} = \beta + \frac{{||{\bf y}_{.s}-XB_{.s}||}_2^2}{2}+ \frac{B_{.s}^TB_{.s}}{2\tau_1^2},$$
$$|B_{.s}| := \text{ number of non-zero entries in $ B_{.s}$}. $$
\end{itemize}
\noindent
These properties allow us to construct a Gibbs sampler to generate approximate
samples from the generalized posterior distribution of $(B, (\sigma_j^2)_{j=1}^q)$. We can construct an 
estimate $\hat{\boldsymbol \gamma}_{stepwise}$ of the sparsity pattern in  $B$ using the majority voting 
approach similar to the one mentioned in Section \ref{MCMCoutput}. An estimate $\hat{B}$ of $B$ can also be
obtained as follows. Let $B^{*}$ be a $p \times q$ matrix whose $k$-th column, $B_{.k}^{*}$ is given by the
posterior mean 
$$ E(B_{.k}|\gamma_{.k},Y)$$ 
which has a closed form expression given in Section A.3 of the supplementary document. Our estimate $\hat{B}_{stepwise}$ of $B$ is obtained from $B^{*}$, replacing $\gamma$ by its estimate $\hat{\gamma}_{stepwise}$. Alternatively, an estimate of $B$ can also be obtained using the Gibbs output in a similar manner as done for the JRNS approach towards the end of Section \ref{MCMCoutput}. For notational simplicity, in the rest of the paper, we will simply write $\hat{B}$ in place of $\hat{B}_{stepwise}$. 

Note that using the generalized likelihood denoted by $\widetilde{L}_{g, individual}$ 
amounts to simultaneously and independently estimating $q$ individual 
regressions with Gaussian errors. The Gibbs sampling approach in 
\cite[Section 7]{Narisetty:He:2014} for univariate regressions with spike-and-slab priors can potentially be used for each of the $q$ regressions. However, 
this approach again relies on first making appropriate moves in the space of 
sparsity patterns and then drawing the regression coefficient vector from the 
relevant multivariate normal distribution. With settings where $p$ and $q$ 
both are large in mind, we prefer to avoid the multivariate normal draws and 
instead use univariate mixture normal updates for each entry of $B$ as previously specified. 

\smallskip

\noindent
{\bf Step 2: Estimating the sparsity pattern in $\Omega$ using error estimates from Step 1}. Using the 
working estimate $\hat{B}$ from Step 1, we construct error estimates 
$$
\hat{\boldsymbol{\varepsilon}}_i = {\bf y}_i - \hat{B}^{T} x_i \mbox{ for } i = 1,2, \cdots, n. 
$$
\noindent
Let $\hat{\boldsymbol{\varepsilon}} = Y - X \hat{B}$ denote the $n \times q$ matrix with $i^{th}$ row given by $\hat{\boldsymbol{\varepsilon}}_i^T$ for $1 \leq i \leq n$. We know that the true errors ${\boldsymbol{\varepsilon}}_1, {\boldsymbol{\varepsilon}}_2, \cdots, {\boldsymbol{\varepsilon}}_n$ are i.i.d. $\mathcal{N}({\bf 0}, \Omega^{-1})$. Using the estimated error $\hat{\boldsymbol{\varepsilon}}_i$ as an approximation for the true error ${\boldsymbol{\varepsilon}}_i$, our task of estimating 
$\Omega$ is now reduced to a sparse precision matrix estimation problem. 
For this purpose, we use the generalized regression based likelihood for $\Omega$ by replacing $\boldsymbol{\varepsilon} = Y - XB$ in (\ref{joint}) by $\hat{\boldsymbol{\varepsilon}}$ as follows.
\\
\begin{equation} \label{generalized:likelihood:omega}
\begin{split}
&\widetilde{L}_{g, omega} (\hat{\boldsymbol{\varepsilon}}|\Omega) \\ &= \left( \prod_{j=1}^q \frac{(\omega_{jj})^n}{(2 \pi)^{n/2}} \right) \exp\left\{-\sum_{j=1}^q 
\frac{1}{2} \left\| \hat{\boldsymbol{\varepsilon}} \Omega_{\cdot j} \right\|_2^2 \right\}.
\end{split}
\end{equation}
We use spike-and-slab prior distributions (mixture of point mass at zero and a normal density) for the off-diagonal entries for $\Omega$, 
and exponential priors for the diagonal entries of $\Omega$. In particular for $1\leq s < t \leq q$, 
\begin{align*}
 \omega_{st} &\sim (1-q_2) \delta_0 + q_2 N(0, \tau_2^2),   \\
 \omega_{ss} &\sim \lambda \exp{(-\lambda\omega_{ss})},\quad \omega_{ss}>0.
\end{align*}
\noindent
The resulting generalized posterior distribution is intractable in the sense that closed 
form computation or direct sampling is not feasible. However, straightforward 
calculations show that the full conditional posterior distributions of the 
off-diagonal and diagonal elements of $\Omega$ are exactly as in 
(\ref{condomegao}) and (\ref{condomegad}) with $B$ replaced by $\hat{B}$ as 
needed. These properties allow us to construct a Gibbs sampler to generate 
approximate samples from the generalized posterior distribution of $\Omega$, which can 
further be used to construct an estimator $\hat{\boldsymbol \eta}_{stepwise}$ 
of the sparsity pattern of $\Omega$ using the majority voting approach in a 
similar manner as was done for $B$ in Step 1 of Method 2. 

The issue of hyperparameter selection is also important in this approach. In the Stepwise approach we have the Inverse-Gamma parameters $\alpha$  and $\beta$ from the prior on the diagonals of $\Omega^{-1}$ in Step 1 along with the other hyperparameters considered for the JRNS approach, namely the prior mixture probabilities $q_1,q_2$, the prior slab variances $\tau_1^2$, $\tau_2^2$  and $\lambda$. For the hyperparameters $q_1,q_2,\tau_1^2,\tau_2^2$ and $\lambda$ similar choices can be taken as in the JRNS algorithm. As for the prior distributions on the diagonals of $\Omega^{-1}$, one might consider the objective Inverse-Gamma priors with shape = $10^{-4}$ and rate = $10^{-8}$ as considered for $\tau_1^2$ and $\tau_2^2$.

\subsection{High dimensional selection consistency for the step-wise approach} \label{high:dimensional:consistency}

\noindent
We establish high-dimensional consistency of the stepwise procedure for estimation of the sparsity 
patterns of $B$ and $\Omega$ described in Section \ref{stepwisemodel}. We will consider a high-dimensional setting, where the number of responses $q$ and the number of predictors $p$ vary with $n$. Under the true model, the response matrix $Y$ is obtained as 
$$
Y = X B_0 + \boldsymbol{\varepsilon}, 
$$
\noindent
or, equivalently, 
$$
{\bf y}_i = B_0^T {\bf x}_i + \boldsymbol{\varepsilon}_i \quad \mbox{ for } i = 1,2, \cdots, n. 
$$
\noindent
The predictor vectors ${\bf x}_1, {\bf x}_2, \cdots, {\bf x}_n$ and the error vectors $\boldsymbol{\varepsilon}_1, \boldsymbol{\varepsilon}_2, \cdots, \boldsymbol{\varepsilon}_n$ are assumed to be i.i.d. $\mathcal{N}_p ({\bf 0}, R_0)$ and i.i.d. $\mathcal{N}_q ({\bf 0}, \Omega_0^{-1})$ respectively.
Since both $p$ and $q$ grow with $n$, the true parameters $B_0, \Omega_0, 
R_0$ also change with $n$, but we suppress this dependence for ease of exposition and remind the reader of this 
dependence as needed. Let the indicator matrices $\bf{\gamma}_t$ and 
$\bf{\eta}_t$ respectively denote the sparsity patterns in $B_0$ and 
$\Omega_0$ respectively, and $\mathbb{P}_0$ denote the probability measure 
underlying the true model. We define $\nu_{t_k}$ as the number of non-zero entries in $(\gamma_t)_{.k}$, 
the $k^{th}$ column of $\gamma_t$, $k_n  = \max\limits_{1\leq k\leq q} \nu_{t_{k}} + 1$, and $\delta_n = 
\sum_{k=1}^q \nu_{t_k}$. Under standard and mild regularity conditions on the eigenvalues of $R_0$ and the 
hyperparameters $q_1,q_2, \tau_1^2$ and $\tau_2^2$ (see Supplementary Sections \ref{suppsec:theorem1a} and 
\ref{suppsec:theorem1b} for details), the following consistency result can be established in a regime where
$pq$ essentially can grow sub-exponentially with $n$. 

\begin{thm}(Selection and Estimation Consistency of the generalized posterior) \label{jp:theorem1}
Suppose \\ $\frac{k_n^2\log(pq)}{n}\to 0$. Then,
\begin{enumerate}[label=(\alph*)]
    \item (Selection Consistency for $B$) Under Assumptions A1-A4 stated 
in Supplementary Section \ref{suppsec:theorem1a}, the (sequence of) sparsity pattern estimates 
$\hat{\boldsymbol \gamma}_{stepwise}$ 
for $B$ obtained from the step-wise approach satisfy 
$$
\mathbb{P}_0 \left( \hat{\boldsymbol \gamma}_{stepwise} = \bf{\gamma}_t \right) 
\rightarrow 1 \; \text{ as } n \to \infty.
$$
\item (Estimation Consistency for $B$) Under Assumptions A1-A4 stated 
in Supplementary Section \ref{suppsec:theorem1a}, the pseudo-posterior distribution on B concentrates 
around the truth at a rate of $\sqrt{\frac{\delta_n \log(pq)}{n}}$ (in Frobenius norm). In particular, 
{\small $$ 
\mathbb{E}_0 \left[ \Pi_{g,individual} \left( \|B - B_0\|_F > K \sqrt{\frac{\delta_nlog(pq)}{n}} 
\mid Y \right) \right] 
$$}

\noindent
converges to $0$ as $n \rightarrow \infty$ for a large enough constant $K$. 
\item (Selection Consistency for $\Omega$) Under Assumptions A1 - A4 and B1 - B4 stated 
in Supplementary Sections \ref{suppsec:theorem1a} and \ref{suppsec:theorem1b}, the (sequence of) sparsity pattern estimates $\hat{\boldsymbol \eta}_{stepwise}$
for $\Omega$ obtained from the step-wise approach satisfy
$$
\mathbb{P}_0 \left( \hat{\boldsymbol \eta}_{stepwise} = \bf{\eta}_t \right) 
\rightarrow 1 \; \text{ as } n \to \infty. 
$$
\end{enumerate} 
\end{thm}
\noindent
The proof of the above results leverages arguments in 
\cite{Narisetty:He:2014} and \cite{KOR:2015} for univariate spike-and-slab regression and 
standard graphical models with no covariates. However, some careful modifications and 
additional arguments are needed for the multivariate setting and the fact that 
pseudo-errors with an estimate $\hat{B}$ are being used in Step 2 of the 
step-wise approach. The proof is provided in Supplementary Sections  \ref{suppsec:theorem1a} and \ref{suppsec:theorem1b}.

\section{Performance Evaluation} \label{Sim:results}

\noindent
We evaluate the performance of the joint JRNS approach presented in Section \ref{jointmodel} and the step-wise approach in Section \ref{stepwisemodel} under diverse 
simulation settings. The data generating model is 
$Y = X B_0 + \boldsymbol{\varepsilon}$, where the $n$ rows of the error matrix $\boldsymbol{\varepsilon}$ are i.i.d. 
multivariate normal with mean vector ${\bf 0}$ and precision matrix
$\Omega_0$. We consider six different combinations
of the triplet $(n,p,q)$ along with the number of non-zero entries in 
$B_0$ and the off-diagonal part of $\Omega_0$ provided in 
Table \ref{npq:combination}. For each combination, the rows of $X$ are 
independently generated according to $\mathcal{N}_p(\textbf{0},R_0)$, 
where $R_0 = \left(0.7^{|j-k|}\right)_{j,k = 1}^p$. The non-zero 
entries of $B_0$ are drawn independently from a $U(1,2)$ distribution. Further, the non-zero entries of the off-diagonals of $\Omega_0$ are drawn 
independently from $U((-1,-0.5) \cup (0.5,1))$, while diagonal entries are drawn independently from a $U(1,2)$ distribution.  
\begin{table}[htbp]
\centering
\caption{Six different simulation settings with different $(n,p,q)$ combinations and number of true non-zero entries.}
\resizebox{0.95\textwidth}{!}{%
\begin{tabular}{|c|c|c|c|}
\hline
Combination & $(n,p,q)$ & Non-zeros in $B_0$ & Non-zeros in $\Omega_0$ (off-diagonals)\\
\hline
1 & $(100,30,60)$ & $p/5$ & $q/5$\\
2 & $(100,60,30)$ & $p/5$ & $q/5$\\
3 & $(150,200,200)$ & $p/5$ & $q/5$\\
4 & $(150,300,300)$ & $p/5$ & $q/5$\\
5 & $(100,200,200)$ & $p/30$ & $q/5$\\
6 & $(200,200,200)$ & $p/30$ & $q/5$\\
\hline
\end{tabular}
}%
\label{npq:combination}
\end{table}

\noindent
For each simulation setting in Table \ref{npq:combination}, we generate 200 replicated data sets to
evaluate the computational performance and selection 
accuracy with respect to $B$ and $\Omega$ of the proposed methods along with 
state-of-the-art Bayesian methods. Specifically, we compare the following methods: Joint (JRNS 
algorithm in Section \ref{jointmodel}), Stepwise (step-wise algorithm in 
Section \ref{stepwisemodel}), BANS (Bayesian node-wise selection algorithm
from \cite{ha2020bayesian}), DPE (Spike-and-slab lasso with dynamic 
posterior exploration from \cite{deshpande2019simultaneous}), 
DCPE (Spike-and-slab lasso with dynamic conditional posterior 
exploration from \cite{deshpande2019simultaneous}) and \textcolor{black}{HS-GHS (horseshoe-graphical horseshoe) from \cite{li2021joint}}. Note that any estimator 
obtained by maximizing a penalized likelihood can be interpreted as the posterior 
mode of an appropriate Bayesian model. The DPE and DCPE esitmators are essentially 
penalized likelihood estimators obtained by using spike and (Laplace) slab penalties 
for individual entries of $B$ and $\Omega$. In detailed simulations in 
\cite{deshpande2019simultaneous}, these methods are shown to provide significantly 
superior selection performance than the other penalized likelihood approaches such as 
MRCE \cite{rothman2010sparse} and CAPME \cite{cai2013covariate}, and we use them here as 
benchmarks for the selection performance of the proposed methods. Of course, these 
optimization based approaches do not generate samples from the posterior distribution
and can not provide uncertainty quantification in the form of posterior credible intervals/inclusion probabilities. The HS-GHS method of \cite{li2021joint} is a fully Bayesian approach based on the Gaussian likelihood.

The joint and step-wise methods were both run for 1000 burn-in iterations 
and then 2000 more follow-up iterations. The hyperparameters were chosen 
as described towards the end of Section \ref{MCMCoutput} (theoretically motivated choices for $q_1$ and $q_2$ and objective inverse-gamma priors for $\tau_1^2$ and $\tau_2^2$). We also consider learning $q_1$ and $q_2$ adaptively by using Beta hyperpriors on $q_1$ and $q_2$ and the results are presented in Tables 8 and 9 of Supplementary Section \ref{suppsec:hpselection}. We use traceplots and cumulative average plots to monitor and ensure the convergence of the MCMC. Some of these plots are provided in Figures \ref{fig:traceplot} and \ref{fig:avgplot}. The BANS algorithm was run using the default hyperparameter settings in \cite{ha2020bayesian} again with 
1000 burn-in and 2000 more follow-up iterations. DPE and DCPE are 
optimization algorithms for identifying the relevant posterior mode, and 
they were run with default settings provided in 
\cite{deshpande2019simultaneous}. 

\begin{figure*}
    \centering
    \includegraphics[width = 0.95\linewidth]{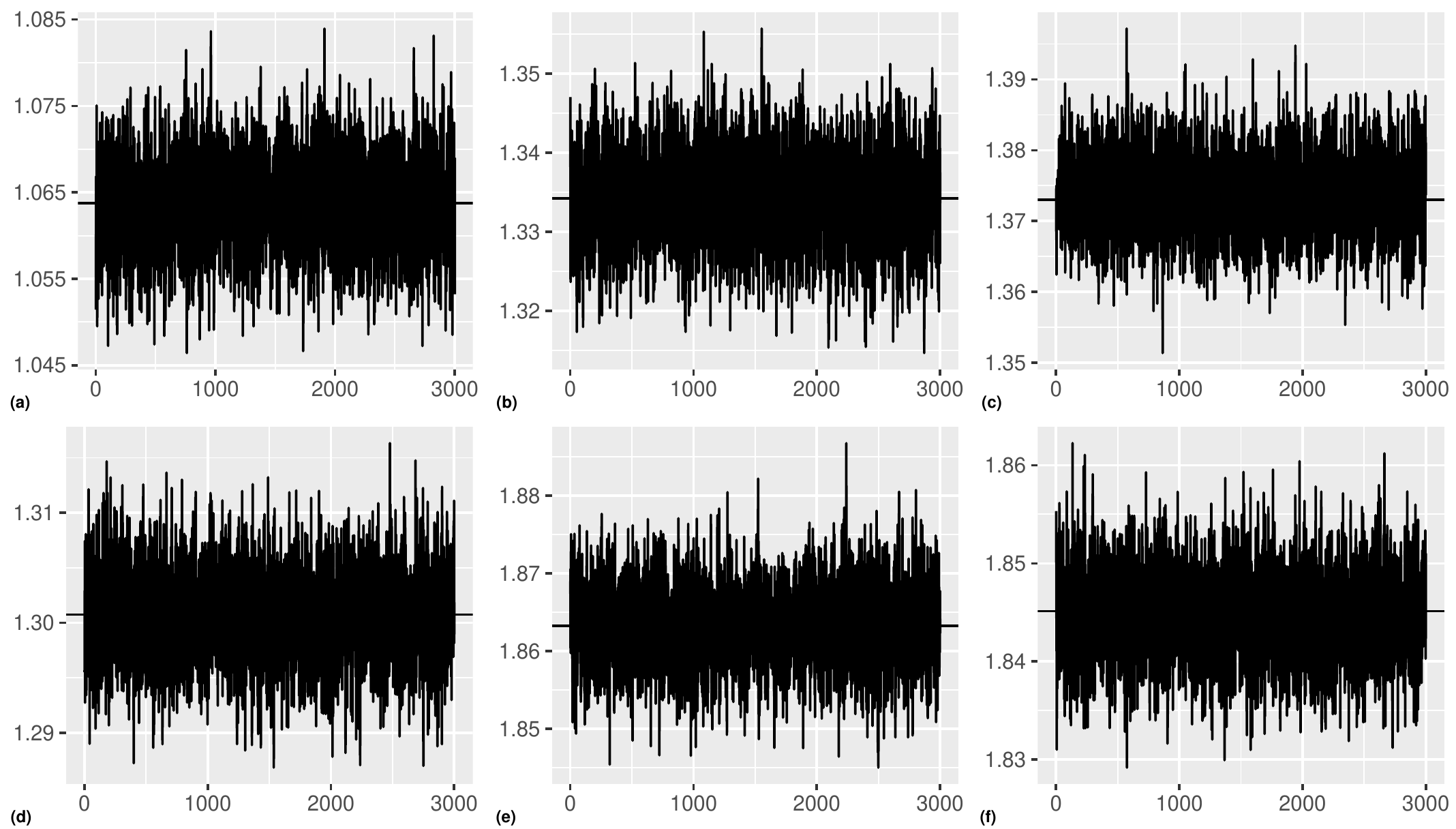}
    \caption{ Traceplots for a few randomly selected entries in $B$ when $(n,p,q)=(150,300,300)$ over the total 3000 Gibbs sampling iterations of the JRNS algorithm. The coordinates selected are (a) (162,37), (b)   (14,295), (c) (231,151), (d) (299,102), (e) (162,277), (f) (98,102). The black bold line represents the corresponding true value in $B_0$. These plots indicate sufficient mixing of the Markov chains.}
    \label{fig:traceplot}
\end{figure*}

\begin{figure*}
    \centering
    \includegraphics[width = 0.95\linewidth]{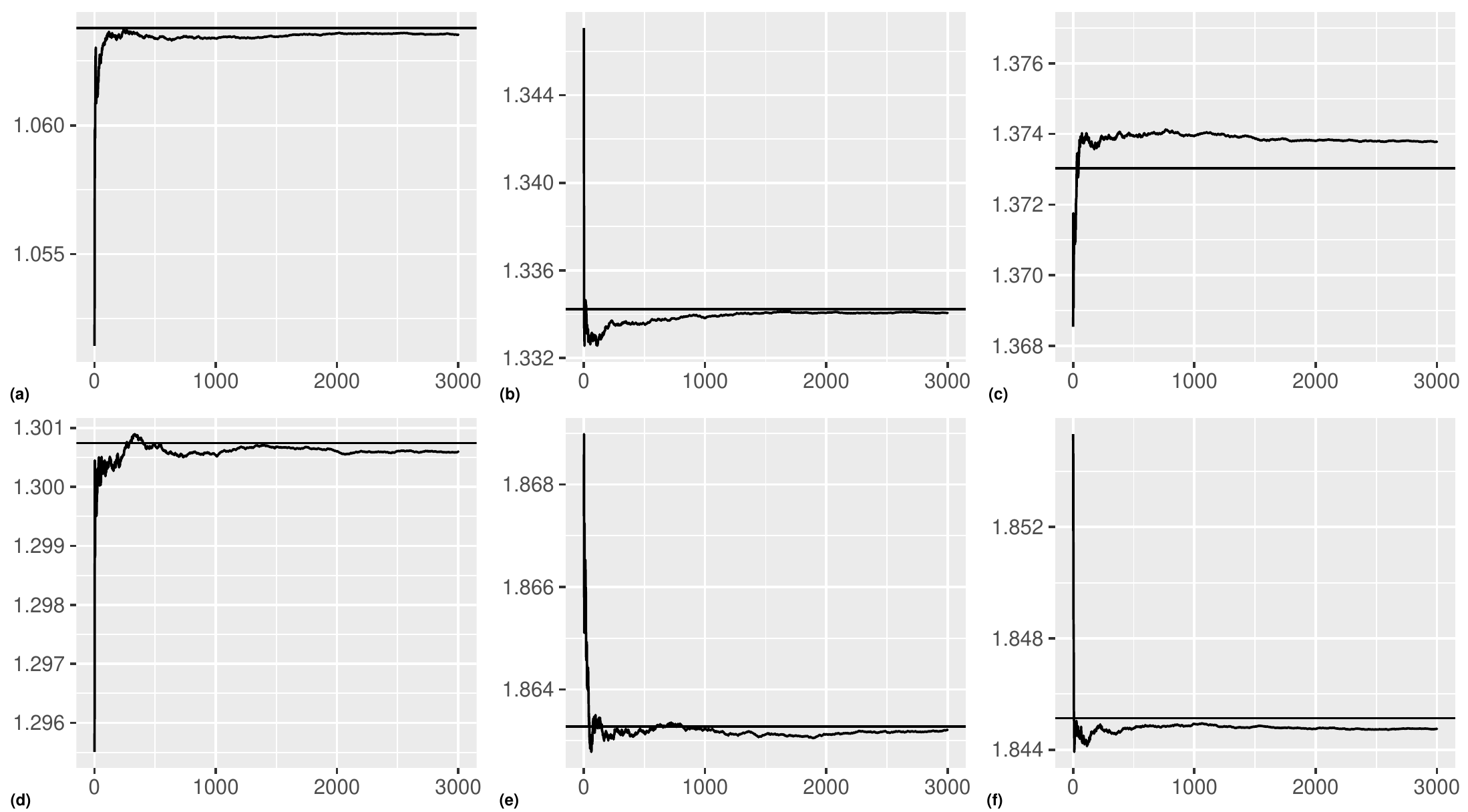}
    \caption{ Cumulative average plots for a few randomly selected entries in $B$ when $(n,p,q)=(150,300,300)$ over the total 3000 Gibbs sampling iterations of the JRNS algorithm. The coordinates selected are (a) (162,37), (b) (14,295), (c) (231,151), (d) (299,102), (e) (162,277), (f) (98,102). The black bold line represents the corresponding true value in $B_0$. These plots illustrate that the MCMC cumulative averages are converging to the respective posterior means, and these posterior means are very close to the corresponding true values in the data generating model. }
    \label{fig:avgplot}
\end{figure*}

{\bf Computational performance}: We start by evaluating the computational performance
of each of the methods. For each of the six combinations from Table \ref{npq:combination}, we
run each of these methods for the 200 replicated data sets and report the average wall clock 
time. All methods were run on HiPerGator, a 
high-performance computing cluster at the University of Florida with each node 
running an Intel Haswell E5-2698 processor. The average wall clock-times are reported
in Table \ref{wall:clock:time}. The results demonstrate the challenges with 
scalability for BANS due to the matrix inversion issues discussed in 
Section \ref{comparison}. For the four combinations of $(n,p,q)$ with at least $200$ 
predictors and responses, the required 3000 iterations of the MCMC algorithm for a single replication could 
not be completed in 4 days (mostly, less than 100 iterations were completed). 
\textcolor{black}{The HS-GHS method also encountered similar issues. For 
Settings 3, 4, 5 and 6 (where $p$ and $q$ are both at least 200), the HS-GHS 
could not complete the required number of iterations in 4 days. In fact, for 
all these settings less than 150 iterations were completed in 4 days. For all 200 replications in Setting 1
we get an error involving positive definiteness of an intermediate matrix calculation. Hence, results are 
only provided for Setting 2.} On the 
other hand, we see that both JRNS and Stepwise approaches scale well and can 
easily handle settings with large $p$ and $q$ values. As expected, the stepwise 
approach takes less computing time than the joint approach. While the DPE 
algorithm also has scalability issues with increasing $p$ and $q$, the DCPE algorithm
scales very well and is the fastest among all the five algorithms in most settings. 
In the $p=q=300$ setting, the stepwise algorithm is faster, and the Joint (JRNS) 
algorithm also roughly takes the same time as DCPE. However, as noted in 
\cite{deshpande2019simultaneous}, the faster speed of DCPE can come at the cost of 
sub-optimal performance (see also Table \ref{mccomegavalues} below). More 
importantly, the DCPE algorithm focuses on optimization of the posterior mode, and does not \textit{provide samples from the posterior distribution} for uncertainty quantification. On the other hand, output from the Joint (JRNS) and Stepwise methods 
can be used to construct posterior marginal inclusion probabilities and credible intervals. This is demonstrated below in 
Tables \ref{inclusion:probability:JRNS}, \ref{inclusion:probability:both} and 
\ref{tab:cov2} for the JRNS method.


\begin{table}[htbp]
    \centering
    \caption{Average wall-clock time (in seconds) over 200 replications for different methods. `TO' is short for `Timeout' which implies that the method could not complete the required number of iterations in 4 days. `PDE' refers to an error caused by intermediate matrices not being positive definite (PD).}
    \resizebox{0.95\textwidth}{!}{%
    \begin{tabular}{*{9}{c}}
        \toprule
          Density & \multicolumn{2}{c}{Cases}& Joint &  Stepwise & DPE & DCPE &  BANS & HS-GHS\\
        \cmidrule(lr){2-3} 
       & $n$ & $(p,q)$ & & &  &  &  \\ 
         \midrule
         \multirow{4}{*}{$(p/5,q/5)$}& 100 & $(30,60)$ & 6.86 & 6.72 & 1.46 & 0.13  & 3890.11 & PDE
        \\
         &100 & $(60,30)$ & 3.79 & 4.48 & 0.67 & 0.05 & 5922.96 & 5811.68 \\
         &150 & $(200,200)$ & 241.25 & 159.19 & 67295.27  & 62.41 & TO & TO \\
           &150 & $(300,300)$ & 833.80 & 280.29 & TO & 785.83 & TO & TO \\
           \midrule
           \multirow{ 2}{*}{$(p/30,q/5)$} & 100 & $(200,200)$& 233.93 & 97.57 & 126175.95  & 29.40  &  TO & TO \\
           &200 & $(200,200)$ & 294.20 & 138.50 & 4956.46 & 78.47 & TO & TO \\
        \bottomrule
    \end{tabular}
     }
     \label{wall:clock:time}
\end{table}

\noindent
{\bf Sparsity selection performance}: To assess the sparsity selection performance of
the methods developed, the following measures were evaluated after running each method on each of the 200 replicates, and comparing the estimated sparsity patterns with the true sparsity pattern: 
$$ \text{Sensitivity} = \frac{\text{TP}}{\text{TP} + \text{FN}}, \quad \text{Specificity} = \frac{\text{TN}}{\text{TN} + \text{FP}} $$ 
\begin{align*}
   &\text{\textbf{M}atthews \textbf{C}orrelation \textbf{C}oefficient (MCC)} \\&= \frac{\text{TP}\times\text{TN} - \text{FP}\times\text{FN}}{\sqrt{(\text{TP+FP})(\text{TP+FN})(\text{TN+FP})(\text{TN+FN})}} 
\end{align*}

\noindent
where TP, TN, FP, FN are the total number of true positive, true negative, false 
positive and false negative identifications made. The average MCC values for sparsity
estimation in $\Omega$ and $B$ for all methods across all combinations are provided 
in Tables \ref{mccomegavalues} and \ref{mccBvalues}, respectively. 

It can be seen that for sparsity selection in $\Omega$, the JRNS and Stepwise 
approaches significantly outperfrom DPE and DCPE in most settings. On closer 
examination of the outputs, one of the reasons appears to be that in some cases the DPE 
and DCPE estimate $\Omega$ by a diagonal matrix, thus failing to identify all 
non-zero off-diagonal elements. The performance of BANS in the settings where results
are available is quite sub-optimal compared to other approaches. For sparsity 
selection in $B$, the JRNS method gives the best performance.
Here the performance of BANS improves
compared to $\Omega$ sparsity selection, but remains sub-optimal compared to 
competing approaches.  The computationally faster approximations DCPE and 
Stepwise are in general less accurate than DPE and JRNS, respectively. 
We have results for HS-GHS only in Setting 2 (due to timeout 
issues discussed before) and its performance with respect to sparsity 
selection in that setting is comparable to other methods. 

\noindent
\textcolor{black}{{\bf Estimation performance:} To assess the estimation performance of the 
proposed methods, we compute the relative estimation error of the final estimates of $B$ and $\Omega$ which have been constructed using the majority voting approach described in Section \ref{MCMCoutput}. The relative estimation errors for $B$ are presented in Table \ref{tab:errorB} and those for $\Omega$ are presented in Table \ref{tab:erroromega}. As is seen from Table \ref{tab:errorB} the JRNS method performs very well in all the simulation settings, in fact it is the best performing method in terms of relative estimation error of $B$ in most of the settings. The performance of the Stepwise method is also quite competitive here. For estimation of $\Omega$, the transformation $h(\cdot)$ described at the end of 
Section \ref{MCMCoutput} was used to ensure positive definiteness of the iterates and the resulting 
estimate. We have included an additional column in Table \ref{tab:erroromega} for the refitted estimates of $\Omega$ as described in Section \ref{MCMCoutput}. It is evident from Table \ref{tab:erroromega} that the performance of the JRNS method (without \textit{refitting}) is competitive with the other methods. In Setting 2 where we were able to get HS-GHS output in a reasonable time, its performance is slightly better than JRNS, Stepwise, DPE and DCPE. The refitting based $\Omega$ estimates for JRNS exhibit better performance than any other method in most of the settings including Setting 2. Note that for refitting based estimates, the connection to the magnitudes of the entries in the $\Omega$ iterates of the JRNS MCMC output, and hence the corresponding credible intervals, is lost (for uncertainty quantification).}

\begin{table}[htbp]
    \centering
    \caption{MCC values for sparsity selection in $\Omega$ averaged over 200 replicates for different methods. `TO' is short for `Timeout' which implies that the method could not complete the required number of iterations in 4 days. `PDE' refers to an error caused by intermediate matrices not being positive definite (PD).}
    \resizebox{0.95\textwidth}{!}{%
    \begin{tabular}{*{9}{c}}
        \toprule
          Sparsity & \multicolumn{2}{c}{Cases}& Joint &  Stepwise &  DPE & DCPE & BANS & HS-GHS\\
        \cmidrule(lr){2-3} 
       & $n$ & $(p,q)$ & & &  &  & \\ 
         \midrule
         \multirow{4}{*}{$(p/5,q/5)$}& 100 & $(30,60)$ & 0.783 & 0.778 & 0.593 & 0.576  & 0.374 & PDE \\
         &100 & $(60,30)$ & 0.821 & 0.820 & 0.708 & 0.623  & 0.305 & 0.831 \\
         &150 & $(200,200)$ & 0.918 & 0.899 & 0.888 & 0.881 & 
         TO & TO  \\
           &150 & $(300,300)$ &  0.912 & 0.831  & 
         TO & 0.752 & TO & TO \\   \midrule
           \multirow{ 2}{*}{$(p/30,q/5)$} & 100 & $(200,200)$& 0.867 & 0.846  &  0.533 & 0.571  & 
         TO & TO \\
           &200 & $(200,200)$ &  0.969 & 0.968  & 0.959 & 0.964 & TO & TO \\
        \bottomrule
    \end{tabular}
     }%
     \label{mccomegavalues}
\end{table}

\begin{table}[htbp]
    \centering
    \caption{MCC values for sparsity selection in $B$ averaged over 200 replicates for different methods. `TO' is short for `Timeout' which implies that the method could not complete the required number of iterations in 4 days. `PDE' refers to an error caused by intermediate matrices not being positive definite (PD).}
    \resizebox{0.95\textwidth}{!}{%
    \begin{tabular}{*{9}{c}}
        \toprule
          Sparsity & \multicolumn{2}{c}{Cases}& Joint &  Stepwise &  DPE & DCPE & BANS & HS-GHS \\
        \cmidrule(lr){2-3} 
       & $n$ & $(p,q)$ & & &  &  & \\ 
         \midrule
         \multirow{4}{*}{$(p/5,q/5)$}& 100 & $(30,60)$ & 1.000  & 1.000  & 1.000  & 1.000  & 0.613 & PDE \\
         &100 & $(60,30)$ & 1.000 & 1.000 & 1.000 & 1.000  & 0.913 & 0.985\\
         &150 & $(200,200)$ & 1.000 & 0.997  & 1.000 & 1.000 & 
         TO & TO \\
           &150 & $(300,300)$ &  0.998  & 0.770  & TO & 0.938 & TO & TO \\   \midrule
           \multirow{ 2}{*}{$(p/30,q/5)$} & 100 & $(200,200)$ & 0.991 & 0.961 & 0.950 & 0.943 & 
         TO & TO \\
           &200 & $(200,200)$ & 1.000 & 0.956 & 0.997 & 0.924 & TO & TO\\
        \bottomrule
    \end{tabular}
     }%
     \label{mccBvalues}
\end{table}

\begin{table}[H]
    \centering
    \caption{Relative estimation error for $B$ averaged over 200 replicates for different methods. `TO' is short for `Timeout' which implies that the method could not complete the required number of iterations in 4 days. `PDE' refers to an error caused by intermediate matrices not being positive definite (PD). }
    \resizebox{0.95\textwidth}{!}{%
    \begin{tabular}{*{9}{c}}
        \toprule
          Sparsity & \multicolumn{2}{c}{Cases}& Joint &  Stepwise & DPE & DCPE &  BANS & HS-GHS \\
        \cmidrule(lr){2-3} 
       & $n$ & $(p,q)$ & & & &  & &\\ 
         \midrule
         \multirow{4}{*}{$(p/5,q/5)$}& 100 & $(30,60)$ & 0.0167 & 0.0169 & 0.0169 & 0.0169 & 0.9323 & PDE\\
         &100 & $(60,30)$ & 0.0269 & 0.0276 & 0.0275 & 0.0277 & 0.9317 & 0.0308 \\
         &150 & $(200,200)$& 0.0154 & 0.0172 & 0.0152 & 0.0153  & TO & TO \\
           &150 & $(300,300)$ & 0.0038 &  0.0434 & TO & 0.0140 & TO & TO\\
           \midrule
           \multirow{ 2}{*}{$(p/30,q/5)$} & 100 & $(200,200)$& 0.0043 & 0.0141 & 0.0063 & 0.0089  &  TO & TO \\
           &200 & $(200,200)$ &  0.00350 & 0.0116 & 0.0033 & 0.0109 & TO & TO \\
        \bottomrule
    \end{tabular}
     }
     \label{tab:errorB}
\end{table}

\begin{table}[H]
    \centering
    \caption{Relative estimation error for $\Omega$ averaged over 200 replicates for different methods. `TO' is short for `Timeout' which implies that the method could not complete the required number of iterations in 4 days. `PDE' refers to an error caused by intermediate matrices not being positive definite (PD). }
    \resizebox{0.95\textwidth}{!}{%
    \begin{tabular}{*{10}{c}}
        \toprule
          Sparsity & \multicolumn{2}{c}{Cases}& Joint & Joint-Refitted & Stepwise & DPE & DCPE &  BANS & HS-GHS \\
        \cmidrule(lr){2-3} 
       & $n$ & $(p,q)$ & & & &  & &\\ 
         \midrule
         \multirow{4}{*}{$(p/5,q/5)$}& 100 & $(30,60)$ & 0.2444 &  0.1794 & 0.2361 & 0.2300 & 0.2271 & 1.0247 & PDE\\
         &100 & $(60,30)$ & 0.2475 & 0.1874 & 0.2389 & 0.2424 & 0.2538 & 1.0125 & 0.2180 \\
         &150 & $(200,200)$ & 0.2197 & 0.1442 & 0.2092 & 0.1429 & 0.1444 & TO & TO  \\
           &150 & $(300,300)$ & 0.2181 & 0.1424 & 0.2306 & TO  & 0.1679 & TO & TO\\
           \midrule
           \multirow{ 2}{*}{$(p/30,q/5)$} & 100 & $(200,200)$& 0.2352 & 0.1854 & 0.2262 & 0.2423 & 0.2320 & TO  & TO  \\
           &200 & $(200,200)$ &  0.2209 & 0.1159 & 0.2029 & 0.1129 & 0.1100 & TO & TO  \\
        \bottomrule
    \end{tabular}
     }
     \label{tab:erroromega}
\end{table}

\noindent
{\bf Uncertainty quantification based on the generalized posterior distribution}: 
Next, we illustrate uncertainty quantification for JRNS using inclusion 
probabilities (see Section \ref{MCMCoutput}) and credible intervals 
obtained from the generalized posterior distribution. Note that the DPE 
and DCPE algorithms do not provide posterior samples for this purpose. 
We first consider the simulation setting where $(n,p,q) = $ $
(100,200,200)$, and randomly choose one out of the $200$ replicated data 
sets. Table \ref{inclusion:probability:JRNS} shows the estimated 
marginal inclusion probabilities for selected entries in $B$ and 
$\Omega$ using the JRNS algorithm. For the matrix, $B$, entries 
$(47, 4)$, $(30, 14)$, $(181, 43)$ are true positives: they are 
estimated as non-zero, since all have estimated inclusion probability $1$
(the corresponding values were chosen as non-zero for all $2000$ post 
burn-in iterations), and their true values in $B_0$ are non-zero. Entry $(78, 84)$ is a false positive: it is estimated as non-zero since 
the estimated inclusion probability is $0.632 > 0.5$ (the corresponding 
values were chosen as non-zero for $1262$ out of $2000$ post burn-in 
iterations), but its true value in $B_0$ is zero. Hence, the inclusion
probabilities indicate that the decision to classify $(78, 84)$ as 
non-zero is not supported with the same certainty by the posterior distribution as the 
decision to classify $(47, 4)$, $(30, 14)$, $(181, 43)$. Finally, entries $(67,5)$, $(12, 72)$ are true negatives: 
they are estimated as zero since the inclusion probabilities $0.005$ and $0.0915$
are less than $0.5$, and their true values in $B_0$ are zero. For the $\Omega$ matrix, 
entries $(109, 136), (30, 32),$ $ (9, 200)$ are true positives with estimated marginal 
inclusion probabilities $1$, entry $(101,122)$ is false positive with
estimated marginal inclusion probability $0.568$, and entries $(180, 2), (103, 
13)$ are true negatives both with estimated marginal inclusion probabilities $0$. The network plots indicating the associations between the predictors and the response variables and also among the response variables for this replication are presented in Figure \ref{fig:Sim.igraph}.
Note that the BANS algorithm can not provide a full set of iterations for the 
$(n,p,q) = (100,200,200)$ setting due to computational scalability issues. Hence, 
inclusion probabilities for the BANS algorithm are not included in 
Table \ref{inclusion:probability:JRNS}. 
\begin{figure*}
    \centering
    \includegraphics[width = 1\linewidth, trim={1cm 0 2cm 0}]{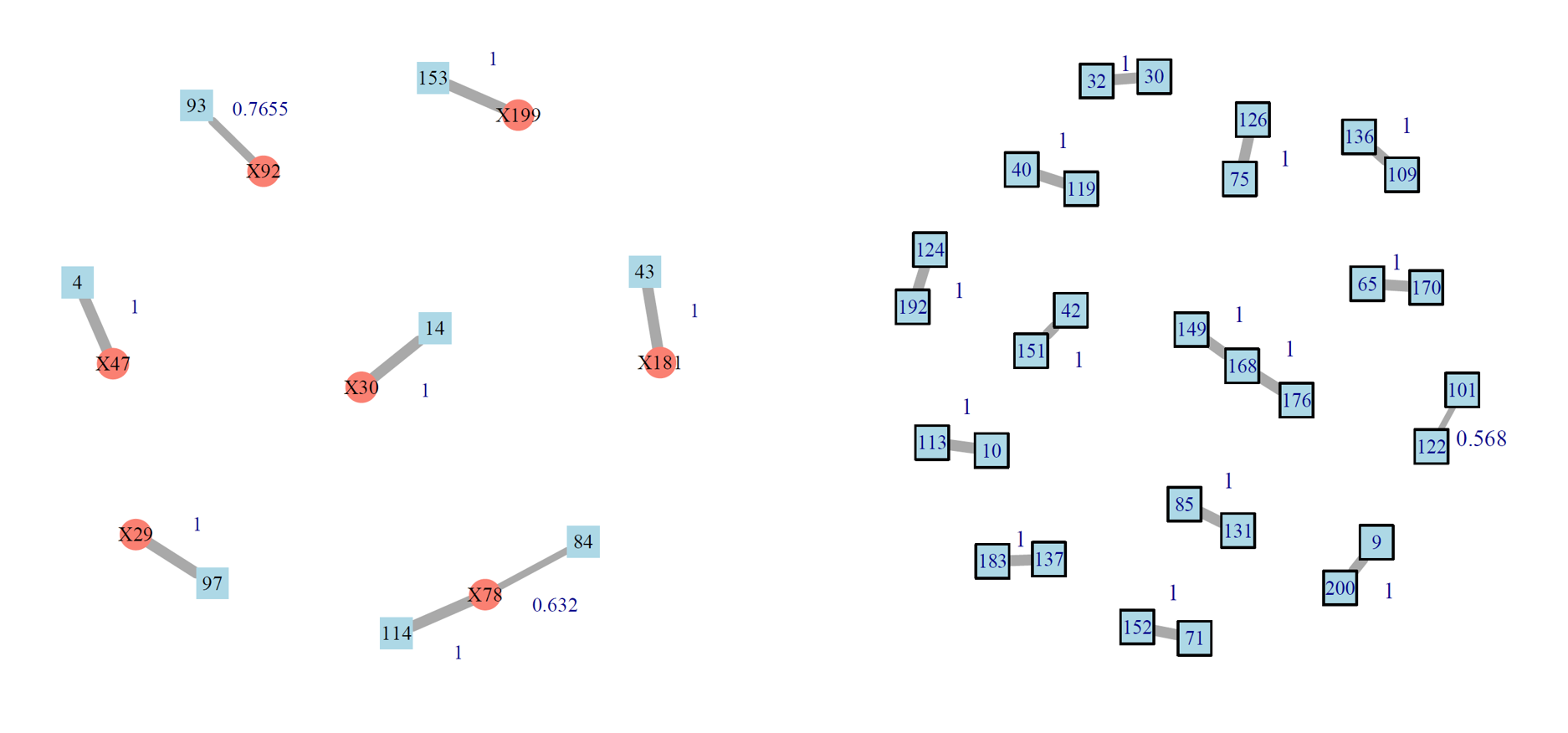}
    \caption{ Left : Network graph indicating associations between the predictors and the response variables for a randomly chosen replication when $(n,p,q) = (100,200,200)$ . Right : Network graph indicating associations among the response variables for the same replication. The red circles represent the different predictor variables and the blue squares represent the different response variables. The inclusion probabilities are mentioned on each edge. \textit{All the edge widths are proportional to the corresponding inclusion probabilities.}}
    \label{fig:Sim.igraph}
\end{figure*}

\begin{table}[htbp]
\caption{Illustration of classification based on marginal posterior inclusion probabilities using 
the joint (JRNS) method for selected entries of $B$ and $\Omega$ for the randomly 
chosen replication for $(n,p,q) = (100,200,200).$}
\resizebox{1.0\textwidth}{!}{%
\begin{tabular}{ccccc}
\hline
Matrix & Entry & JRNS classification & Inclusion probability & True classification\\
\hline
$B$ & $(47, 4)$ & Non-zero & $1$ & Non-zero\\
$B$ & $(30, 14)$ & Non-zero & $1$ & Non-zero\\
$B$ & $(181, 43)$ & Non-zero & $1$ & Non-zero\\
$B$ & $(78, 84)$ & Non-zero & $0.632$ & Zero\\
$B$ & $(67, 5)$ & Zero & $0.005$ & Zero\\
$B$ & $(12, 72)$ & Zero & $0.0915$ & Zero\\
$\Omega$ & $(109, 136)$ & Non-zero & $1$ & Non-zero\\
$\Omega$ & $(30, 32)$ & Non-zero & $1$ & Non-zero\\
$\Omega$ & $(9, 200)$ & Non-zero & $1$ & Non-zero\\
$\Omega$ & $(101, 122)$ & Non-zero & $0.568$ & Zero\\
$\Omega$ & $(180, 2)$ & Zero & $0$ & Zero\\
$\Omega$ & $(103, 13)$ & Zero & $0$ & Zero\\
\hline
\end{tabular}
}%
\label{inclusion:probability:JRNS}
\end{table}
For a comparative illustration with both JRNS and BANS, we consider the $(n,p,q) = 
(100,30,60)$ setting. Marginal inclusion probability estimates for selected entries 
of $B$ and $\Omega$ for both the joint (JRNS) method and the BANS approach (based on 
$2000$ post burn-in iterations) are provided in 
Table \ref{inclusion:probability:both}. Entries $(30, 5), (25, 6)$ in 
$B$, and $(9,47)$ in $\Omega$ are true positives for both methods 
(correctly identified as non-zero), but the inclusion probabilities for BANS are smaller
than those of JRNS for all three entries.
Entries $(21,48)$ in $B$ and $(8,31)$ 
in $\Omega$ are falsely identified as non-zero by BANS based on inclusion 
probabilities greater than $0.5$, but correctly identified as zero by JRNS. Entry $(10,40)$ in $\Omega$ is correctly identified as non-zero by JRNS with an inclusion probability of $1$ while BANS incorrectly identifies it as zero with a low inclusion probability. Other 
entries in the table are true negatives for both methods (correctly identified as 
zero), but JRNS has a lower inclusion probability for all as compared to BANS. While 
the entries reported in Table \ref{inclusion:probability:both} are just a small 
subset, we found that the pattern of JRNS having a higher/lower inclusion probability
than BANS when the true value is non-zero/zero is repeated for most entries of 
$B$ and $\Omega$. This is not surprising given the significantly better selection 
performance of JRNS in this setting (see Tables \ref{mccomegavalues} and \ref{mccBvalues}).

\begin{table}[htbp]
\caption{Illustration of classification based on marginal posterior inclusion probabilities using 
the joint (JRNS) method and the BANS algorithm in \cite{ha2020bayesian} for selected 
entries of $B$ and $\Omega$ for a randomly chosen replication for $(n,p,q) = 
(100,30,60)$.}
\resizebox{1.0\textwidth}{!}{%
\begin{tabular}{ccccccc}
\toprule
Matrix & Entry &\multicolumn{2}{c}{Classification}& \multicolumn{2}{c}{Inclusion probability} & True classification\\
 \cmidrule(lr){3-4} \cmidrule(lr){5-6}
  & & JRNS & BANS & JRNS & BANS & \\
\hline
$B$ & $(30, 5)$ & Non-zero & Non-zero & $1$ & 0.972 & Non-zero\\
$B$ & $(25, 6)$ & Non-zero & Non-zero & $1$ & 0.765 & Non-zero\\
$B$ & $(21,48)$ & zero & Non-zero & $0.04$ & 0.772 & zero\\
$B$ & $(6,24)$ & zero & zero & $0.0625$ & 0.253 & zero\\
$B$ & $(24,57)$ & zero & zero & $0.0295$ & 0.2505 & zero\\
$\Omega$ & $(30, 9)$ & zero & zero & $0$ & 0.0495 & zero\\
$\Omega$ & $(53,60)$ & zero & zero & $0$ & 0.293 & zero\\
$\Omega$ & $(10,40)$ & Non-zero & zero & $1$ & 0.112 & Non-zero\\
$\Omega$ & $(9,47)$  & Non-zero & Non-zero &  1 & 0.952 & Non-zero\\
$\Omega$ & $(8,31)$ & zero & Non-zero & 0 & 0.5605 & zero\\
$\Omega$ & $(16,6)$ &  zero & zero & $0$ & 0.492 & zero\\
$\Omega$ & $(21,59)$ &  zero & zero & $0$ & 0.4075 & zero\\
\hline
\end{tabular}
}%
\label{inclusion:probability:both}
\end{table}

\noindent

Next, we consider the second simulation setting, where $(n,p,q) = (100, 60, 
30)$, for a comparison of the empirical coverage probabilities of the $95\%$ 
posterior credible intervals by JRNS and HS-GHS. For each of 12 true non-zero 
entries of $B$, and each of the $200$ replications, we compute the $95\%$ 
posterior credible interval obtained by using the relevant sample quantiles of
the non-zero values in the $2000$ post burn-in iterations (for both the 
methods). The proportion of credible intervals (out of $200$) which contain
the true value gives us an estimate of the coverage probability for each 
method. Table \ref{tab:cov2} presents the average coverage over the 200 replicated datasets of true value in the 95$\%$ credible intervals for these 12 entries of $B_0$. For $B$ entries, both methods perform very well with respect to including the true value in their corresponding credible intervals. The average coverage probability for JRNS is 0.948, while that for HS-GHS is 0.945. We also provide Figure 4 for a visual comparison of the posterior credible intervals for the non-zero entries of the true $B$ in one of the replicates. The plot shows the credible intervals by both methods for all 12 non-zero values in the true $B$ for a single data set. In this particular data set, most of the credible intervals by JRNS are in general narrower than the corresponding credible intervals by HS-GHS, though the difference is relatively small. However, for co-ordinate (56,8) the HS-GHS credible interval fails to capture the true value, while the JRNS credible interval contains the true one. Similar patterns were observed in the credible intervals for other replicates.

\noindent
We also obtain credible intervals for the three true non-zero entries of $\Omega$ in this setting and computed the coverage probability in a similar process as mentioned above for $B$. For a randomly selected replicate, we plot the credible intervals of the true non-zero entries of $\Omega$ by JRNS and HS-GHS in Figure \ref{fig:compCredibleomega} and the average coverage probabilities are listed in Table \ref{tab:cov2}. The credible intervals by JRNS are narrower, however, the comparison of coverage performance is mixed. Due to the narrower credible intervals of JRNS, the true value can sometimes lie just outside the credible interval and hence the coverage probability gets negatively impacted by this. We also obtain credible intervals for the three true non-zero entries of $\Omega$ in this setting and computed the coverage probability in a similar process as mentioned above for $B$. For a randomly selected replicate we plot the credible intervals of true non-zero entries of $\Omega$ by JRNS and HS-GHS in Figure \ref{fig:compCredibleomega} and the average coverage probabilities are listed in Table \ref{tab:cov2}. The credible intervals by JRNS are narrower, however, the comparison of coverage performance is mixed. Due to the narrower credible intervals of JRNS the true value can sometimes lie just outside the credible interval and hence the coverage probability gets negatively impacted by this.

We also consider a simulation setting with $(n,p,q)$ $ = (150, 300, 300)$, and select a group of 
entries in $B_0$ which are non-zero. The coverage probabilities for the $95\%$ credible intervals,  as described before, are estimated by the proportion of credible intervals (out  of  200) containing  the  true  value. The average coverage probability over all true non-zero entries in $B$ and over all $200$ replications is 0.9422. 
Recall that the values for BANS and HS-GHS in this setting are not available 
due to computational scalability issues. Next, we present a comparison between the credible intervals obtained from JRNS and the frequentist confidence intervals obtained from the debiased lasso approach \cite{van2014asymptotically} in Figure \ref{fig:CI4} for 7 randomly selected coordinates of $B$. The plot indicates that for all of these coordinates the JRNS approach provides narrower and more precise intervals while containing the corresponding true values for most of these coordinates. The codes implementing the two proposed methods, namely the JRNS and the Stepwise methods are available at \url{https://github.com/srijata06/JRNS_Stepwise}.
\begin{table}[ht]
\centering
\caption{A comparison of the average Coverage of true value in $95\%$ posterior credible intervals for the non-zero values in $B_0$ and $\Omega_0$ for $(n,p,q) = (100,60,30)$.}
\resizebox{0.45\textwidth}{!}{%
\begin{tabular}{rrrr}
  \hline
 & coordinates & JRNS & HS-GHS \\ 
  \hline
$B$ & (30,3) & 0.940 & 0.945 \\ 
$B$ &  (55,3) & 0.945 & 0.930 \\ 
$B$ &  (20,5) & 0.955 & 0.925 \\ 
$B$ &  (29,7) & 0.970 & 0.970 \\ 
$B$ &  (45,8) & 0.930 & 0.950 \\ 
$B$ &  (56,8) & 0.940 & 0.950 \\ 
$B$ &  (17,14) & 0.945 & 0.935 \\ 
$B$ &  (55,20) & 0.940 & 0.945 \\ 
$B$ &  (9,22) & 0.945 & 0.925 \\ 
$B$ &  (53,23) & 0.970 & 0.980 \\ 
$B$ &  (60,27) & 0.955 & 0.945 \\ 
$B$ & (57,30) & 0.940 & 0.935 \\ 
   \hline
$\Omega$ & (2,8) & 0.607 & 0.865\\
$\Omega$ & (3,19) & 0.938 & 0.800\\
$\Omega$ & (18,26) & 0.778 &  0.810\\
\hline
\end{tabular}
}%
\label{tab:cov2}
\end{table}

\begin{figure}[htbp]
    \centering
    \includegraphics[width = 0.95\linewidth]{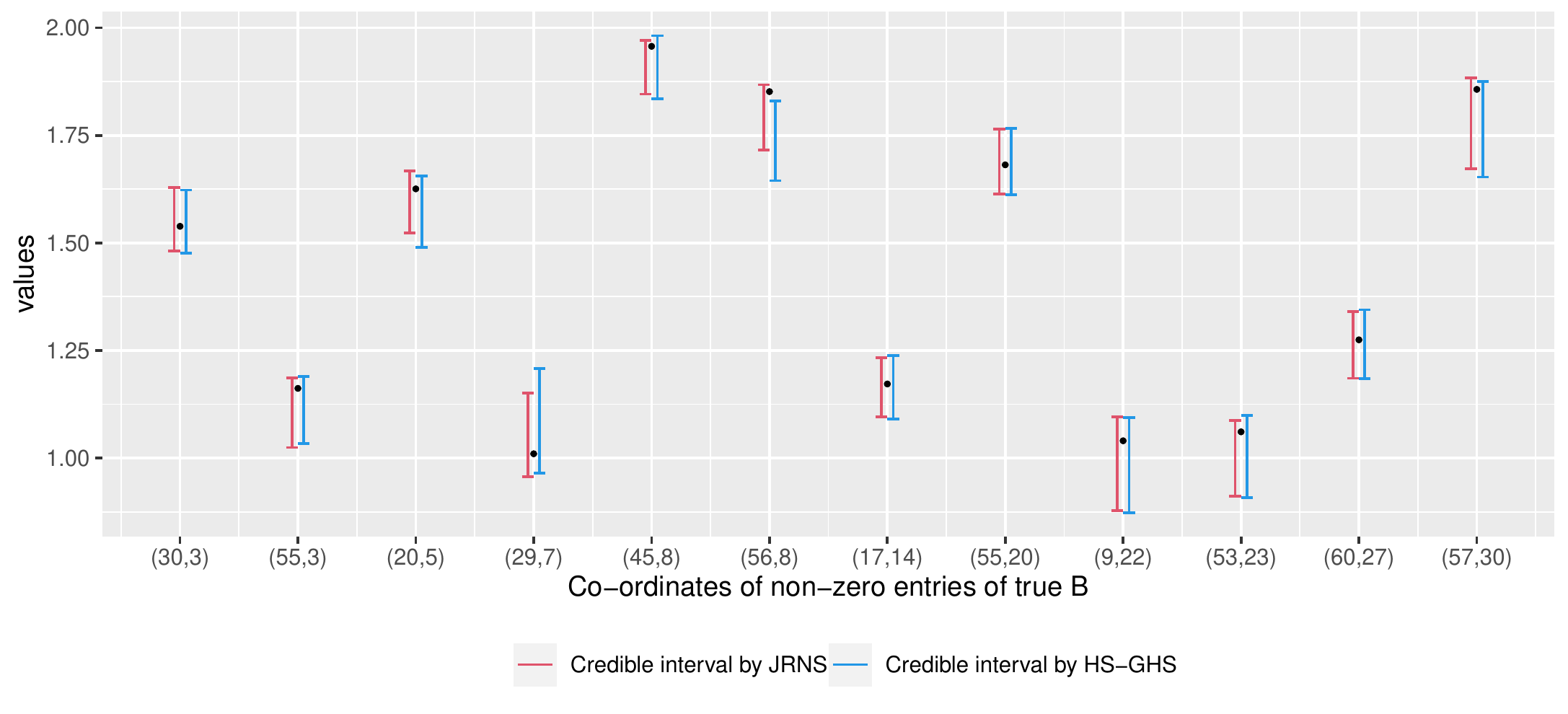}
    \caption{A comparison of the coverage of credible interval by JRNS and by HS-GHS for non-zero entries of $B_0$ when $(n,p,q) = (100,60,30)$. The true values are represented by the black circles.}
    \label{fig:compCredible}
\end{figure}

\begin{figure}[htbp]
    \centering
    \includegraphics[width = 0.99\linewidth]{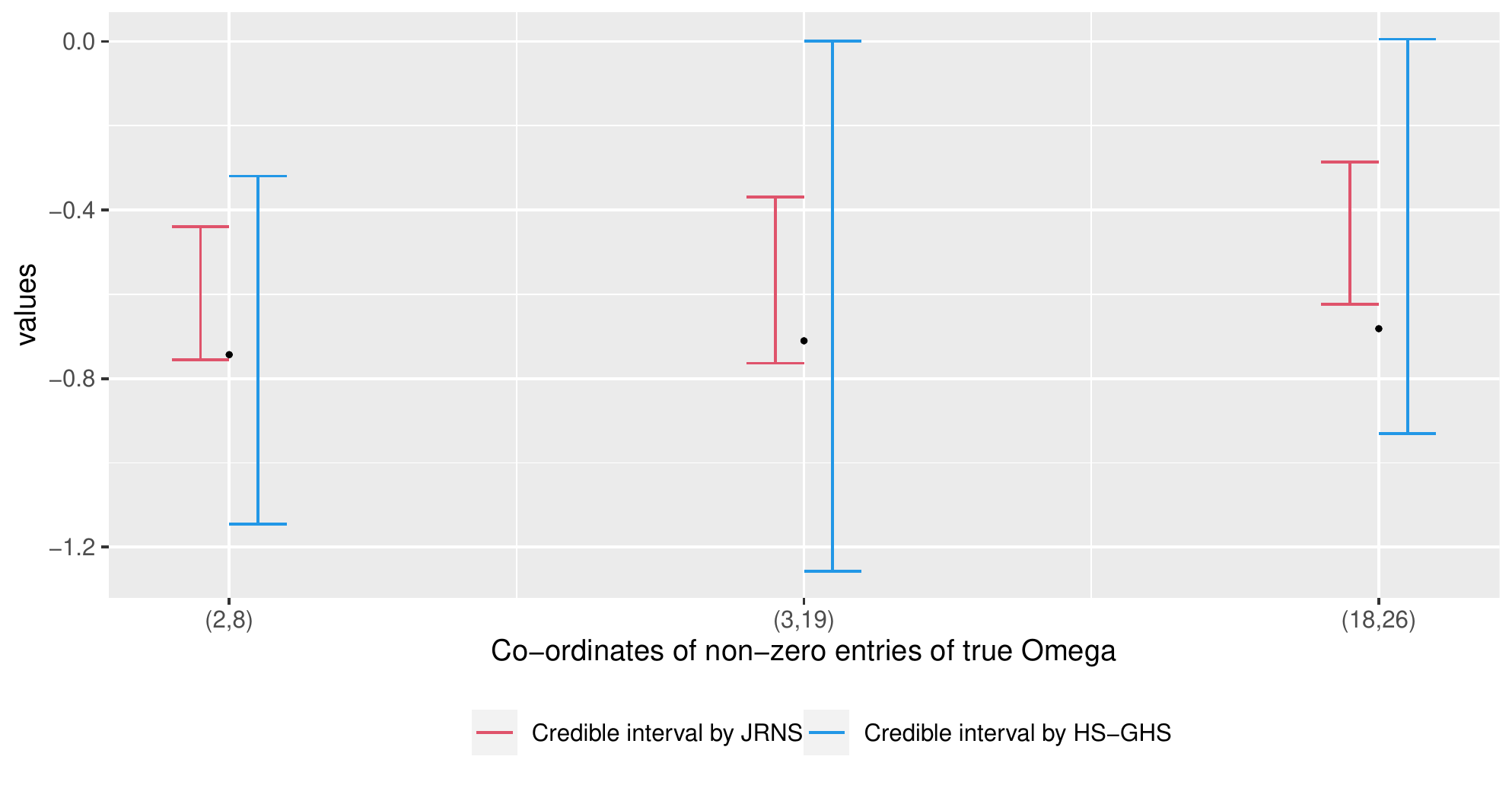}
    \caption{A comparison of the coverage of credible interval by JRNS and by HS-GHS for non-zero entries of $\Omega_0$ when $(n,p,q) = (100,60,30)$. The true values are represented by the black circles.}
    \label{fig:compCredibleomega}
\end{figure}

\begin{figure}[htbp]
    \centering
    \includegraphics[width = 0.95\textwidth]{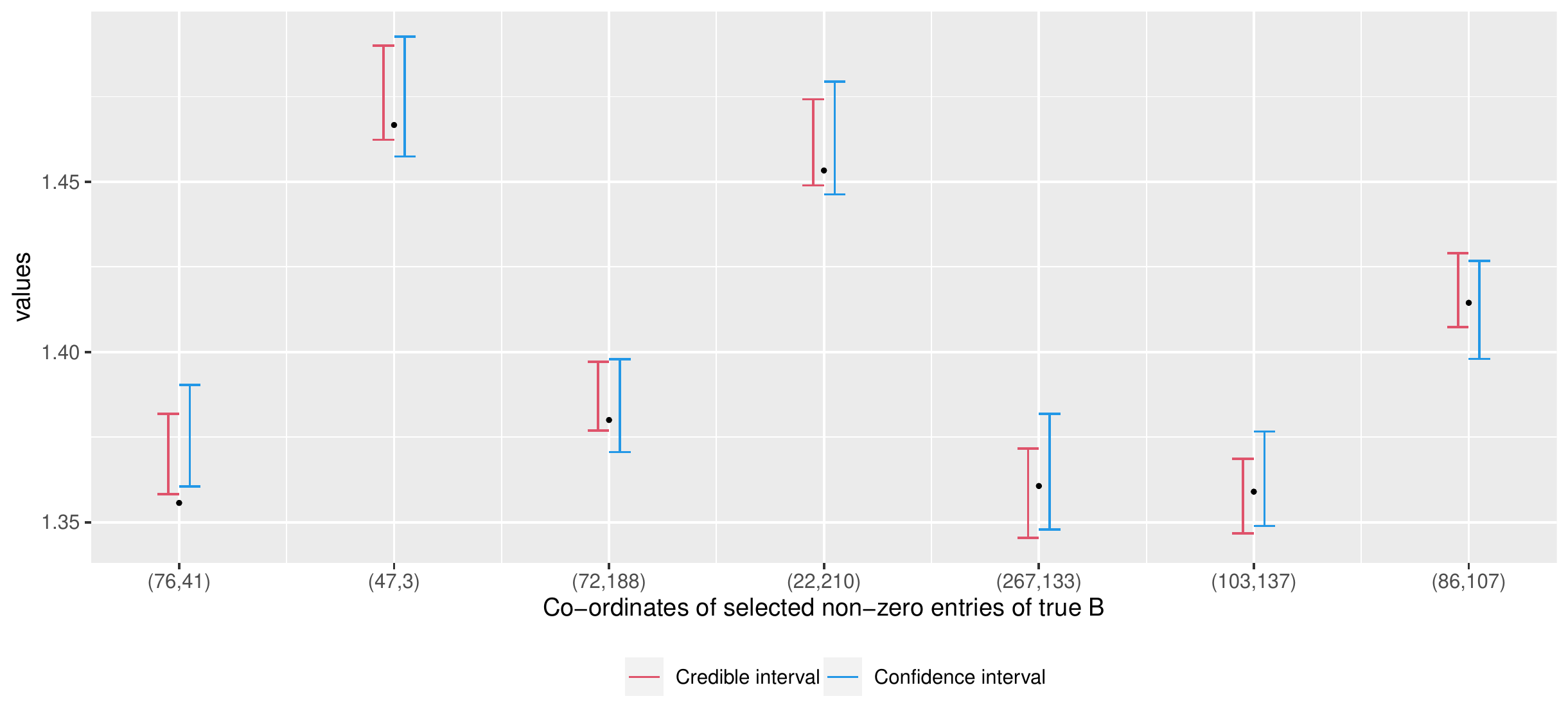}
    \caption{A comparison of the coverage of credible interval by JRNS and Confidence interval by Debiased Lasso for selected coordinates of $B_0$ when $(n,p,q) = (150,300,300)$. The true values are represented by the black circles. The plot shows that JRNS provides narrower intervals than the corresponding frequentist confidence intervals.}
    \label{fig:CI4}
\end{figure}

\section{Analysis of TCGA cancer data} \label{data:analysis}
\noindent

\begin{table}[htbp]
\centering
\caption{ $(n,p,q)$ values for the datasets on seven different cancer types. }
\begin{tabular}{|c|c|c|c|c|}
\hline
 & Cancer type  & $n$ & $p$ & $q$\\
\hline
1 & READ & 121 & 73 & 76\\
2 & LUAD & 356 & 73 & 76\\
3 & COAD & 338 & 73  & 76\\
4 & LUSC & 309 & 73 & 86\\
5 & OV & 227 & 73 & 77\\
6 & SKCM & 333 & 73 & 76\\
7 & UCEC & 393 & 73 & 77\\
\hline
\end{tabular}
\label{npq:ov}
\end{table}

To further illustrate the performance of the proposed methods, we present results from the analysis of cancer data from TCGA (The Cancer Genome Atlas). We consider data for 7 different TCGA tumor types: colon adenocarcinoma
(COAD), lung adenocarcinoma (LUAD),
lung squamous cell carcinoma (LUSC),  ovarian serous cystadenocarcinoma (OV), rectum adenocarcinoma (READ) skin cutaneous melanoma (SKCM) and uterine corpus endometrial carcinoma (UCEC). For each of these cancer types we have mRNA expression data and RPPA-based proteomic data. As mentioned in the introduction, since mRNA is translated to protein, it is natural to consider protein expression data to be the response variable and the mRNA expression data to be the predictors. The sample size $(n)$, number of predictors $(p)$ and the number of response variables $(q)$ for the 7 data sets corresponding to each cancer type are  given in Table \ref{npq:ov}.

\begin{table*}
\centering
    \caption{Inclusion probability of each edge for the LUAD Network graph indicating associations between mRNA and protein presented in Figure \ref{fig:LUADnetwork05}.}
    \resizebox{0.9\textwidth}{!}{%
    \begin{minipage}{.4\linewidth}
      
      \centering
        \begin{tabular}{cccc}
            \hline
 & Gene  & Protein  & Inclusion Probability  \\ 
  \hline

 1 & X1 &   1 & 0.98 \\ 
  2 & X2 &   2 & 1.00 \\ 
  3 & X64 &   2 & 0.88 \\ 
  4 & X60 &   3 & 0.78 \\ 
  5 & X4 &   4 & 1.00 \\ 
  6 & X32 &   4 & 0.53 \\ 
  7 & X5 &   5 & 1.00 \\ 
  8 & X16 &   5 & 0.56 \\ 
  9 & X6 &   6 & 1.00 \\ 
  10 & X7 &   7 & 1.00 \\ 
  11 & X20 &   7 & 0.98 \\ 
  12 & X8 &   8 & 1.00 \\ 
  13 & X56 &   9 & 0.96 \\ 
  14 & X60 &   9 & 1.00 \\ 
  15 & X11 &  10 & 1.00 \\ 
  16 & X11 &  11 & 1.00 \\ 
  17 & X20 &  11 & 0.76 \\ 
  18 & X12 &  12 & 1.00 \\ 
  19 & X13 &  13 & 1.00 \\ 
  20 & X14 &  13 & 0.91 \\ 
  21 & X16 &  13 & 0.63 \\ 
  22 & X70 &  14 & 0.79 \\ 
  23 & X15 &  16 & 1.00 \\ 
  24 & X40 &  16 & 0.72 \\ 
  25 & X10 &  17 & 0.94 \\ 
  26 & X16 &  17 & 1.00 \\ 
  27 & X17 &  18 & 1.00 \\
   \hline
        \end{tabular}
    \end{minipage}%
    \quad \quad
    \begin{minipage}{.4\linewidth}
      \centering
         \begin{tabular}{cccc}
            \hline
 & Gene  & Protein  & Inclusion Probability   \\ 
  \hline
 28 & X68 &  18 & 0.52 \\ 
  29 & X18 &  19 & 1.00 \\ 
  30 & X11 &  20 & 0.69 \\ 
  31 & X17 &  21 & 0.72 \\ 
  32 & X21 &  21 & 1.00 \\ 
  33 & X6 &  22 & 0.86 \\ 
  34 & X21 &  24 & 0.75 \\ 
  35 & X24 &  24 & 1.00 \\ 
  36 & X23 &  25 & 1.00 \\ 
  37 & X28 &  25 & 0.81 \\ 
  38 & X63 &  25 & 0.66 \\ 
  39 & X70 &  25 & 0.97 \\ 
  40 & X10 &  26 & 0.54 \\ 
  41 & X8 &  27 & 0.60 \\ 
  42 & X11 &  27 & 0.89 \\ 
  43 & X27 &  27 & 1.00 \\ 
  44 & X28 &  28 & 1.00 \\ 
  45 & X58 &  28 & 0.98 \\ 
  46 & X29 &  29 & 1.00 \\ 
  47 & X64 &  30 & 0.89 \\ 
  48 & X17 &  31 & 1.00 \\ 
  49 & X31 &  31 & 1.00 \\ 
  50 & X32 &  31 & 0.81 \\ 
  51 & X42 &  32 & 0.99 \\ 
  52 & X34 &  34 & 1.00 \\ 
  53 & X70 &  34 & 0.83 \\ 
  54 & X33 &  35 & 0.99 \\
  \hline
  \end{tabular}
    \end{minipage}%
    \quad \quad
    \begin{minipage}{.4\linewidth}
      \centering
         \begin{tabular}{cccc}
            \hline
 & Gene  & Protein  & Inclusion Probability   \\ 
  \hline
  55 & X35 &  35 & 1.00 \\ 
  56 & X36 &  37 & 0.51 \\ 
  57 & X31 &  38 & 0.58 \\ 
  58 & X37 &  38 & 1.00 \\ 
  59 & X38 &  39 & 1.00 \\ 
  60 & X25 &  40 & 0.65 \\ 
  61 & X39 &  40 & 1.00 \\ 
  62 & X8 &  41 & 0.53 \\ 
  63 & X47 &  41 & 1.00 \\ 
  64 & X24 &  43 & 0.53 \\ 
  65 & X43 &  43 & 1.00 \\ 
  66 & X44 &  44 & 0.92 \\ 
  67 & X46 &  44 & 0.94 \\ 
  68 & X11 &  47 & 0.96 \\ 
  69 & X47 &  47 & 1.00 \\ 
  70 & X36 &  49 & 0.91 \\ 
  71 & X50 &  50 & 0.71 \\ 
  72 & X6 &  52 & 0.81 \\ 
  73 & X61 &  53 & 1.00 \\ 
  74 & X68 &  53 & 0.88 \\ 
  75 & X72 &  56 & 0.96 \\ 
  76 & X58 &  57 & 1.00 \\ 
  77 & X60 &  57 & 1.00 \\ 
  78 & X15 &  58 & 0.81 \\ 
  79 & X58 &  58 & 1.00 \\ 
  80 & X28 &  59 & 0.78 \\ 
  81 & X58 &  59 & 1.00 \\ 
  \hline
  \end{tabular}
    \end{minipage}%
    \quad \quad
    \begin{minipage}{.4\linewidth}
      \centering
         \begin{tabular}{cccc}
            \hline
 & Gene  & Protein  & Inclusion Probability  \\ 
  \hline
  82 & X59 &  59 & 0.92 \\ 
  83 & X61 &  61 & 0.51 \\ 
  84 & X6 &  62 & 0.64 \\ 
  85 & X29 &  63 & 0.98 \\ 
  86 & X62 &  63 & 1.00 \\ 
  87 & X70 &  63 & 1.00 \\ 
  88 & X63 &  64 & 1.00 \\ 
  89 & X33 &  65 & 0.90 \\ 
  90 & X63 &  65 & 1.00 \\ 
  91 & X63 &  66 & 1.00 \\ 
  92 & X17 &  68 & 0.83 \\ 
  93 & X46 &  68 & 0.86 \\ 
  94 & X24 &  69 & 0.78 \\ 
  95 & X11 &  71 & 0.52 \\ 
  96 & X57 &  71 & 0.64 \\ 
  97 & X67 &  71 & 0.94 \\ 
  98 & X59 &  72 & 0.51 \\ 
  99 & X69 &  73 & 1.00 \\ 
  100 & X3 &  74 & 0.53 \\ 
  101 & X70 &  74 & 1.00 \\ 
  102 & X71 &  74 & 0.99 \\ 
  103 & X14 &  75 & 0.85 \\ 
  104 & X72 &  75 & 1.00 \\ 
  105 & X73 &  76 & 1.00 \\
  & & &\\
  & & & \\
  & & & \\
  \hline
        \end{tabular}
        \end{minipage} 
    }
    \label{tab:B:incprob}
\end{table*}

\begin{figure*}
    \centering
    \includegraphics[width = \linewidth]{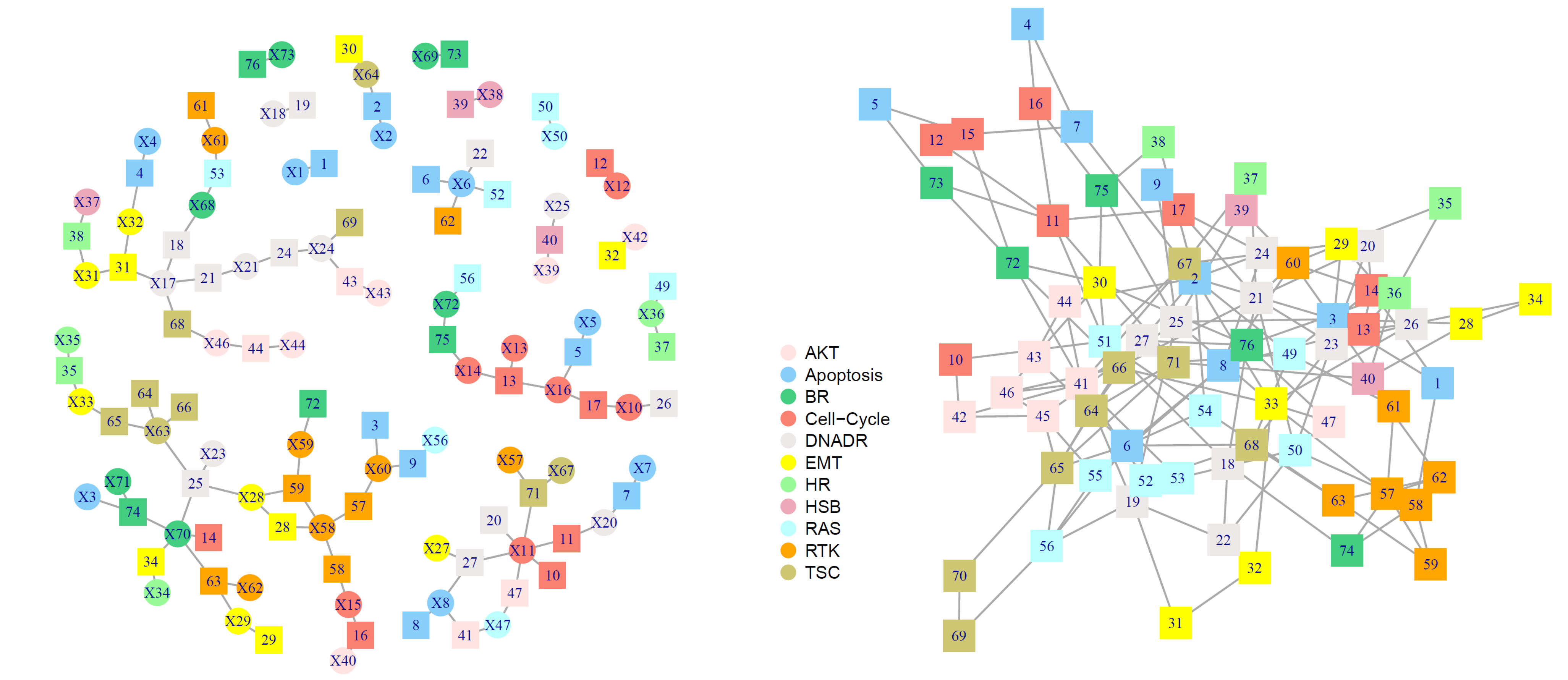}
    \caption{LUAD networks with 0.5 as the inclusion probability cutoff. The circles represent genes and the squares represent proteins. The different colors represent the different pathways listed in Table \ref{tab:pathway}. Left : Network graph indicating associations between mRNA and protein. The inclusion probabilities are listed in Table \ref{tab:B:incprob}. Right : Network graph indicating associations among proteins. The inclusion probabilities are listed in Table \ref{LUAD:Omega:incprob} in the Supplementary document. }
    \label{fig:LUADnetwork05}
\end{figure*}

We carry out a separate data analysis for each of the seven cancer types.  
For JRNS and the Stepwise estimation methods the Gibbs samplers were run 
for 1000 iterations for burn-in followed by additional 2000 iterations for 
calculating the regression coefficients and the precision matrices. As  
noted earlier, DPE and DCPE do not provide uncertainty quantification. While
BANS does provide uncertainty quantification, computationally it takes a 
prohibitively long time with the above $(n,p,q)$ values. In 
\cite{ha2020bayesian}, this dataset was analyzed but the dataset for each 
cancer type was further broken based on pathway information, which 
significantly reduces the dimensionality of the problem.

We present the estimated network plots obtained using JRNS depicting (1) 
the associations between mRNA and proteins and (2) that among the proteins 
for the LUAD cancer type in Figure \ref{fig:LUADnetwork05}. The indices/serial numbers for 
genes and proteins for the LUAD dataset are given in Table \ref{tab:geneIndices} in the Appendix. Figure \ref{fig:LUADnetwork05} depicts the sparsity estimates of $B$ and $\Omega$ for the LUAD type cancer based on a 0.5 cutoff for the inclusion probabilities. The genes and proteins are mapped to their respective functional pathways to aid interpretation. The list of pathways and the corresponding genes for each of the pathways is listed in Table \ref{tab:pathway} in the Appendix. For the associations encoded in matrix $B$ (see left panel of Figure \ref{fig:LUADnetwork05}), we see genes and proteins from the following pathways to be involved : RTK, EMT, Cell Cycle and Apoptosis. The results are broadly consistent with 
known functional mechanisms for the disease including stimulation of RTK to activate downstream signaling that encodes EMT's inducing transcription factors \cite{gonzalez2014signaling}. The epithelial mesenchymal transition (EMT) is an essential mechanism that contributes to the progression in cancer and involves apoptotic responses and the cell cycle, all elements captured in some of the connections depicted in the Figure. Further, we see similar connections at the protein expression network in the right panel of Figure \ref{fig:LUADnetwork05}. One can also see that there are strong connections within members of the same pathway, as well as cross-talk with members of other pathways. We particularly focus on the LUAD network plots here as it shows some very interesting biological connections. The network plots for the other cancer types are included in the Supplementary file. 

Next, we present Figure \ref{fig:LUAD_CI} which depicts the coverage of the credible intervals by JRNS and the confidence intervals by Debiased Lasso for six randomly selected entries of $B$ for the lung adenocarcinoma (LUAD) cancer data. We randomly selected 6 gene-protein coordinates in $B$. Here the credible intervals are not only much shorter than the corresponding confidence intervals, but in most cases are subsets of their corresponding confidence intervals. 

\begin{figure*}
    \centering
    \includegraphics[width = 0.95\textwidth]{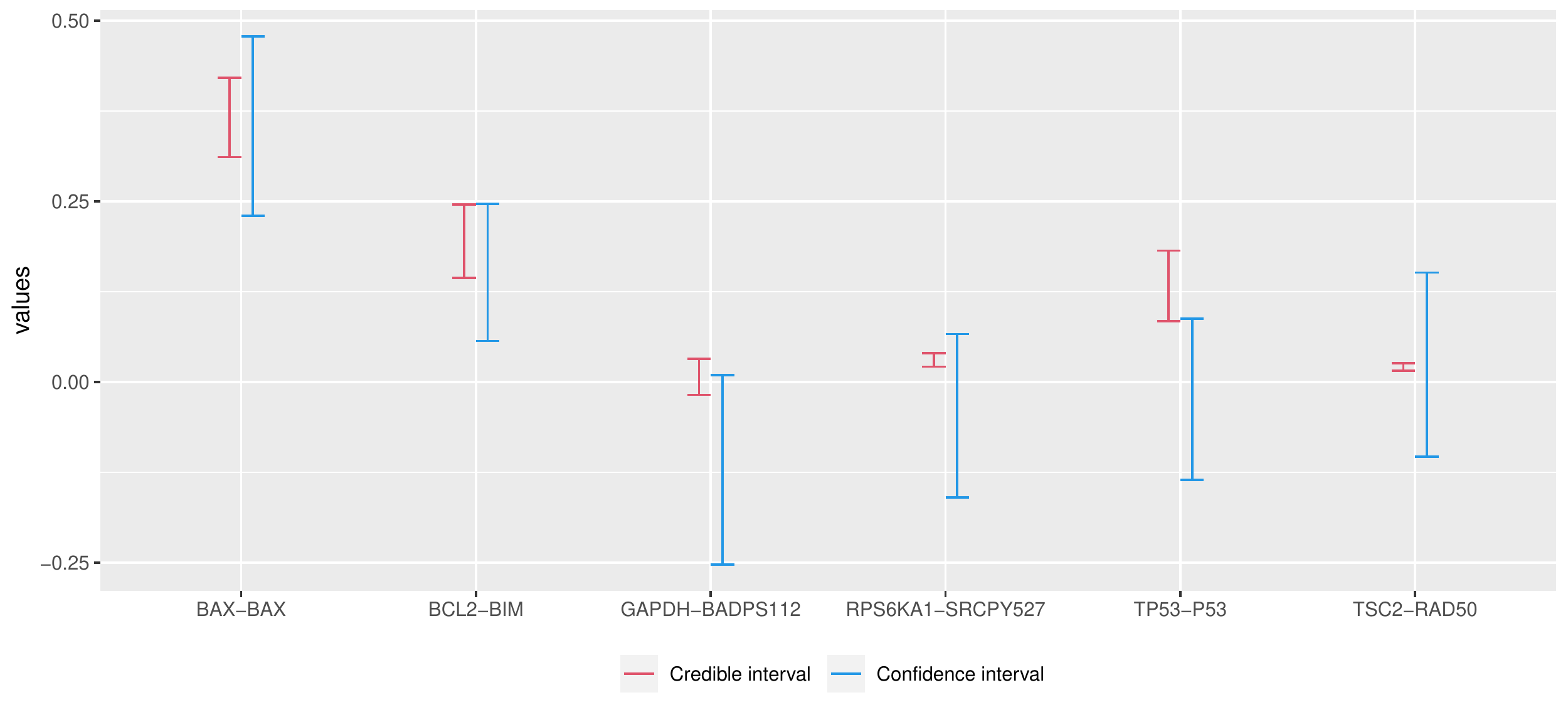}
    \caption{A comparison of the coverage of credible interval by JRNS and Confidence interval by debiased lasso for LUAD lung cancer for a few coordinates of $B$. The x-labels are in the form "gene-protein".}
    \label{fig:LUAD_CI}
\end{figure*}

\begin{table}[H]
    \centering
    \caption{Relative prediction errors using 5-fold cross-validation and  normalized with respect to the vanilla linear regression approach  for JRNS, Stepwise method, DPE and DCPE.}
    \resizebox{0.45\textwidth}{!}{%
    \begin{tabular}{ccccc}
    \toprule
    & JRNS & Stepwise & DPE & DCPE\\
    \midrule
        READ & 0.509 & 0.513 & 0.516 & 0.513\\ 
        LUAD & 0.878 & 0.869 & 0.874 & 0.876\\
        LUSC & 0.838 & 0.839 & 0.834 & 0.836\\
        COAD & 0.887 & 0.881 & 0.886 & 0.884 \\
        OV   & 0.778 &  0.779 & 0.774 & 0.779 \\
        SKCM & 0.859 & 0.860 & 0.859 & 0.860\\
        UCEC & 0.919 & 0.912 & 0.917 & 0.919 \\
       
        \bottomrule
    \end{tabular}
    }%
    \label{tab:pred.error}
\end{table}

       

We also compare the prediction accuracy of the proposed methods with DPE 
and DCPE. Default settings were chosen for these methods as mentioned in 
Section \ref{Sim:results}. Results for HSGHS could not be obtained since we get the same error involving positive definiteness of an intermediate matrix calculation here as well. For prediction evaluation purposes, we perform a 5-fold cross validation in which we randomly divide the data set for each cancer type into 5 parts. The model for each of the listed approaches in Table \ref{tab:pred.error} is built 5 times, each time using one of the parts as the test set and the rest as the training set. The average prediction error is then normalized with respect to that corresponding to the vanilla regression method ($q$ separate response-specific linear models). A relative prediction error less than $1$ implies that the corresponding method has better prediction performance than the vanilla regression approach. All the relative prediction errors are listed in Table \ref{tab:pred.error}. The results show that the proposed methods have a very similar and competitive predictive performance compared to DPE and DCPE, while additionally providing uncertainty quantification by sampling from the posterior distribution.

\section{Discussion}

\noindent
In this paper, we use a biconvex generalized likelihood function along with sparsity
inducing spike-and-slab prior distributions for joint sparsity selection and estimation 
of the regression coefficient and error covariance matrices in multivariate
linear regression models. The proposed JRNS and Stepwise algorithms are 
significantly faster than related (generalized) Bayesian methods both due 
to the simpler algebraic structure of the generalized likelihood used and 
also due to more efficient MCMC implementation (as discussed in Section 
\ref{comparison}), provide samples from the generalized posterior 
distribution for uncertainty quantification, and perform competitively in 
terms of selection/estimation performance in simulated data settings and 
in the TCGA cancer data application. 

Intuitively, the joint JRNS approach should provide better accuracy than 
the Stepwise approach as it utilizes the cross-correlations among the 
errors in estimation of $B$ (while the Stepwise approach ignores them). 
This is borne out in the simulations, especially for Setting 4 with 
$p=q=300$. However, theoretical analysis of the joint generalized posterior
of $(B,\Omega)$ for the JRNS approach is much more complicated than the 
corresponding analysis for the Stepwise approach. One possible direction of
future enquiry is to establish high-dimensional posterior consistency 
results for the joint JRNS approach (analogous to those in 
Theorem \ref{jp:theorem1} for the Stepwise approach). Another possible 
future direction would be to explore the use of the biconvex generalized 
likelihood functions along with continuous shrinkage prior distributions, such as the Horseshoe one and
study the computational and theoretical properties of such an approach.

\section*{Appendix: Pathways for TCGA cancer data}

Table \ref{tab:geneIndices} listes the indices of all the genes and proteins in the LUAD cancer data and Table \ref{tab:pathway} lists all the pathways that have been considered in the analysis of the TCGA cancer data in Section \ref{data:analysis} and their gene members.

    \begin{table*}[htbp]
\centering
    \caption{Indices of genes and proteins for LUAD lung cancer data. The first column lists the components of the dataset mRNA(genes) and the second column lists the components of the dataset RPPA(proteins).}
    \resizebox{0.8\textwidth}{!}{%
    \begin{minipage}{.5\linewidth}
      
      \centering
        \begin{tabular}{rll}
            \hline
 & Gene & Protein \\ 
  \hline
1 & BAK1 & BAK \\ 
  2 & BAX & BAX \\ 
  3 & BID & BID \\ 
  4 & BCL2L11 & BIM \\ 
  5 & CASP7 & CASPASE7CLEAVEDD198 \\ 
  6 & BAD & BADPS112 \\ 
  7 & BCL2 & BCL2 \\ 
  8 & BCL2L1 & BCLXL \\ 
  9 & BIRC2 & CIAP \\ 
  10 & CDK1 & CDK1 \\ 
  11 & CCNB1 & CYCLINB1 \\ 
  12 & CCNE1 & CYCLINE1 \\ 
  13 & CCNE2 & CYCLINE2 \\ 
  14 & CDKN1B & P27PT157 \\ 
  15 & PCNA & P27PT198 \\ 
  16 & FOXM1 & PCNA \\ 
  17 & TP53BP1 & FOXM1 \\ 
  18 & ATM & 53BP1 \\ 
  19 & BRCA2 & ATM \\ 
  20 & CHEK1 & CHK1PS345 \\ 
  21 & CHEK2 & CHK2PT68 \\ 
  22 & XRCC5 & KU80 \\ 
  23 & MRE11A & MRE11 \\ 
  24 & TP53 & P53 \\ 
  25 & RAD50 & RAD50 \\ 
  26 & RAD51 & RAD51 \\ 
  27 & XRCC1 & XRCC1 \\ 
  28 & FN1 & FIBRONECTIN \\ 
  29 & CDH2 & NCADHERIN \\ 
  30 & COL6A1 & COLLAGENVI \\ 
  31 & CLDN7 & CLAUDIN7 \\ 
  32 & CDH1 & ECADHERIN \\ 
  33 & CTNNB1 & BETACATENIN \\ 
  34 & SERPINE1 & PAI1 \\ 
  35 & ESR1 & ERALPHA \\ 
  36 & PGR & ERALPHAPS118 \\ 
  37 & AR & PR \\ 
  38 & INPP4B & AR \\
   \hline
        \end{tabular}
    \end{minipage}%
    \quad \quad
    \begin{minipage}{.5\linewidth}
      \centering
         \begin{tabular}{rll}
            \hline
 & Gene & Protein \\ 
  \hline
 39 & GATA3 & INPP4B \\ 
  40 & AKT1 & GATA3 \\ 
  41 & AKT2 & AKTPS473 \\ 
  42 & AKT3 & AKTPT308 \\ 
  43 & GSK3A & GSK3ALPHABETAPS21S9 \\ 
  44 & GSK3B & GSK3PS9 \\ 
  45 & AKT1S1 & PRAS40PT246 \\ 
  46 & TSC2 & TUBERINPT1462 \\ 
  47 & PTEN & PTEN \\ 
  48 & ARAF & ARAFPS299 \\ 
  49 & JUN & CJUNPS73 \\ 
  50 & RAF1 & CRAFPS338 \\ 
  51 & MAPK8 & JNKPT183Y185 \\ 
  52 & MAPK1 & MAPKPT202Y204 \\ 
  53 & MAPK3 & MEK1PS217S221 \\ 
  54 & MAP2K1 & P38PT180Y182 \\ 
  55 & MAPK14 & P90RSKPT359S363 \\ 
  56 & RPS6KA1 & YB1PS102 \\ 
  57 & YBX1 & EGFRPY1068 \\ 
  58 & EGFR & EGFRPY1173 \\ 
  59 & ERBB2 & HER2PY1248 \\ 
  60 & ERBB3 & HER3PY1298 \\ 
  61 & SHC1 & SHCPY317 \\ 
  62 & SRC & SRCPY416 \\ 
  63 & EIF4EBP1 & SRCPY527 \\ 
  64 & RPS6KB1 & 4EBP1PS65 \\ 
  65 & MTOR & 4EBP1PT37T46 \\ 
  66 & RPS6 & 4EBP1PT70 \\ 
  67 & RB1 & P70S6KPT389 \\ 
  68 & CAV1 & MTORPS2448 \\ 
  69 & MYH11 & S6PS235S236 \\ 
  70 & RAB11A & S6PS240S244 \\ 
  71 & RAB11B & RBPS807S811 \\ 
  72 & GAPDH & CAVEOLIN1 \\ 
  73 & RBM15 & MYH11 \\ 
  74 &  & RAB11 \\ 
  75 &  & GAPDH \\ 
  76 &  & RBM15 \\ 
  \hline
        \end{tabular}
        \end{minipage} 
    }
    \label{tab:geneIndices}
\end{table*}

\begin{table*}[htbp]
\centering
\caption{ Pathways and gene membership }
\resizebox{1.0\textwidth}{!}{%
\begin{tabular}{|c|c|c|}
\hline
    & Pathway & Genes \\
\hline
1 & AKT/PI3K & AKT1, AKT2, AKT3, GSK3A, GSK3B, CDKN1B, AKT1S1, TSC2,
INPP4B, PTEN\\
2 & Apoptosis &  BAK1, BAX, BID, BCL2L11, CASP7, BAD, BCL2, BCL2L1, BIRC2 \\
3 & Breast Reactive & CAV1, MYH11, RAB11A, RAB11B, CTNNB1, GAPDH, RBM15 \\
4 & Cell Cycle & CDK1, CCNB1, CCNE1, CCNE2, CDKN1B, PCNA, FOXM1\\
5 & DNA damage response & TP53BP1, ATM, BRCA2, CHEK1, CHEK2, XRCC5, MRE11A, TP53,RAD50, RAD51, XRCC1  \\
6 & EMT & FN1, CDH2, COL6A1, CLDN7, CDH1, CTNNB1, SERPINE1  \\
7 & Hormone Receptor & ES1, EGR, PR \\
8 & Hormone Signaling (Breast) & INPP4B, GATA3, BCL2\\
9 & RAS & ARAF, JUN, RAF1, MAPK8, MAPK1, MAPK3, MAP2K1, MAPK14,
RPS6KA1, YBX1\\
10 & RTK & EGFR, ERBB2, ERBB3, SHC1, SRC\\
11 & TSC & EIF4EBP1, RPS6KB1, MTOR, RPS6, RB1\\
\hline
\end{tabular}
}%
\label{tab:pathway}
\end{table*}


%
%


\section{Details of the proof of Theorem 1(a), 1(b)}
\label{suppsec:theorem1a}
\subsection{Assumptions required for Theorem 1(a), 1(b)} \label{sec:assumptions}

\noindent
We recall that $\gamma_{jk} = 1_{\{b_{jk} \neq 0\}}$ $ (j = 1, \ldots, p, \; k = 1, \ldots, q),$ and $\gamma = ((\gamma_{jk}))$ represents the sparsity indicator of $B$. Also, $\gamma_t$ denotes the sparsity indicator of the true parameter $B_0$. Let $\gamma_{t_k}(\gamma_k)$ denote the $k$-th column of $\gamma_t(\gamma) (k = 1,\ldots,q)$ and $\nu_{t_k}(\nu_k)$ be the number of non-zero entries in $\gamma_{t_k} (\gamma_k)$. We will consider only the models with sparsity indicator $\gamma$ for which $\nu_k \leq M_n$ for all $k$ where $M_n$ is a realistic model cut-off size (See Assumption \ref{assumpA:q1defn}). Below, for a matrix $A$ we will use the operator norm $\norm{A}_2 = \sqrt{eig_{max}(A'A)}$, the Frobenius norm $\norm{A}_F = \sqrt{\sum_{i}\sum_{j} a_{ij}^2}$ and the norms $\norm{A}_1 = \max \limits_{j}\sum_{i}|a_{ij}|$ and $\norm{A}_{max} = \max\limits_{(i,j)} |a_{ij}|$. Let $\delta > 0.02$ be an arbitrarily fixed constant. Also, we define
$$k_n  = \max\limits_{1\leq k\leq q} \nu_{t_{k}} + 1 \quad \text{and} \quad
s^2_n = \inf\limits_{j,k : B_{0n}(j,k)\neq0}B_{0n}^2(j,k).$$
where $B_{0n}(j,k)$ is the $(j,k)$-th element of $B_{0} = B_{0n}$.

\begin{assumpA}\label{assumpA:eigen}
There exists  $0 < \lambda_1 < \lambda_2 < \infty$ and $0 < \sigma_{min}^2 < \sigma_{max}^2 < \infty,$ not depending on $n$ such that the eigenvalues of all submatrices of $R_0$ are bounded below and above by $\lambda_1$ and $\lambda_2$ respectively, and $$ \sigma_{min}^2 \leq \sigma_{k0,n}^2\leq \sigma_{max}^2 \quad  \text{ for all } n,k.$$ where $\sigma^2_{k0,n} = \sigma^2_{k0}$ is the $k$-th diagonal element of $\Omega^{-1}_0$.

\end{assumpA}

\begin{assumpA}\label{assumpA:q1defn}
$q_1 = (pq)^{-(1+\delta)\kappa}\text{ where } \delta >0,\; \kappa > \frac{2(8\sigma_{max}^2(1+\frac{\delta^{'}}{8})+\epsilon)}{\delta^{'}\sigma_{min}^2} $ for some $\epsilon > 0$ and $\delta^{'} = \frac{5}{64}$  and  
 $M_n= k_0 \frac{n}{\log{(pq)}}$ 
where $k_0 < \text{min} \left(\frac{1}{1024}, \frac{(\delta^{*}\lambda_1)^2}{1024\lambda_2^2}, \frac{(1-2\delta^{'})\sigma_{min}^2}{16\sigma_{max}^2}\right)$ for some $ 0 < \delta^{*} < 1.$
\end{assumpA}

\begin{assumpA} \label{assumpA:tau1}

The slab variance $\tau_1^2$ satisfies $\max\left(\frac{k_n}{n}, \max\limits_{1\leq k\leq q}\frac{\norm{b_{0k}}^2_2}{\log{(pq)}}\right) = o(\tau_1^2)$ where $b_{0k}$ denotes the $k$-th column of $B_0$.
\end{assumpA}

\begin{assumpA} \label{assumpA:sn2}
$k_n\frac{\log(n\tau_1^2)+\log(pq)}{ns^2_n} = o(1)$
\end{assumpA}

\subsection{Proof of Theorem 1(a)} \label{sec:proof}
    Let $\pi(\gamma|Y)$ denote the posterior probability of $\gamma$. Given the true model with sparsity indicator $\gamma_t$ and another arbitrary model with sparsity indicator $\gamma_m$, the ratio of posterior probabilities can be shown to satisfy
    \begin{align} \label{B:defn}
        \frac{\pi(\gamma_m|Y)}{\pi(\gamma_t|Y)} &:=  \prod_{k=1}^q \frac{\pi_k(\gamma_{m_k}|Y)}{\pi(\gamma_{t_k}|Y)}\nonumber \\
        &\leq 8  \prod_{k=1}^q {\left(\frac{2q_1}{\tau_1\sqrt{n}}\right)}^{\nu_{m_k}-\nu_{t_k}}
        \frac{\left|{\frac{X_{m_k}'X_{m_k}}{n} + \frac{I_{\nu_{m_k}}}{n\tau_1^2} }\right|^{-1/2}}{\left|{\frac{X_{t_k}'X_{t_k}}{n} + \frac{I_{\nu_{t_k}}}{n\tau_1^2} }\right|^{-1/2}}
        \left(\frac{S_{{t_k}} + \beta/n}{S_{{m_k}}+\beta/n}\right)^{(n/2 + \alpha)}\\
        &:= 8 \prod_{k=1}^q B(\gamma_{m_k}, \gamma_{t_k})
    \end{align}
    Here $X_{m_k}(X_{t_k})$ represents the submatrix of $X$ consisting of columns corresponding to the active indices in $\gamma_{m_k}(\gamma_{t_k})$, $I_{\nu}$ represents the identity matrix of order $\nu$ and $$S_{m_k} = \frac{y_{.k}^{'}y_{.k}}{n} - \frac{y_{.k}^{'}X_{m_k}}{n}\left(\frac{X_{m_k}^{'}X_{m_k}}{n} + \frac{I_{\nu_{m_k}}}{n\tau_1^2}\right)^{-1}\frac{X_{m_k}^{'}y_{.k}}{n}.$$
    The derivation of (\ref{B:defn}) follows from computations similar to those given in \cite{ghoshstrong}. Let $P_{m_k}$ denote the projection matrix into the column space of $X_{m_k}$ and $$ \tilde{P}_{m_k} = X_{m_k}(\frac{1}{\tau_1^2}I_{\nu_{m_k}} + X_{m_k}^{'}X_{m_k})^{-1} X_{m_k} ^{'}.$$
	We define four events below and show that they occur with  probability tending to 1.
	
	   \begin{align*}
	   G_{1,n} &:= {\bigcap}_{k = 1}^{q}{\bigcap}_{\gamma_{m_k}:1\leq\nu_{m_k}\leq M_n}\left\{\norm{\frac{X_{m_k}^{'}X_{m_k}}{n}-R_{m_k}}_{2}\leq 32\lambda_2\sqrt{\frac{M_n\log{(pq)}}{n}}\right\}\\
	   G_{2,n} &:= {\bigcap}_{k = 1}^{q}{\bigcap}_{\gamma_{m_k}:1\leq\nu_{m_k}\leq \frac{n}{2}}\left\{ \varepsilon_{.k}^{'}P_{m_k}\varepsilon_{.k} \leq 8\sigma_{k0}^2\nu_{m_k}\log{(pq)} \right\}\\
	   G_{3,n} &:= {\bigcap}_{k=1}^{q}\left\{(1-\delta^{'})\sigma_{min}^2\leq \frac{\varepsilon_{.k}^{'}\varepsilon_{.k}}{n}\leq (1+\delta^{'})\sigma_{max}^2 \right\}\\
	   G_{4,n} &:= {\bigcap}_{k=1}^{q} {\bigcap}_{\gamma_{m_k}:\gamma_{m_k}\supset\gamma_{t_k},\nu_{m_k}\leq \frac{n}{2}}\left\{ \varepsilon_{.k}^{'}(P_{{m_k}}-P_{{t_k}})\varepsilon_{.k} \leq 8\sigma_{k0}^2(\nu_{m_k}-\nu_{t_k})\log{(pq)}\right\}
	    \end{align*} where $R_{m_k}$ represents the sub-matrix of R consisting of rows and columns corresponding to the active indices in $\gamma_{m_k}.$
	  We also define $G_{n} := G_{1,n} \cap G_{2,n} \cap G_{3,n} \cap G_{4,n}.$\\ 
	  \vskip10pt
	  Using Theorem 6.2.1 from \cite{vershynin2018highdim} and Lemma F.2 from \cite{basu2015regularized}  we get 
	  
	  	  \begin{align*}
	      \mathbb{P}_0\left(\norm{\frac{X_{m_k}^{'}X_{m_k}}{n}-R_{m_k}}_{2} \geq 32\lambda_2\sqrt{\frac{M_n\log(pq)}{n}}\right) \leq 2(pq)^{-3\nu_{m_k}}.
	  \end{align*}
	  Hence \begin{align}
	      \mathbb{P}_0\left(G_{1,n}^{c}\right) &\leq \sum_{k=1}^{q}\sum_{\gamma_{m_k}:1\leq\nu_{m_k}\leq M_n} 2(pq)^{-3\nu_{m_k}}\nonumber \\
	      &\leq\sum_{k=1}^{q}\sum_{i=1}^{M_n} \binom{p}{i}2(pq)^{-3i}\nonumber \\
	      &\leq 2q^{-3}\sum_{k=1}^{q}\sum_{i=1}^{M_n}p^{i}p^{-3i}\nonumber \\
	      &\leq 2q^{-3}\sum_{k=1}^{q}\sum_{i=1}^{\infty}p^{-2i}\nonumber \\
	      &\leq \frac{2}{q^2(p^2-1)} \to 0 \text{ as } n\to \infty.
	  \end{align}
	  
	  Using Lemma 4.1 from \cite{cao2020high} and the fact that $\varepsilon_{.k}^{'} P_{{m_k}}\varepsilon_{.k} \sim \sigma_{k0}^2 \chi_{\nu_{m_k}}^2$, it can be shown that
	  \begin{align}
	      \mathbb{P}_0\left(\varepsilon_{.k}^{'}P_{{m_k}}\varepsilon_{.k} \geq 8\sigma_{k0}^2\nu_{m_k}\log{(pq)}\right) &\leq 2(pq)^{-\frac{3}{2}\nu_{m_k}}.\nonumber
	  \end{align}
	  Hence,
	  \begin{align}
	      \mathbb{P}_0\left(G_{2,n}^{c}\right) &\leq \sum_{k=1}^q \sum_{\gamma_{m_k}:1\leq\nu_{m_k}\leq n/2} 2(pq)^{-\frac{3}{2}\nu_{m_k}} \nonumber \\
	      &\leq \sum_{k=1}^q \sum_{i =1} ^{ n/2} 2 \binom{p}{i} p^{-\frac{3i}{2}}q^{-\frac{3}{2}}\nonumber \\
	      & \leq 2 q^{-3/2} \sum_{k=1}^q \sum_{i=1}^{n/2}p^{i} p^{-\frac{3i}{2}} \nonumber \\
	      & \leq  \frac{2}{\sqrt{q}(\sqrt{p}-1)} \to 0 \text{ as } n \to \infty.
	      \end{align}
	  \begin{align}
	      \mathbb{P}_0\left(G_{3,n}^{c}\right) &\leq \sum_{k=1}^{q}\mathbb{P}_0\left(\left|\frac{1}{n}\varepsilon_{.k}^{'}\varepsilon_{.k}-\sigma_{k0}^2\right|> \delta^{'}\sigma_{k0}^2\right)\nonumber \\
	      &= \sum_{k=1}^{q}\mathbb{P}_0\left(\left|\frac{1}{\sigma_{k0}^2}\varepsilon_{.k}^{'}\varepsilon_{.k}-n\right|> \delta^{'}n\right)\nonumber \\
	      &= \sum_{k=1}^q P\left(|\chi_n^2 - n| > \delta^{'}n \right)\nonumber\\
	      &\leq 2q \exp\left[- \frac{(\delta^{'})^2n}{4(1+\delta^{'})} \right]  \to 0 \text{ as } n \to \infty.
	  \end{align}
	  
	  Here we use an upper bound for $P\left(|\chi_p^2 - p| > a \right)$ as obtained in the proof of Lemma 4.1 of \cite{cao2020high}. Using arguments similar to those in the proof for $G_{2,n}$ it can be shown that $\varepsilon_{.k}^{'}( P_{{m_k}}- P_{{t_k}})\varepsilon_{.k} \sim \sigma_{k0}^2 \chi_{(\nu_{m_k}-\nu_{t_k})}^2$ and that
	   \begin{align}
	      \mathbb{P}_0\left(\varepsilon_{.k}^{'}(P_{{m_k}}-P_{{t_k}})\varepsilon_{.k} \geq 8\sigma_{k0}^2\nu_{m_k}\log{(pq)}\right) &\leq 2(pq)^{-\frac{3}{2}(\nu_{m_k}-\nu_{t_k})}\nonumber
	  \end{align}
	  It then follows that
	  \begin{align}
	      \mathbb{P}_0\left(G_{4,n}^{c}\right) &\leq \sum_{k=1}^q \sum_{\gamma_{m_k}:\gamma_{m_k} \supset \gamma_{t_k}, \nu_{m_k}< n/2} 2(pq)^{-\frac{3}{2}\nu_{m_k}} \nonumber \\
	      &\leq \sum_{k=1}^q \sum_{i=1}^{n/2} 2 \binom{p- \nu_{t_k}}{i} p^{-\frac{3i}{2}}q^{-\frac{3}{2}}\nonumber \\
	      & \leq 2 q^{-3/2} \sum_{k=1}^q \sum_{i=1}^{n/2}p^{i} p^{-\frac{3i}{2}} \nonumber \\
	      &\leq 2 q^{-3/2} \sum_{k=1}^q \sum_{i=1}^{\infty} p^{-\frac{i}{2}} \nonumber \\
	      & \leq  \frac{2}{\sqrt{q}(\sqrt{p}-1)} \to 0 \text{ as } n \to \infty.
	      \end{align}
	  
	We now state and prove two lemmas which will be used to prove Theorem 1(a).  	
	
	\begin{lem}
	\label{lem1}
		If for a particular $k \; ( k = 1,2,\ldots,q), \gamma_{m_{k}} \supset \gamma_{t_{k}} $, then there exists $N_1$ (not depending on $m$ or $k$) such that for all $n \geq N_1$ on the set $G_{n}$, we have 
		$$  B(\gamma_{m_{k}},\gamma_{t_{k}}) \leq (pq)^{-(1+\delta)(\nu_{m_{k}}-\nu_{t_{k}})}$$
	
	\end{lem}
	\begin{proof}
		For any fixed $k$,
		\begin{align}\label{eqn1}
	   S_{{t_k}} &= \frac{1}{n}y_{.k}^{'}(I - \tilde{P}_{{t_k}})y_{.k}\nonumber\\
	   &= \frac{1}{n}y_{.k}^{'}(I - P_{{t_k}})y_{.k} + \frac{1}{n}y_{.k}^{'}( P_{{t_k}} - \tilde{P}_{{t_k}})y_{.k}\nonumber\\
	   &=\frac{1}{n}\varepsilon_{.k}^{'}(I - P_{{t_k}})\varepsilon_{.k} + \frac{1}{n}y_{.k}( P_{{t_k}} - \tilde{P}_{{t_k}})y_{.k}
	\end{align} and,
	\begin{align}\label{eqn2}
	 S_{{m_k}}&=\frac{1}{n}y_{.k}^{'}(I - \tilde{P}_{{m_k}})y_{.k}\nonumber\\
	 &\geq\frac{1}{n}y_{.k}^{'}(I - P_{{m_k}})y_{.k}\nonumber\\
	 &= \frac{1}{n}\varepsilon_{.k}^{'}(I - P_{{m_k}})\varepsilon_{.k}  
	\end{align}
		Hence using (\ref{B:defn}), (\ref{eqn1}) and (\ref{eqn2}),  $B(\gamma_{m_k}, \gamma_{t_k})$ can be written as
		
		\begin{align} \label{lemma2.1}
		    B(\gamma_{m_k}, \gamma_{t_k}) &\leq \left(\frac{2q_1}{\tau_1\sqrt{n}}\right)^{\nu_{m_k}-\nu_{t_k}}\frac{\left|{\frac{X_{m_k}'X_{m_k}}{n} + \frac{I_{\nu_{m_k}}}{n\tau_1^2} }\right|^{-1/2}}{\left|{\frac{X_{t_k}'X_{t_k}}{n} + \frac{I_{\nu_{t_k}}}{n\tau_1^2} }\right|^{-1/2}}\nonumber\\&\qquad\times\left(1 +\frac{\varepsilon_{.k}^{'}(P_{m_{k}}-P_{t_{k}})\varepsilon_{.k}/n + y_{.k}^{'}(P_{t_{k}}-\tilde{P}_{t_{k}})y_{.k}/n }{\varepsilon_{.k}^{'}(I-P_{m_{k}})\varepsilon_{.k}/n + 2\beta/n}\right)^{n/2 + \alpha}
		\end{align}
	  Next using Woodbury's Identity and Assumptions \ref{assumpA:q1defn} and \ref{assumpA:tau1} we have on $G_{2,n}$,
	\begin{align} \label{lemma2.num2}
	   &\quad y_{.k}^{'}(P_{t_{k}}-\tilde{P}_{t_{k}})y_{.k}\nonumber\\ 
	 &= y_{.k}^{'}X_{t_{k}}(X_{t_{k}}^{'}X_{t_{k}})^{-1}\left[X_{t_{k}}^{'}X_{t_{k}} - X_{t_{k}}^{'}X_{t_{k}}(X_{t_{k}}^{'}X_{t_{k}} +\frac{I}{\tau_1^2})^{-1}X_{t_{k}}^{'}X_{t_{k}}\right](X_{t_{k}}^{'}X_{t_{k}})^{-1}X_{t_{k}}^{'}y_{.k}\nonumber\\
	 &=  y_{.k}^{'}X_{t_{k}}(X_{t_{k}}^{'}X_{t_{k}})^{-1}[\tau_1^2 I_{t_k} + (X_{t_{k}}^{'}X_{t_{k}})^{-1}]^{-1}(X_{t_{k}}^{'}X_{t_{k}})^{-1}X_{t_{k}}^{'}y_{.k}\nonumber\\
	 & \leq \frac{1}{\tau_1^2}y_{.k}^{'}X_{t_{k}}(X_{t_{k}}^{'}X_{t_{k}})^{-2}X_{t_{k}}^{'}y_{.k}\nonumber\\
	 &\leq \frac{2}{\tau_1^2}b_{0t_{k}}^{'}b_{0t_{k}} + \frac{2}{\tau_1^2}\varepsilon_{.k}^{'}X_{t_{k}}(X_{t_{k}}^{'}X_{t_{k}})^{-2}X_{t_{k}}^{'}\varepsilon_{.k}\nonumber\\
	 &\leq \text{log}(pq) o(1)
	 	\end{align}
	 	where $o(1) \to 0$ uniformly in $m$ and $k$. Note that $M_n \leq \frac{n}{2}$ for all sufficiently large $n$. Hence for $\nu_{m_k} < M_n ,\nu_{m_k} < n/2$ and on $G_{4,n}$ we have
	 	\begin{align} \label{lemma2.num1}
	 	   \varepsilon_{.k}^{'}(P_{m_k}- P_{t_k})\varepsilon_{.k} 
	 	   &\leq 8\sigma_{max}^2(\nu_{m_k}-\nu_{t_k})\log (pq).
	 	\end{align}
	 	On $G_{2,n} \cap G_{3,n}$ we have for all $\gamma_{m_k}$ with $\nu_{m_k} \leq M_n$
	 	\begin{align}\label{lemma2.denom}
	 	    	\frac{\varepsilon_{.k}^T(I- P_{m_k})\varepsilon_{.k}}{n} 
	 	    	&=\frac{\varepsilon_{.k}^T\varepsilon_{.k}}{n}-\frac{\varepsilon_{.k}^TP_{m_k}\varepsilon_{.k}}{n}\nonumber\\
	 	    	&\geq (1-\delta^{'})\sigma_{min}^2-\frac{8\sigma_{max}^2\nu_{m_k}\log (pq)}{n}\nonumber\\
	 	    	&\geq (1-\delta^{'})\sigma_{min}^2-8\sigma_{max}^2 k_0\nonumber\\
	 	    	&\geq\delta^{'}\sigma_{min}^2
	 	\end{align}
	 
	 

	 by Assumption \ref{assumpA:q1defn}. Since $\gamma_{m_k} \supset \gamma_{t_k}$ it can be shown that
	 \begin{align} \label{lemma2.det}
	     (n\tau_1^2)^{(\nu_{t_{k}}-\nu_{m_{k}})/2}\frac{\left|{\frac{X_{m_k}'X_{m_k}}{n} + \frac{I_{\nu_{m_k}}}{n\tau_1^2} }\right|^{-1/2}}{\left|{\frac{X_{t_k}'X_{t_k}}{n} + \frac{I_{\nu_{t_k}}}{n\tau_1^2} }\right|^{-1/2}} &\leq 1.
	 \end{align}
     Finally using (\ref{lemma2.1}),(\ref{lemma2.num2}),(\ref{lemma2.num1}),(\ref{lemma2.denom}),(\ref{lemma2.det}) we get
     \begin{align}
        B(\gamma_{m_k}, \gamma_{t_k}) &\leq (2q_1)^{(\nu_{m_k}-\nu_{t_k})}\left[1+\frac{8\sigma_{max}^2(\nu_{m_k}-\nu_{t_k})\log(pq) + \log(pq)o(1)}{\delta^{'}\sigma_{min}^2n}\right]^{\frac{n}{2} + \alpha}\nonumber\\
	   &\leq (pq)^{-(1+\delta)(\nu_{m_k}-\nu_{t_k})}. \nonumber
     \end{align}
     The last inequality follows from Assumption \ref{assumpA:q1defn} and the inequality $(1+x) \leq e^x$.
	\end{proof}
	\begin{lem}
	\label{lem2}
		If for a particular $k \; ( k = 1,2,\ldots,q), \gamma_{m_{k}}$ is such that $\gamma_{m_{k}}^c \cap \gamma_{t_{k}} \neq \phi$, then there exists $N_2$ (not depending on $m$ or $k$) such that for all $n \geq N_2$ on the set $G_{n}$, we have $$ B(\gamma_{m_{k}},\gamma_{t_{k}}) \leq (pq)^{-(1+\delta)(\nu_{m_{k}}-\nu_{t_{k}})}$$ if $\nu_{m_{k}} > (1+\frac{8}{\delta^{'}})\nu_{t_{k}}$, and
		$$ B(\gamma_{m_{k}},\gamma_{t_{k}}) \leq (pq)^{-(1+\delta)(1+\frac{8}{\delta^{'}})\nu_{t_{k}}}$$
		 if $\nu_{m_{k}} \leq (1+\frac{8}{\delta^{'}})\nu_{t_{k}}$.
	\end{lem}
	
	\begin{proof}
	
	Let $\gamma_{{\tilde{m}}_k} = \gamma_{m_k}\cup\gamma_{t_k}$.
	Then
	\begin{align}\label{smtilde}
	 S_{m_k}&=\frac{1}{n}y_{.k}(I - \tilde{P}_{{m_k}})y_{.k}
	 \geq\frac{1}{n}y_{.k}(I - P_{{m_k}})y_{.k}\nonumber\\
	 &\geq \frac{1}{n}y_{.k}(I - P_{{{\tilde{m}}_k}})y_{.k}  \nonumber\\
	 &= \frac{1}{n}\varepsilon_{.k}(I - P_{\gamma_{{\tilde{m}}_k}})\varepsilon_{.k}.  
	\end{align}
	Using (\ref{eqn1}) and (\ref{smtilde}) it can be shown that
	$$ \frac{S_{t_k} + 2\beta/n}{S_{m_k} + 2\beta/n} \leq 1 + \frac{\frac{1}{n}\varepsilon_{.k}^{'}( P_{{\tilde{m}}_k}-P_{t_k} )\varepsilon_{.k} +\frac{1}{n}y_{.k}^{'}( P_{t_k} - \tilde{P}_{t_k})y_{.k}}{\frac{1}{n}\varepsilon_{.k}^{'}(I - P_{{\tilde{m}}_k})\varepsilon_{.k} + 2\beta/n}$$ and hence we get
	\begin{align} \label{lemma3.1}
		    B(\gamma_{m_k}, \gamma_{t_k}) &\leq \left(\frac{2q_1}{\tau_1\sqrt{n}}\right)^{\nu_{m_k}-\nu_{t_k}}\frac{\left|{\frac{X_{m_k}'X_{m_k}}{n} + \frac{I_{\nu_{m_k}}}{n\tau_1^2} }\right|^{-1/2}}{\left|{\frac{X_{t_k}'X_{t_k}}{n} + \frac{I_{\nu_{t_k}}}{n\tau_1^2} }\right|^{-1/2}}\nonumber\\&\qquad\times\left(1 +\frac{\varepsilon_{.k}^{'}(P_{{\tilde{m}}_k}-P_{t_{k}})\varepsilon_{.k}/n + y_{.k}^{'}(P_{t_{k}}-\tilde{P}_{t_{k}})y_{.k}/n }{\varepsilon_{.k}^{'}(I-P_{{\tilde{m}}_k})\varepsilon_{.k}/n + 2\beta/n}\right)^{n/2 + \alpha}.
		\end{align}
	\textbf{CASE I:}  $(1+\frac{8}{\delta^{'}})\nu_{t_k}<\nu_{m_k}\leq M_n$
	\vskip2pt
	  For all sufficiently large $n, \nu_{{\tilde{m}}_k} < n/2$ and $\gamma_{{\tilde{m}}_k} = \gamma_{m_k}\cup\gamma_{t_k} \supset\gamma_{t_k}.$ 
	Thus, on $G_{4,n}$ we have 
	\begin{align}
	    \varepsilon_{.k}^T(P_{{\tilde{m}}_k}-P_{t_{k}})\varepsilon_{.k} &\leq 8\sigma_{max}^2(\nu_{{\tilde{m}}_k}-\nu_{t_k})\log(pq)\nonumber\\
	    &\leq 8\sigma_{max}^2\nu_{m_k}\log(pq)\nonumber\\
	    &\leq 8\sigma_{max}^2(1+\delta^{'}/8)(\nu_{m_k}-\nu_{t_k})\log(pq) .
	\end{align}
	We have already shown in the proof of Lemma 1 that
	\begin{align*}
	    y_{.k}^{'}(P_{t_{k}}-\tilde{P}_{t_{k}})y_{.k} \leq \log(pq)o(1).
	\end{align*}
	As $\nu_{{\tilde{m}}_k} \leq n/2$ for all large $n$, using arguments similar to the proof of Lemma 1 it can be shown that on $G_{2,n}\cap G_{3,n}$ we have
    \begin{align}
	 	    	\frac{1}{n} \varepsilon_{.k}^T(I- P_{\tilde{m}_k})\varepsilon_{.k}
	 	    	&\geq\delta^{'}\sigma_{min}^2.
	 	\end{align}
	On $G_{1,n}$,
	\begin{align}
	    \left|{\frac{X_{t_k}'X_{t_k}}{n} + \frac{I_{\nu_{t_k}}}{n\tau_1^2} }\right|^{1/2} < (2\lambda_2)^{(\nu_{t_k}/2)} \nonumber
	\end{align} and
	\begin{align}
	    \left|{\frac{X_{m_k}'X_{m_k}}{n} + \frac{I_{\nu_{m_k}}}{n\tau_1^2} }\right|^{1/2} > ((1-\delta^{*})\lambda_1)^{(\nu_{m_k}/2)} \nonumber
	\end{align} for $0<\delta^{*}<1.$
	Then on $G_{1,n}$,
	\begin{align}
	    \frac{\left|{\frac{X_{m_k}'X_{m_k}}{n} + \frac{I_{\nu_{m_k}}}{n\tau_1^2} }\right|^{-1/2}}{\left|{\frac{X_{t_k}'X_{t_k}}{n} + \frac{I_{\nu_{t_k}}}{n\tau_1^2} }\right|^{-1/2}} 
	    & \leq  C^{(\nu_{m_k}-\nu_{t_k})} 
	\end{align}
	for some appropriate constant $C$. Hence from (\ref{lemma3.1}) we have
	\begin{align}\label{eq13}
	   B(\gamma_{m_k}, \gamma_{t_k}) 
	   &\leq (pq)^{-(1+\delta)(\nu_{m_k}-\nu_{t_k})}\left[2(pq)^{-\left(\frac{\kappa}{2}-\frac{\left(8\sigma_{max}^2\left(1+\frac{\delta^{'}}{8}\right)+\epsilon\right)}{\delta^{'}\sigma_{min}^2}\right)}\right]^{\nu_{m_k}-\nu_{t_k}}\nonumber\\&\qquad\times({n\tau_1^2})^{(\nu_{t_k}-\nu_{m_k})/2} C^{(\nu_{m_k}-\nu_{t_k})}\nonumber\\
	   &\leq (pq)^{-(1+\delta)(\nu_{m_k}-\nu_{t_k})} 
	\end{align}
	by Assumption \ref{assumpA:q1defn}.\\

 		\textbf{CASE II:}  $\nu_{m_k} \leq (1+\frac{8}{\delta^{'}})\nu_{t_k}$
 		\vskip2pt
	 Let $\gamma_{a_k} = \gamma_{m_k}^c \cap \gamma_{t_k} \text{ and }
   \gamma_{m_k \cap t_k} = \gamma_{m_k} \cap \gamma_{t_k}.$ Also, let $b_{0t_k}, b_{0a_k}$ and $ b_{0m_k \cap t_k}$ denote the vectors consisting of the elements of $b_{0k}$ ( the $k$-th column of $B_0$) which correspond to the active indices of $\gamma_{t_k}$, $\gamma_{a_k}$ and $\gamma_{m_k \cap t_k}$ respectively.	We first find a lower bound for $S_{{m_k}} - S_{{t_k}}$. Using Woodbury's identity it can be shown that
   
   \begin{align}
       S_{{m_k}} - S_{{t_k}} &\geq \frac{y_{.k}^{'}(P_{t_k} - P_{m_k})y_{.k}}{n} - o(s_n^2) \nonumber\\
       &= \frac{\varepsilon_{.k}^{'}(P_{t_k} - P_{m_k})\varepsilon_{.k}}{n} + \frac{b_{0t_{k}}^{'}X_{t_{k}}^{'}(P_{t_k} - P_{m_k})X_{t_{k}}b_{0t_{k}}}{n} \nonumber\\
       &\qquad + \frac{2b_{0t_{k}}^{'}X_{t_{k}}^{'}(P_{t_k} - P_{m_k})X_{t_{k}}\varepsilon_{.k}}{n} - o(s_n^2).
   \end{align}
   We show that the second term is the dominating term and is bounded below by $(1-\delta^{*})\lambda_1s_n^2$ where $0< \delta^{*}<1$ is as in Assumption \ref{assumpA:q1defn}. 
   Without loss of generality, we assume that $ X_{t_k}$ is composed as $[X_{m_k \cap t_k}| X_{a_k}]$ where $X_{m_k \cap t_k} = X_{\gamma_{m_k \cap t_k}}$
   \begin{align}
       X_{t_k}^{'}(P_{t_k} - P_{m_k}) X_{t_k} &=\begin{pmatrix}
       X_{m_k\cap t_k}^{'}\\X_{a_k}^{'}
       \end{pmatrix}
       \begin{pmatrix}
       P_{t_k} - P_{m_k}
       \end{pmatrix}
       \begin{pmatrix}
       X_{m_k \cap t_k} |& X_{a_k}
       \end{pmatrix}\nonumber\\
      &=\begin{pmatrix}
      (P_{t_k}X_{_{m_k \cap t_k}})^{'}-(P_{m_k}X_{_{m_k \cap t_k}})^{'}\\
      (P_{t_k}X_{a_k})^{'}-(P_{m_k}X_{a_k})^{'}
      \end{pmatrix}
      \begin{pmatrix}
       X_{m_k \cap t_k} |& X_{a_k}
       \end{pmatrix}\nonumber\\
      &=\begin{pmatrix}
      0 & 0\\
      0 & X_{a_k}^{'}(I-P_{m_k})X_{a_k}
      \end{pmatrix}.
   \end{align}
   Hence,
   \begin{align}
      b_{0t_{k}}^{'}X_{t_{k}}^{'}(P_{t_k} - P_{m_k})X_{t_{k}}b_{0t_{k}} &= \begin{pmatrix}
      b_{0m_k\cap t_k}^{'} & b_{0a_{k}}^{'}
      \end{pmatrix}\begin{pmatrix}
      0 & 0\\
      0 & X_{a_k}^{'}(I - P_{m_k})X_{a_k}
      \end{pmatrix}
      \begin{pmatrix}
      b_{0m_k\cap t_k}\\
      b_{0a_{k}}
      \end{pmatrix}\nonumber\\
      &= b_{0a_{k}}^{'}X_{a_k}^{'}(I - P_{{m_k}})X_{a_k} b_{0a_{k}}.
   \end{align}
   
  Now by Lemma S1.4 of \cite{ghoshstrong} there exists a $(\nu_{m_k} + \nu_{a_k}) \times 1$ vector $u$ such that
  $$ \frac{1}{n} b_{0t_{k}}^{'}X_{t_{k}}^{'}(P_{{t_k}} - P_{{m_k}})X_{t_{k}}b_{0t_{k}} = \frac{1}{n} b_{0a_{k}}^{'}X_{a_k}^{'}(I - P_{{m_k}})X_{a_k} b_{0a_{k}} = \frac{1}{n} u^{'}X_{m_k \cup a_k}^{'}X_{m_k \cup a_k}u$$
  where $u^{'} = (u_{m_k}^{'} | b_{0a_k}^{'}) $ and $\norm{u} ^2 \geq \norm{b_{0a_k}}^2 \geq \nu_{a_k}s_n^2 \geq s_n^2$.\\
  On $G_{1,n}$ we have
  \begin{align}\label{lemma2:S:upper}
      \frac{1}{n} b_{0t_{k}}^{'}X_{t_{k}}^{'}(P_{t_k} - P_{m_k})X_{t_{k}}b_{0t_{k}} &= \frac{1}{n} u^{'}X_{m_k \cup a_k}^{'}X_{m_k \cup a_k}u\nonumber\\&=
      \frac{u'}{\norm{u}} \left(\frac{1}{n}X_{m_k \cup a_k}^{'}X_{m_k \cup a_k}\right)\frac{u}{\norm{u}}\norm{u}^2\nonumber\\
      &\geq \norm{u}^2 \inf_{\norm{v}=1}v^{'}\left(\frac{1}{n}X_{m_k \cup a_k}^{'}X_{m_k \cup a_k}\right)v \nonumber\\
      &\geq (1-\delta^{*})\lambda_1 s_1^2
  \end{align}
  by Assumption \ref{assumpA:q1defn}.
  Thus,
  $$ b_{0t_{k}}^{'}X_{t_{k}}^{'}(P_{{t_k}} - P_{{m_k}})X_{t_{k}}b_{0t_{k}} \geq n(1-\delta^{*})\lambda_1 s_1^2. $$
  By Assumption \ref{assumpA:sn2}, on the set $G_{2,n}$,
  $$ \frac{1}{n}\varepsilon_{.k}^{'}(P_{{t_k}} - P_{{m_k}})\varepsilon_{.k} = o(s_n^2)$$ 
  Now, since $(P_{{t_k}} - P_{{ m_k}})X_{t_k} = P_{\gamma_{t_k}\cap \gamma_{m_k^c}}X_{t_k}$, we have
  \begin{align*}
      b_{0t_{k}}^{'}X_{t_{k}}^{'}(P_{{t_k}} - P_{{m_k}})\varepsilon_{.k} &\leq \sqrt{b_{0t_{k}}^{'}X_{t_{k}}^{'}P_{\gamma_{t_k}\cap \gamma_{m_k^c}}X_{t_{k}}b_{0t_{k}}} \sqrt{\varepsilon_{.k}^{'}P_{\gamma_{t_k}\cap \gamma_{m_k^c}}\varepsilon_{.k}}\\
      &= o(b_{0t_{k}}^{'}X_{t_{k}}^{'}(P_{{t_k}} - P_{{m_k}})X_{t_{k}}b_{0t_{k}} )
  \end{align*}
  
  Also on $G_{3,n}$ and using (\ref{lemma2.num2}),
  \begin{align}\label{lemma2:S:denom}
      S_{{t_k}} &= \frac{y_{.k}^{'}(I- \tilde{P}_{t_k})y_{.k}}{n} \nonumber\\
      &= \frac{\varepsilon_{.k}^{'}\varepsilon_{.k}}{n} + \frac{y_{.k}^{'}(P_{t_k}- \tilde{P}_{t_k})y_{.k}}{n} \nonumber\\
      & \leq (1+\delta^{'})\sigma_{max}^2 + o(1) \nonumber
  \end{align}
  So, 
  \begin{align}
      S_{{t_k}} + 2\beta/n &\leq c_1 \qquad \text{for some appropriate constant } c_1.
       \end{align} and hence,
       
    \begin{align}\label{lemma2:Sterm}
        \left(\frac{S_{{m_k}} +2\beta/n}{S_{{t_k}} +2\beta/n}\right)^{-(\frac{n}{2}+\alpha)} 
        & = \left(1+\frac{S_{{m_k}}-S_{{t_k}}}{S_{{t_k}} +2\beta/n}\right)^{-(\frac{n}{2}+\alpha)} \nonumber\\
        & \leq (1+c_2s_n^2)^{-(\frac{n}{2}+\alpha)} 
    \end{align} 
    for some appropriate constant $c_2$.
    Now
    \begin{align}\label{lemma2:tau2term}
     (n\tau_1^2)^{\frac{\nu_{t_k}-\nu_{m_k}}{2}} \frac{\left|{\frac{X_{m_k}'X_{m_k}}{n} + \frac{I_{\nu_{m_k}}}{n\tau_1^2} }\right|^{-1/2}}{\left|{\frac{X_{t_k}'X_{t_k}}{n} + \frac{I_{\nu_{t_k}}}{n\tau_1^2} }\right|^{-1/2}} 
     &\leq \frac{(n\tau_1^2)^{(\nu_{t_k}-\nu_{m_k})/2}(2\lambda_2)^{\nu_{t_k}/2} }{\left|\frac{I_{\nu_{m_k}}}{n\tau_1^2} \right|^{1/2}} \nonumber\\
     & \leq (2n\tau_1^2\lambda_2)^{\nu_{t_k}/2}.
    \end{align}
    
  From  (\ref{B:defn}), (\ref{lemma2:S:upper}), (\ref{lemma2:Sterm}), (\ref{lemma2:tau2term}) and using Assumption \ref{assumpA:sn2} we get
  
  \begin{align*}
      B(\gamma_{m_k}, \gamma_{t_k}) &\leq (pq)^{-(1+\delta)(1+\frac{8}{\delta^{'}})(\nu_{m_k}-\nu_{t_k})}.
  \end{align*}
  
  	\end{proof}
\noindent \textbf{\textit{Proof of Theorem 1(a)}.} We first prove that $$\pi (\gamma_t |Y) \overset{\mathbb{P}_0}{\longrightarrow} 1 \text{ as } n \to \infty$$.
Let $N_1, N_2$ and $G_n$ be as in Lemmas 1 and 2 and $N=\max(N_1, N_2)$. Then for all
$n \geq N$ and $k=1, \ldots, q$, on the set $G_n$, 
\begin{eqnarray}
 &  & \frac{1-\pi_k(\gamma_{t_k} | Y)}{\pi_k(\gamma_{t_k} | Y)} \nonumber \\
 & = & \sum_{\gamma_{m_k} \neq \gamma_{t_k}} \frac{\pi_k(\gamma_{m_k} | Y)}{\pi_k(\gamma_{t_k} | Y)} \nonumber \\
 & \leq &  \sum_{\gamma_{m_k}: \gamma_{m_k} \supset \gamma_{t_k}, \nu_{m_k} \leq M_n} \frac{\pi_k(\gamma_{m_k} | Y)}{\pi_k(\gamma_{t_k} | Y)}  
          + \sum_{\gamma_{m_k}: \gamma_{m_k}^c \cap \gamma_{t_k} \neq \phi , (1+8/\delta^{'} )\nu_{t_k} < \nu_{m_k} \leq M_n} 
            \frac{\pi_k(\gamma_{m_k} | Y)}{\pi_k(\gamma_{t_k} | Y)} \nonumber \\
 & &   + \sum_{\gamma_{m_k}: \gamma_{m_k}^c \cap \gamma_{t_k} \neq \phi , \nu_{m_k} \leq (1+8/\delta^{'} )\nu_{t_k}} 
        \frac{\pi_k(\gamma_{m_k} | Y)}{\pi_k(\gamma_{t_k} | Y)} \label{proof1} \\
 & \leq &  8 \sum_{i=\nu_{t_k}+1}^{M_n} {p-\nu_{t_k} \choose i- \nu_{t_k}} 	(pq)^{-(1+\delta)(i- \nu_{t_k})} 
         + 8 \sum_{(1+8/\delta^{'})\nu_{t_k} < i \leq M_n} {p \choose i} (pq)^{-(1+\delta)(i - \nu_{t_k})} \nonumber \\
 & &     + 8 \sum_{0 \leq i \leq (1+8/\delta^{'})\nu_{t_k}} {p \choose i} (pq)^{-(1+\delta)(1+8/\delta^{'})\nu_{t_k}} 
 \label{proof2}
\end{eqnarray}
by Lemmas 1 and 2. 
In order to find a bound for (\ref{proof2}), we use the inequalities ${p \choose i} \leq p^i$, $\sum_{i=0}^r p^i \leq 2 p^r$ and 
$q^{-(1+\delta)r} \leq q^{-(1+\delta )}$ for $r \geq 1$. We also note that $\delta > \delta^{'} /8$ so that $(1+\delta)8 / (8+\delta^{'}) > 1$ 
and for $i > (1+8/\delta^{'})\nu_{t_k}$, we have $i-\nu_{t_k} \geq 1 $ and  
$i-\nu_{t_k} >8i/(8+\delta^{'}) $.
We then have
\begin{eqnarray*}
 &  & \frac{1-\pi_k(\gamma_{t_k} | Y)}{\pi_k(\gamma_{t_k} | Y)} \\
& \leq &  8q^{-(1+\delta)} \left[ \sum_{i=1}^{M_n - \nu_{t_k}} a(i) + \sum_{i= (1+8/\delta^{'})\nu{t_k} }^{M_n} a(i)  \right]  
          + 16q^{-(1+\delta)} p^{- \delta (1+ 8/\delta{'}) \nu_{t_k} } 	\\
& & \qquad \qquad \qquad \qquad \left(\mbox{where } \, a(i) = p^i p^{- (1 + \delta )8i/(8+\delta^{'})}\right) \\
& \leq &	16q^{-(1+\delta)} \left[ 1 +  \sum_{i=1}^{\infty} a(i) \right] \\
& = & 16q^{-(1+\delta)} \left[ 1 + \frac{1}{p^{(1+\delta)8 / (8+\delta^{'}) -1} -1} \right]
\end{eqnarray*}
and therefore,
$$\pi_k (\gamma_{t_k} | Y) \geq \left[ 1 + 16q^{-(1+\delta)} \left( 1 +  \frac{1}{p^{(1+\delta)8 / (8+\delta^{'}) -1} -1} \right) \right]^{-1}.$$
Now
\begin{align*}
\pi (\gamma_t |Y) & = \prod_{k=1}^q \pi_k (\gamma_{t_k} | Y) \\
 & \geq \left[ 1 + 16q^{-(1+\delta)} \left( 1 +  \frac{1}{p^{(1+\delta)8 / (8+\delta^{'}) -1} -1} \right) \right]^{-q} \\
 & \geq \exp \left[ -q 16q^{-(1+\delta)} \left( 1 +  \frac{1}{p^{(1+\delta)8 / (8+\delta^{'}) -1} -1} \right)  \right] \\
 & = \exp \left[ - 16q^{-\delta} \left( 1 +  \frac{1}{p^{(1+\delta)8 / (8+\delta^{'}) -1} -1} \right)  \right] \\
  &\to 1 \qquad \text{ as } p,q \to \infty.
\end{align*}

As $\mathbb{P}_0 (G_n) \to 1$, we get $\pi (\gamma_t |Y) \overset{\mathbb{P}_0}{\longrightarrow} 1$ as $n \to \infty$ .

Now we recall that $$\hat{\pi}_{jk} = P(b_{jk} \neq 0| Y) = \pi(\gamma_{jk} = 1| Y)$$ and $(\hat{\gamma}_{stepwise})_{jk}$ is defined as 1 if $\hat{\pi}_{jk} \geq 1/2$ and $0$ if ${\hat{\pi}}_{jk} <1/2.$ Let $E^{0} = \{(j,k) : (\gamma_t)_{jk} = 1\}$. \\
For $(j,k) \in E^0, \gamma = \gamma_t \implies \gamma_{jk} = 1$.
Thus for $(j,k) \in E^0,  \pi(\gamma_t|Y) \leq \pi(\gamma_{jk} = 1 | Y) = \hat{\pi}_{jk}$.\\
For $(j,k) \not \in E^0, \gamma = \gamma_t \implies \gamma_{jk} = 0$.
Thus for $(j,k) \not\in E^0,  \pi(\gamma_t|Y) \leq \pi(\gamma_{jk} = 0 | Y) = 1-\hat{\pi}_{jk}$. Then
\begin{align}\label{eq:therorem1a}
    \mathbb{P}_0 (\hat{\gamma}_{stepwise}=\gamma_t) &= \mathbb{P}_0((\hat{\gamma}_{stepwise})_{jk} = (\gamma_t)_{jk} \; \forall \; (j,k) ) \nonumber \\
    &= \mathbb{P}_0\left((\hat{\gamma}_{stepwise})_{jk} = 1 \;\forall \;(j,k)\in E^0 \text{ and } (\hat{\gamma}_{stepwise})_{jk} = 0 \; \forall \;(j,k)\not\in E^0\right) \nonumber\\
    &= \mathbb{P}_0\left(\hat{\pi}_{jk}\geq \frac{1}{2} \; \forall \; (j,k)\in E^0 \text{ and } \hat{\pi}_{jk} < \frac{1}{2} \; \forall \; (j,k)\not\in E^0\right) \nonumber\\
    &\geq \mathbb{P}_0\left(\pi(\gamma_t|Y) > \frac{1}{2}\right) \to 1 \qquad \text{as } n \to \infty.
\end{align}

\subsection{Proof of Theorem 1(b)}
\label{suppsec:theorem1b}

Next we prove the theorem on estimation consistency for $B$ where we show that there exists a constant $K >0$ such that $$\mathbb{E}_0 \left(\Pi_n\left\{\norm{B - B_0}_F > K\sqrt{\frac{\delta_n \log (pq)}{n}} \;|\; Y\right\}\right) \to 0 \text{ as } n\to \infty$$ where $\delta_n = \sum_{k = 1}^q \nu_{t_k}$ and $\Pi_n$ denotes the posterior distribution. 

For each $k\; (k= 1,\ldots,q)$, let $ \tilde{B}_{.k}$  denote the vector of dimension $\nu_{t_k}$ consisting of the non-zero entries of $B_{.k}$ given the true sparsity pattern $\gamma_t$ and $ (\tilde{B_0})_{.k}$ or simply $ \tilde{B}_{0,.k}$ denote the vector consisting of the non-zero entries of $(B_0)_{.k}$, the $k^{th}$ column of $B_0$. Let $B^{*}$ be a $p \times q$ matrix whose $k$-th column, $B_{.k}^{*}$ is given by the posterior mean 
$$ E(B_{.k}|\gamma_{t_k},Y).$$ Let $\tilde{B}^{*}_{.k}$ be the vector of dimension $\nu_{t_k}$ consisting of the non-zero entries of $B_{.k}^{*}$. Then it can be shown that $$\tilde{B}^{*}_{.k}  =  (X^{'}_{{t_k}}X_{{t_k}} + \frac{1}{\tau_1^2} I_{\nu_{t_k}})^{-1}X^{'}_{{t_k}}y_{.k}.$$ 

First we note that for any $\epsilon > 0$,

\begin{align}
   &\mathbb{E}_0 \left(\Pi_n\left\{\norm{B - B_0}_F > K\epsilon \;|\; Y\right\}\right)\nonumber\\
   & \leq \mathbb{E}_0 \left(\Pi_n\left\{\norm{B - B_0}_F > K\epsilon \;|\; Y, \gamma_t\right\}\right)  + \mathbb{E}_0{\Pi_n(\gamma \neq \gamma_t \;|\; Y)}
\end{align}

Thus it is enough to show that $$\mathbb{E}_0 \left(\Pi_n\left\{\norm{B - B_0}_F > K \sqrt{\frac{ \delta_n \log (pq)}{n}} \;|\; Y, \gamma_t\right\}\right) \to 0 \text{ as } n \to \infty$$ since it is proved earlier that $\Pi_n(\gamma \neq \gamma_t| Y) \overset{\mathbb{P}_0}{\longrightarrow} 0$. Also it can be shown that

\begin{align} \label{E0first}
    &\mathbb{E}_0 \left(\Pi_n\left\{\norm{B - B_0}_F > K \sqrt{\frac{ \delta_n \log (pq)}{n}} \;|\; Y, \gamma_t\right\}\right)\nonumber\\ 
    & \leq q \max_{1\leq k \leq q}\mathbb{E}_0 \left(\Pi_n\left\{\norm{\tilde{B}_{.k} - \tilde{B}_{0,.k}}_2 \geq K \epsilon_{n,k} \;|\; Y, \gamma_t\right\}\right)
\end{align} where $ \epsilon_{n,k} = \sqrt{\frac{\nu_{t_k}\log (pq)}{n}}$ and the maximum is over those $k$ for which $\nu_{t_k} \geq 1$. Further,

\begin{align}\label{Ec3}
 &\mathbb{E}_0 \left(\Pi_n\left\{\norm{\tilde{B}_{.k} - \tilde{B}_{0,.k}}_2 \geq K\epsilon_{n,k} \;|\; Y, \gamma_t\right\}\right)\nonumber\\
 & \leq \mathbb{E}_0 \left(\Pi_n\left\{\norm{\tilde{B}_{.k} - \tilde{B}^{*}_{.k}}_2 \geq \frac{K}{2} \epsilon_{n,k} \;|\; Y, \gamma_t)\right\}\right) + \mathbb{P}_0 \left(\norm{\tilde{B}^{*}_{.k} - \tilde{B}_{0,.k}}_2 \geq \frac{K}{2} \epsilon_{n,k} \right).
\end{align}

The posterior distribution of $\tilde{B}_{.k}$ and $\sigma^2_k$ are given by 

    $$\tilde{B}_{.k} | \gamma_t, \sigma^2_k, Y \overset{ind}{\sim} \mathcal{N}_{\nu_{t_k}}\left(\tilde{B}^{*}_{.k}, \sigma^2_{k} \left(X_{t_k}' X_{t_k} + \frac{1}{\tau^2_1}I_{\nu_{t_k}}\right)^{-1} \right)$$

$$\sigma^2_k | \gamma_t, Y \overset{ind}{\sim} \text{ Inv-Gamma } \left(\frac{n}{2} + \alpha, \frac{y_{.k}'(I - \tilde{P}_{t_k})y_{.k} + 2\beta}{2}\right)$$

Now \begin{align}\label{ecnorm}
    \norm{\tilde{B}_{.k} - \tilde{B}^{*}_{.k}}_2 &= \norm{\sigma_k(X_{t_k}'X_{t_k} + \frac{1}{\tau_1^2}I_{\nu_{t_k}})^{-1/2}(X_{t_k}'X_{t_k} + \frac{1}{\tau_1^2}I_{\nu_{t_k}})^{1/2} \frac{\tilde{B}_{.k} - \tilde{B}^{*}_{.k}}{\sigma_k}}_2\nonumber\\
    & = \sigma_k\norm{ \left(X_{t_k}'X_{t_k} + \frac{1}{\tau_1^2}I_{\nu_{t_k}}\right)^{-1/2} z}_2 \nonumber\\
    & \leq \sigma_k \lambda_{max} \left(X_{t_k}'X_{t_k} + \frac{1}{\tau_1^2}I_{\nu_{t_k}}\right)^{-1/2} \norm{z}_2 
\end{align}
where $z$ is a $\nu_{t_k} \times 1$ standard normal vector.

Now for any $M^{*} > 0$, using (\ref{ecnorm}) the first quantity in the RHS of (\ref{Ec3}) can be written as 

\begin{align}\label{Ec4}
    &\mathbb{E}_0 \left(\Pi_n\left\{\norm{\tilde{B}_{.k} - \tilde{B}^{*}_{.k}}_2 \geq \frac{K}{2}\epsilon_{n,k} \;|\; Y, \gamma_t\right\}\right)\nonumber\\
    & \leq \mathbb{P}_0\left(\lambda_{min}\left(\frac{X_{t_k}'X_{t_k}}{n}\right) < \lambda_1/2\right) + \mathbb{P}_0\left(\norm{z}_2 \geq \sqrt{n}\epsilon_{n,k}K \frac{\sqrt{\lambda_1/2}}{2M^{*}}\right) + \mathbb{E}_0 \Pi_n(\sigma_k > M^{*} | Y, \gamma_t)
\end{align}

We set 
\begin{align*}
    G^{*}_{1,n} := {\bigcap}_{k = 1}^{q}\left\{\norm{\frac{X_{t_k}^{'}X_{t_k}}{n}-R_{t_k}}_{2}\leq 2c_1\sqrt{\frac{\nu_{t_k}\log (pq)}{cn}}\right\}.
\end{align*}
Similar to the set $G_{1,n}$ defined above, it can be shown that $\mathbb{P}_0(G^{*}_{1,n}) \geq 1- \frac{2}{q^2(p^2-1)} $ for some suitably chosen constants $c_1$ and $c$. Then for a fixed $k$ on the set $G^{*}_{1,n}$,

\begin{align}
    &\quad\norm{\frac{X_{t_k}^{'}X_{t_k}}{n}-R_{t_k}}_{2} \leq 2c_1\sqrt{\frac{\nu_{t_k}\log (pq)}{cn}}\nonumber\\
    &  \implies \lambda_{min}\left(\frac{X_{t_k}'X_{t_k}}{n}\right) > \lambda_{min}(R_{t_k}) - 2c_1\sqrt{\frac{\nu_{t_k}\log (pq)}{cn}} > \lambda_1/2\nonumber
\end{align} and hence 

\begin{align}
    \mathbb{P}_0\left(\lambda_{min}\left(\frac{X_{t_k}'X_{t_k}}{n}\right) < \lambda_1/2\right) \leq 1-\mathbb{P}_0(G^{*}_{1,n}) \leq \frac{2}{q^2(p^2-1)}.\nonumber
\end{align}

To bound the third term in the RHS of (\ref{Ec4}) we recall the distribution of $\sigma^2_k \;|\; Y, \gamma_t$ and use a slight modification of Remark S1.1 of \cite{ghoshstrong} with $\log p$ replaced by $\log (pq)$ wherein we show both the shape and scale of the Inverse gamma distribution are of appropriate order.

We set $G^{(k)}_{2,n}:= \{\varepsilon_{.k}^{'}P_{{t_k}}\varepsilon_{.k} \leq 8\sigma_{k0}^2\nu_{t_k} \log{(pq)}\}$ and note that $\mathbb{P}_0(G^{(k)}_{2,n}) \geq 1- 2(pq)^{-3/2} $. Now on $G^{(k)}_{2,n} \cap G_{3,n}$, by arguments used in proving equation (\ref{lemma2.num2})

\begin{align}
    \frac{y_{.k}'(I - \tilde{P}_{t_k})y_{.k} + \beta}{2n} &= \frac{\epsilon_{.k}'(I - \tilde{P}_{t_k})\epsilon_{.k}}{2n} +  \frac{y_{.k}'(P_{t_k} - \tilde{P}_{t_k})y_{.k}}{2n} + \frac{\beta}{2n}\nonumber\\
    & \leq \frac{\epsilon_{.k}'\epsilon_{.k}}{2n} + o(1) \nonumber\\
    & \leq \frac{3(1+\delta^{'}\sigma_{max}^2)}{2} \nonumber
\end{align} and also  $\frac{n}{2} + \alpha \sim n$.
Then by choosing $M^{*}$ properly we can make $\mathbb{E}_0 \Pi_n(\sigma_k > M^{*} | Y, \gamma_t) < (pq)^{-2}$ for all large $n$.

Using Corollary 5.35 in \cite{vershynin2018highdim} one can show that
\begin{align}
    \mathbb{P}_0\left(\norm{z}_2 \geq \sqrt{\nu_{t_k}} + t +1\right)\leq 2e^{-t^2/2} \text{ for all } t>0.\nonumber
\end{align} Then setting $t = 2\sqrt{\log (pq)}$ and choosing $K$ such that $ \frac{K\sqrt{\lambda_1}}{4\sqrt{2}M^{*}} > 2$ one can show that 

\begin{align}
    \mathbb{P}_0\left(\norm{z}_2 \geq \sqrt{n}\epsilon_{n,k}K \frac{\sqrt{\lambda_1/2}}{2M^{*}}\right)\leq 2(pq)^{-2}
\end{align}

Thus  we get 

\begin{align}\label{ecfirstpart}
     &\mathbb{E}_0 \left(\Pi_n\left\{\norm{\tilde{B}_{.k} - \tilde{B}^{*}_{.k}}_2 \geq \frac{K}{2} \epsilon_{n,k} \;|\; Y, \gamma_t)\right\}\right)\nonumber\\ 
    & \leq  \frac{5}{q^2(p^2-1)}
\end{align}

Finally we obtain an upper bound for the second term, $\mathbb{P}_0 \left(\norm{\tilde{B}^{*}_{.k} - \tilde{B}_{0,.k}}_2 \geq \frac{K}{2} \epsilon \;|\; Y, \gamma_t\right)$ in (\ref{Ec3}). To that end we first note for each $k$

\begin{align}\label{proof1b.1}
    \norm{{\tilde{B}}^{*}_{.k} - {\tilde{B}}_{0,.k}}_2 & =\norm{(X^{'}_{{t_k}}X_{{t_k}} + \frac{1}{\tau_1^2} I_{\nu_{t_k}})^{-1}X^{'}_{{t_k}}(X_{{t_k}}{\tilde{B}}_{0,.k} + \varepsilon_{.k}) - {\tilde{B}}_{0,.k}}_2 \nonumber\\
   & \leq  {\norm{\Big\{\left(X^{'}_{{t_k}}X_{{t_k}} + \frac{1}{\tau_1^2} I_{\nu_{t_k}}\right)^{-1}X^{'}_{{t_k}}X_{{t_k}} - I_{\nu_{t_k}}\Big\}{\tilde{B}_{0,.k}}}}_2 \nonumber
   \\ &\qquad + {\norm{\left(X^{'}_{{t_k}}X_{{t_k}} + \frac{1}{\tau_1^2} I_{\nu_{t_k}}\right)^{-1}X^{'}_{{t_k}} \varepsilon_{.k}}}_2
   \end{align}
 
 Now, on $G^{*}_{1,n}$ using Assumption \ref{assumpA:tau1} we have
 
 \begin{align} \label{proof1b.2}
     &\quad {\norm{\Big\{\left(X^{'}_{{t_k}}X_{{t_k}} + \frac{1}{\tau_1^2} I_{\nu_{t_k}}\right)^{-1}X^{'}_{{t_k}}X_{{t_k}} - I_{\nu_{t_k}}\Big\}{\tilde{B}_{0,.k}}}}_2 \nonumber\\
     & =  {\norm{\left( \tau_1^2 X^{'}_{{t_k}}X_{{t_k}} + I_{\nu_{t_k}} \right)^{-1}{\tilde{B}}_{0,.k}}}_2\nonumber\\
      &=  \norm{\left( X^{'}_{{t_k}}X_{{t_k}} + \frac{1}{\tau_1^2}I_{\nu_{t_k}} \right)^{-1}\frac{1}{\tau_1^2}{\tilde{B}}_{0,.k}}_2\nonumber\\
       & \leq \norm{\left( X^{'}_{{t_k}}X_{{t_k}} + \frac{1}{\tau_1^2}I_{\nu_{t_k}} \right)^{-1}}_2 \norm{\frac{1}{\tau_1^2}{\tilde{B}}_{0,.k}}_2 \nonumber \\
       & \leq  \frac{2}{\lambda_1n\tau_1^2} \max_{1\leq k \leq q}\norm{{\tilde{B}}_{0,.k}}_2 \nonumber \\
      & = \frac{2}{\lambda_1\sqrt{n\tau_1^2}} \frac{\max_{1\leq k \leq q}\norm{{\tilde{B}}_{0,.k}}_2}{\sqrt{\tau_1^2\log(pq)}}
      \sqrt{\frac{\log(pq)}{n} } \nonumber\\
      & = \sqrt{\frac{ \log(pq)}{n}} o(1).
 \end{align}  and on $G^{*}_{1,n} \cap G^{(k)}_{2,n}$ we have 
 \begin{align} \label{proof1b.3}
     &\quad {\norm{\left(X^{'}_{{t_k}}X_{{t_k}} + \frac{1}{\tau_1^2} I\right)^{-1}X^{'}_{{t_k}} \varepsilon_{.k}}}_2 \nonumber\\
     & = \norm{\left( X^{'}_{{t_k}}X_{{t_k}} + \frac{1}{\tau_1^2} I_{{t_k}}\right)^{-1}}_2 \norm {X^{'}_{{t_k}}\varepsilon_{.k}}_2 \nonumber\\
     & \leq \frac{2}{\lambda_1n}  \norm {X^{'}_{{t_k}}\varepsilon_{.k}}_2 \nonumber\\
     & \leq \frac{2}{\lambda_1n} \sqrt{n}  \norm{\left(\frac{1}{n} X^{'}_{{t_k}}X_{{t_k}} \right)^{1/2}}_2 \sqrt{\varepsilon_{.k}^{'}P_{{t_k}}\varepsilon_{.k}} \nonumber\\
     & \leq \left( \frac{3\sqrt{\lambda_2}}{\lambda_1 \sqrt{n}}\right)\sqrt{8\sigma_{k0}^2\nu_{t_k} \log{(pq)}}\nonumber\\
     & \leq  \frac{3\sigma_{max}\sqrt{\lambda_2}}{\lambda_1} \sqrt{\frac{8\nu_{t_k}\log(pq)}{n}}.
    \end{align}

From (\ref{proof1b.1}) - (\ref{proof1b.3}) we have on the set $G^{*}_{1,n}\cap G^{(k)}_{2,n} $,
\begin{align} \label{B*bound}
    \norm{{\tilde{B}}^{*}_{.k} - {\tilde{B}}_{0,.k}}_2 &\leq  \sqrt{\frac{\log(pq)}{n}}\; o(1) +\frac{3\sigma_{max}\sqrt{\lambda_2}}{\lambda_1}\sqrt{\frac{8\nu_{t_k}\log(pq)}{n}} \nonumber\\
    & \leq \frac{12\sigma_{max}\sqrt{\lambda_2}}{\lambda_1}\sqrt{\frac{\nu_{t_k}\log(pq)}{n}} .
\end{align}

Thus, 
\begin{align}\label{ecsecondpart}
    \mathbb{P}_0\left(\norm{{\tilde{B}}^{*}_{.k} - {\tilde{B}}_{0,.k}}_2 \geq \frac{12\sigma_{max}\sqrt{\lambda_2}}{\lambda_1}\epsilon_{n,k}\right) & \leq \frac{2}{q^2(p^2-1)} + \frac{2}{p^{3/2}q^{3/2}}
\end{align} and hence from (\ref{E0first}), (\ref{ecfirstpart}) and (\ref{ecsecondpart}), choosing $K > \frac{24\sigma_{max}\sqrt{\lambda_2}}{\lambda_1} $ we get

\begin{align} 
    &\mathbb{E}_0 \left(\Pi_n\left\{\norm{B - B_0}_F > K\sqrt{\frac{ \delta_n \log (pq)}{n}} \;|\; Y, \gamma_t\right\}\right)\nonumber\\ 
    & \leq q \left(\frac{5}{q^2(p^2 -1)} + \frac{2}{q^2(p^2-1)} + \frac{2}{p^{3/2}q^{3/2}} \right) \to 0 \text{ as } n \to \infty
\end{align}

\section{Details of proof of Theorem 1(c)}
\label{suppsec:theorem1b}
\subsection{Assumptions required for Theorem 1(c)}

\begin{assumpB}
$k_n\sqrt{\frac{d_t\log(pq)}{n}}\to 0$ \; as $n\to\infty$ \\where $d_t$ denotes the number of non-zero entries in the upper triangle of $\Omega_0$.
\end{assumpB}

\begin{assumpB}
There exists $\tilde{\varepsilon}_0>0$ such that 
$$ \tilde{\varepsilon}_0 \leq eig_{min}(\Omega_0) \leq eig_{max}(\Omega_0) \leq \frac{1}{\tilde{\varepsilon}_0}. $$
\end{assumpB}

\begin{assumpB}
$\frac{\log n +d_t^2k_n\log(pq)}{n\rho_n^2}\to 0$,  \; as $n \to \infty$ \\where $\rho_n$ denotes the smallest absolute value among all the off-diagonal entries of $\Omega_0$.
\end{assumpB}

\begin{assumpB}
We choose $q_2 = (pq)^{-a_2d_t^2k_n^2}$ where $a_2 = 16\frac{\max(1,c_0)}{\min(1,\tilde{\varepsilon}_0)}$ \\where $c_0 > 0$ is an appropriately chosen constant.
\end{assumpB}

As in \cite{KOR:2015} and \cite{peng2009partial} we assume the existence of accurate estimates $\hat{\omega}_{jj}$ of the diagonal elements $\omega_{jj}, \; j = 1,...,q$ and some constant $C>0$ such that
$$ \max\limits_{1\leq j\leq q}|\hat{\omega}_{jj}-\omega_{jj}| \leq Ck_n\sqrt{\frac{\log(pq)}{n}}.$$

\subsection{Proof of Theorem 1(c)}
 A Bayesian approach has been developed in \cite{jalali2020bconcord} for sparse estimation of the error precision matrix $\Omega$ in a high-dimensional setting using spike and slab priors and the regression based generalized likelihood function of \cite{KOR:2015}. Their results are applicable in our case if the regression coefficient matrix $B$ is assumed to known. The expression for the posterior distribution of the sparsity pattern of $\Omega$, as obtained in \cite{jalali2020bconcord} depends on the estimate $S = \frac{1}{n}(Y-XB)^T(Y-XB)$ of the variance-covariance matrix $\Omega^{-1}$. However, in our case $B$ is unknown and so we estimate it by $\hat{B}$ as obtained in Step 1 of the Stepwise method and replace $S$ by $\hat{S} = \frac{1}{n}(Y-X\hat{B})^T(Y-X\hat{B})$. Below we provide a bound for $\hat{B}-B$.
 
 Let $B_0$ denote the true value of $B$ and $ \tilde{B}_{0,.k}$, $B_{0,.k}$ and $B^{*}$ be as defined above in Section \ref{suppsec:theorem1a}. 
Our final estimate $\hat{B}$ of $B$ is obtained from $B^{*}$, replacing $\gamma_t$ by its estimate $\hat{\gamma}_{stepwise}$.

Now,
\begin{align} \label{proof2.1}
    {\norm{B^{*}- B_0}}_1 &= \max_{1\leq k \leq q} {\norm{B^{*}_{.k} - B_{0,.k}}}_1 \nonumber\\
   &= \max_{1\leq k \leq q} {\norm{{\tilde{B}}^{*}_{.k} - {\tilde{B}}_{0,.k}}}_1 \nonumber\\
   &= \max_{1\leq k \leq q} {\norm{(X^{'}_{{t_k}}X_{{t_k}} + \frac{1}{\tau_1^2} I_{\nu_{t_k}})^{-1}X^{'}_{{t_k}}(X_{{t_k}}{\tilde{B}}_{0,.k} + \varepsilon_{.k}) - {\tilde{B}}_{0,.k}}}_1 \nonumber\\
   &\leq \max_{1\leq k \leq q} {\norm{\Big\{\left(X^{'}_{{t_k}}X_{{t_k}} + \frac{1}{\tau_1^2} I_{\nu_{t_k}}\right)^{-1}X^{'}_{{t_k}}X_{{t_k}} - I_{\nu_{t_k}}\Big\}{\tilde{B}_{0,.k}}}}_1 \nonumber
   \\ &\qquad + \max_{1\leq k \leq q} {\norm{\left(X^{'}_{{t_k}}X_{{t_k}} + \frac{1}{\tau_1^2} I_{\nu_{t_k}}\right)^{-1}X^{'}_{{t_k}} \varepsilon_{.k}}}_1
   \end{align}
 
and on $G^{*}_{1,n}$ we have
 
 \begin{align} \label{proof2.2}
     &\quad\max_{1\leq k \leq q} {\norm{\Big\{\left(X^{'}_{{t_k}}X_{{t_k}} + \frac{1}{\tau_1^2} I_{\nu_{t_k}}\right)^{-1}X^{'}_{{t_k}}X_{{t_k}} - I_{\nu_{t_k}}\Big\}{\tilde{B}_{0,.k}}}}_1 \nonumber\\
     & = \max_{1\leq k \leq q} \norm{\left( \tau_1^2 X^{'}_{{t_k}}X_{{t_k}} + I_{\nu_{t_k}} \right)^{-1}{\tilde{B}}_{0,.k}}\nonumber\\
      &= \max_{1\leq k \leq q} \norm{\left( X^{'}_{{t_k}}X_{{t_k}} + \frac{1}{\tau_1^2}I_{\nu_{t_k}} \right)^{-1}\frac{1}{\tau_1^2}{\tilde{B}}_{0,.k}}_1\nonumber\\
       &\leq \max_{1\leq k \leq q} \sqrt{k_n} \norm{\left( X^{'}_{{t_k}}X_{{t_k}} + \frac{1}{\tau_1^2}I_{\nu_{t_k}} \right)^{-1}\frac{1}{\tau_1^2}{\tilde{B}}_{0,.k}}_2\nonumber\\
       & \leq \sqrt{k_n} \max_{1\leq k \leq q}\norm{\left( X^{'}_{{t_k}}X_{{t_k}} + \frac{1}{\tau_1^2}I_{\nu_{t_k}} \right)^{-1}}_2 \norm{\frac{1}{\tau_1^2}{\tilde{B}}_{0,.k}}_2 \nonumber \\
       & \leq \sqrt{k_n} \frac{2}{\lambda_1n\tau_1^2} \max_{1\leq k \leq q}\norm{{\tilde{B}}_{0,.k}}_2 \nonumber \\
      & = \sqrt{k_n} \frac{2}{\lambda_1\sqrt{n\tau_1^2}} \frac{\max_{1\leq k \leq q}\norm{{\tilde{B}}_{0,.k}}_2}{\sqrt{\tau_1^2\log(pq)}}
      \sqrt{\frac{\log(pq)}{n} } \nonumber\\
      & = \sqrt{\frac{k_n \log(pq)}{n}} o(1) 
 \end{align}  
by Assumption \ref{assumpA:tau1} and on $G_{1,n}\cap G_{2,n}$
 \begin{align} \label{proof2.3}
     &\quad\max_{1\leq k \leq q} {\norm{\left(X^{'}_{{t_k}}X_{{t_k}} + \frac{1}{\tau_1^2} I\right)^{-1}X^{'}_{{t_k}} \varepsilon_{.k}}}_1 \nonumber\\
     & = \sqrt{k_n} \max_{1\leq k \leq q}\norm{\left( X^{'}_{{t_k}}X_{{t_k}} + \frac{1}{\tau_1^2} I_{{t_k}}\right)^{-1}}_2 \norm {X^{'}_{{t_k}}\varepsilon_{.k}}_2 \nonumber\\
     & \leq \sqrt{k_n}\frac{2}{\lambda_1n} \max_{1\leq k \leq q} \norm {X^{'}_{{t_k}}\varepsilon_{.k}}_2 \nonumber\\
     & \leq \sqrt{k_n}\frac{2}{\lambda_1n} \sqrt{n} \max_{1\leq k \leq q} \norm{\left(\frac{1}{n} X^{'}_{{t_k}}X_{{t_k}} \right)^{1/2}}_2 \sqrt{\varepsilon_{.k}^{'}P_{{t_k}}\varepsilon_{.k}} \nonumber\\
     & \leq \sqrt{k_n}\left( \frac{3\sqrt{\lambda_2}}{\lambda_1 \sqrt{n}}\right)\sqrt{8\sigma_{k0}^2\nu_{t_k} \log{(pq)}}\nonumber\\
     & \leq  \frac{3\sigma_{max}\sqrt{\lambda_2}}{\lambda_1} \sqrt{\frac{8k_n^2\log(pq)}{n}}.
    \end{align}
From (\ref{proof2.1}) - (\ref{proof2.3}) we have on the set $G_n$,
\begin{align} \label{B*bound}
    {\norm{B^{*}- B_0}}_1 &\leq  \sqrt{\frac{k_n\log(pq)}{n}}\; o(1) +\frac{3\sigma_{max}\sqrt{\lambda_2}}{\lambda_1}\sqrt{\frac{8k_n^2\log(pq)}{n}} \nonumber\\
    & = k_n\sqrt{\frac{\log(pq)}{n}} \; O(1).
\end{align}

Since $\hat{B}$ is obtained from $B^{*}$, replacing $\gamma_t$ by its estimate $\hat{\gamma}_{stepwise}$, it follows from Theorem 1(a) and (\ref{B*bound}) that with $\mathbb{P}_0$-probability tending to one, 

\begin{align}\label{eq:Bbound}
    {\norm{\hat{B}- B_0}}_1 \leq k_n\sqrt{\frac{\log(pq)}{n}} \; O(1)
\end{align}

 Next we provide a bound for $\hat{S}-S$ using the bound provided in (\ref{eq:Bbound}).
We note that $\hat{S} - S$ can be decomposed as follows:
\begin{align}\label{diff:S}
    \hat{S} - S = \frac{2}{n} \varepsilon ^{'} X(B_0 - \hat{B}) + (B_0 - \hat{B})^{'}\left(\frac{X^{'}X}{n}\right)(B_0 - \hat{B}).
\end{align}
We have
\begin{align}
    \norm{\frac{1}{n} X^{'}\varepsilon}_{max} &= \frac{1}{n} \max_{i,j} \left| X_{.i}^{'}\varepsilon_{.j}\right| \nonumber
\end{align}
and for each $(i,j)$ on $G_{1,n}\cap G_{3,n}$,
\begin{align}
    \frac{1}{n} \left| X_{.i}^{'}\varepsilon_{.j}\right| &\leq \frac{1}{n} \norm{X_{.i}}_2\norm{\varepsilon_{.j}}_2\nonumber\\
    &\leq \sqrt{\left(R_{ii} + 32\lambda_2 \sqrt{k_0}\right)} \sqrt{(1+\delta^{'})\sigma_{max}^2} \nonumber\\
    &\leq \left(\sqrt{\lambda_{2}} + \sqrt{32\lambda_2}\;k_0^{1/4}\right)\sqrt{(1+\delta^{'})\sigma_{max}^2} \;. \nonumber
\end{align}
Hence, $$ \norm{\frac{1}{n} X^{'}\varepsilon}_{max} \leq (\sqrt{\lambda_{2}} +\sqrt{32\lambda_2}\;k_0^{1/4})\sqrt{(1+\delta^{'})\sigma_{max}^2}.$$
We also have
\begin{align}
    \norm{\frac{X^{'}X}{n} }_{max} &= \frac{1}{n} \max_{i,j} \left| X_{.i}^{'} X_{.j}\right| \nonumber
\end{align}
and for each $(i,j)$,
\begin{align*}
     \frac{1}{n} \left| X_{.i}^{'} X_{.j}\right| &\leq \frac{1}{n} \norm{X_{.i}}_2\norm{X_{.j}}_2.
\end{align*}
For each $i$ on $G_{1,n}$,
\begin{align}
    \frac{1}{n} \norm{X_{.i}}^2 & \leq \left(R_{ii} + 32\lambda_2\sqrt{k_0}\right)\nonumber\\
    &\leq \left(\lambda_2 + 32\lambda_2\sqrt{k_0}\right).\nonumber
\end{align}
So, for each $(i,j)$,
\begin{align}
    \frac{1}{n} \left| X_{.i}^{'} X_{.j}\right| &\leq \lambda_2 + 32\lambda_2\sqrt{k_0}\nonumber
\end{align} and hence,
$$ \norm{\frac{X^{'}X}{n} }_{max} \leq \lambda_2 + 32\lambda_2\sqrt{k_0}.$$
Now from equation (\ref{diff:S}), on $G_{1,n}\cap G_{2,n}$,
\begin{align}
   \norm{ \hat{S} - S}_{max} &\leq 2 \norm{\hat{B} - B_0}_1 \norm{\frac{1}{n} X^{'}\varepsilon}_{max} + 2 \norm{\hat{B} - B_0}_1^2 \norm{\frac{1}{n} X^{'}X}_{max}\nonumber\\
   &\leq 2 \norm{\hat{B} - B_0}_1 \left(\sqrt{\lambda_{2}} +\sqrt{32\lambda_2}\;k_0^{1/4}\right)\sqrt{(1+\delta^{'})\sigma_{max}^2} + 2 \norm{\hat{B} - B_0}_1^2 \left(\lambda_2 + 32\lambda_2\sqrt{k_0}\right)\nonumber\\
   &\leq c  \sqrt{\frac{k_n^2\log{(pq)}}{n}} 
\end{align}
for some appropriate constant c.

 Using the fact that $\hat{B}$ and $\hat{S}$ are good approximations of $B$ and $S$ respectively and using straightforward modifications of the arguments of \cite{jalali2020bconcord} and some additional arguments we show that the posterior distribution of $\eta$ is consistent in the sense that $$\pi(\eta_t |\hat{B},\hat{\omega}_{11},\ldots,\hat{\omega}_{qq}, Y) \to 1 \quad \text{ as } n \to \infty.$$ Finally the part (c) of Theorem 1 follows from this result using arguments similar to that leading to (\ref{eq:therorem1a}) above in the proof of part (a). 

\section{Detailed algorithm of JRNS}
\label{suppsec:detailedAlgo1}

    \scalebox{0.65}{
\begin{minipage}{1.5\textwidth}
\begin{algorithm}[H]
	\caption{Joint Regression Network Selector} 
	\begin{algorithmic}[1]
	\Procedure{JRNS($B,\Omega,X,Y$)}{}
	   \State $M_1 \gets X^TY\Omega^2$
	   \State $M_2 = B^TX^TX$
		\For {$r=1,2,\ldots,p$}   \Comment{updating matrix, $B$}
			\For {$s=1,2,\ldots,q$}
			    \If{$b_{rs}=0$}
		        \State $\eta = \frac{1}{\tau_1^2} \gets Gamma(10^{-4},10^{-8})$
		        \Else
		        \State $\eta = \frac{1}{\tau_1^2} \gets Gamma(10^{-4}+0.5,10^{-8}+0.5B_{rs}^2)$
		        \EndIf
				\State $C_1\gets \frac{1}{\tau_1^2 }+ (\Omega^2)_{ss}(X^TX)_{rr}$  
				 \State $C_2 \gets M_{1,rs} -(M_{2,.r})^{T}(\Omega^2)_{.s} +b_{rs}(X^TX)_{rr}(\Omega^2)_{ss}$
				\State $P(0) \gets 1, P(1)\gets \frac{q_1}{(1-q_1)\tau_1 \sqrt{C_1}} \exp (C_2^2/2C_1)$
				\If {$P(1) \gets \infty$}
				\State $b_{rs} \gets N(\frac{C_2}{C_1},\frac{1}{C_1})$
				\Else 
				\State $P \gets {P/\text{sum}(P)}$
				\State $b_{rs} \sim P(0) \delta_0 + P(1) N(\frac{C_2}{C_1}, \frac{1}{C_1})$ \Comment{sampling from the mixture distribution}
				\EndIf
			\State $\text{update } M_{2,s.}$	
			\EndFor
			\EndFor
		\State $E = (Y-XB)$
		\State $S = E^TE$	
		\For {$s=1,2,\ldots,q-1$} \Comment{updating off-diagonals of $\Omega$}
		   \For {$t=s+1,2,\ldots,q$}
		   
	        	\If{$\omega_{st}=0$}
		        \State $\psi = \frac{1}{\tau_2^2} \gets Gamma(10^{-4},10^{-8})$
		        \Else
		        \State $\psi = \frac{1}{\tau_2^2} \gets Gamma(10^{-4}+0.5,10^{-8}+0.5\Omega_{st}^2)$
		        \EndIf
		        \State $D_1 \gets S_{ss} + S_{tt} + \psi$
		        \State $D_2 \gets \Omega_{.s}^TS_{.t}+\Omega_{.t}^TS_{.s} -D_1\omega_{st}$
		        \State $P(0) \gets 1, P(1) \gets \sqrt{\frac{\psi}{D_1}}\frac{q_2}{1-q_2}\exp{\left[\frac{n^2b^2}{2D_1}\right]}$
				\If {$P(1) \gets \infty$}
				\State $\omega_{st} \gets N\left(-\frac{D_2}{D_1},\frac{1}{D_1}\right)$
				\Else 
				\State $P \gets {P/\text{sum}(P)}$
				 \State $\omega_{st} \sim P(0) \delta_0 + P(1) N\left(-\frac{D_2}{D_1},\frac{1}{D_1}\right)$ \Comment{sampling from the mixture distribution}
				\EndIf
		    \EndFor 
		\State $\lambda \gets \text{Gamma}(r+1,\omega_{ss}+s) $ \Comment{Metropolis-within-Gibbs for updating diagonals of $\Omega$}
		\State $\text{mode} = \frac{\sqrt{ (\Omega_{.s}^TS_{.s}-\omega_{ss}S_{ss}+ \lambda/n)^2 + 4S_{ss}n}- (\Omega_{.s}^TS_{.s}-\omega_{ss}S_{ss}+ \lambda/n)}{2S_{ss}}$
		\State $v \gets N(\text{mode},0.001)$ \Comment{choosing proposed value}
		\State $\rho = min\{1, \exp[n\log(v/\omega_{ss}) -\frac{1}{2}S_{ss}(v^2 - \omega_{ss}^2) - bb(v-\omega_{ss})]\}$ \Comment{calculating acceptance probability}
		\State $\omega_{ss} \gets \text{sample}(\{v,\omega_{ss}\},1,\{\rho,(1-\rho)\})$ \Comment{choosing proposed value $v$ with probability $\rho$}
		\EndFor
		\State Repeat Steps 42 - 46 for $s = q$
		\State update $\Omega^2$
		\State \textbf{return} $B$
		\State \textbf{return} $\Omega$
	\EndProcedure	
	\end{algorithmic}
	\label{supp:detailedAlgo1}
\end{algorithm}
\end{minipage}
}

\section{Additional simulation results}
\label{suppsec:extrasim}
\subsection{Sparsity selection performance}
In the main paper we have presented the average MCC values for sparsity estimation in both $B$ and $\Omega$ based on 200 replicated datasets for all the methods, namely JRNS(Joint), Stepwise approach, DPE, DCPE and BANS. We considered a variety of combinations of $(n,p,q)$ listed in Table 1 of the main paper. Here we present tables with average values of sensitivity and specificity for sparsity estimation in both $B$ and $\Omega$. We also present average values of relative norm for sparsity estimation in $B$ which is defined as $\frac{\norm{\hat{B}-B_0}_F}{\norm{B_0}_F}$ where $\norm{.}_F$ denotes the Frobenius norm of a matrix.  

For sparsity selection in $\Omega$, we see that the JRNS and Stepwise approaches perform much better than DCPE and DPE with respect to average sensitivity values. For specificity all the methods have values very close to 1. Though BANS performs competitively with respect to specificity, in terms of sensitivity, it is clearly outperformed by the JRNS and Stepwise methods. For sparsity selection in $B$, the JRNS algorithm shows the best performance in most of the settings and the values for the different measures for DPE and DCPE are very close to the corresponding values for JRNS and Stepwise algorithms. For BANS the sensitivity values improve here but it still has sub-optimal performance with respect to relative norm of the estimate of $B$.

\begin{table}[H]
    \centering
    \caption{Sensitivity values for sparsity selection in $\Omega$ averaged over 200 replicates of $\hat{\Omega}$.`TO' is short for `Timeout' which implies that the method could not complete the required number of iterations in 4 days. `PDE' refers to an error caused by intermediate matrices not being positive definite (PD). }
    \resizebox{1.0\textwidth}{!}{%
    \begin{tabular}{*{9}{c}}
        \toprule
          Sparsity & \multicolumn{2}{c}{Cases}& Joint &  Stepwise & DPE & DCPE &  BANS & HSGHS \\
        \cmidrule(lr){2-3} 
       & $n$ & $(p,q)$ & & &  &  & &  \\ 
         \midrule
         \multirow{4}{*}{$(p/5,q/5)$}& 100 & $(30,60)$ & 0.6817 & 0.6967 & 0.2142 & 0.2233  & 0.4817 & PDE \\
         &100 & $(60,30)$ & 0.7267 & 0.7383 & 0.1550 & 0.0283 & 0.2650 &  0.7150 \\
         &150 & $(200,200)$&  0.9290 & 0.9340 & 0.7876 & 0.7780 & TO & TO\\
           &150 & $(300,300)$ & 0.9450 & 0.8753 & TO & 0.7447  & TO & TO\\
           \midrule
           \multirow{ 2}{*}{$(p/30,q/5)$} & 100 & $(200,200)$& 0.8480 & 0.8620 & 0.2924 & 0.3303 & TO & TO\\
           &200 & $(200,200)$ & 0.9780 & 0.9805 & 0.9265 & 0.9323  & TO & TO\\
        \bottomrule
    \end{tabular}
     }%
\end{table}

\begin{table}[H]
    \centering
    \caption{Specificity values for sparsity selection in $\Omega$ averaged over 200 replicates of $\hat{\Omega}$. `TO' is short for `Timeout' which implies that the method could not complete the required number of iterations in 4 days. `PDE' refers to an error caused by intermediate matrices not being positive definite (PD). }
    \resizebox{1.0\textwidth}{!}{%
    \begin{tabular}{*{9}{c}}
        \toprule
          Sparsity & \multicolumn{2}{c}{Cases}& Joint &  Stepwise & DPE & DCPE &  BANS & HS-GHS \\
        \cmidrule(lr){2-3} 
       & $n$ & $(p,q)$ & & &  &  & & \\ 
         \midrule
         \multirow{4}{*}{$(p/5,q/5)$}& 100 & $(30,60)$ & 0.9997506 & 0.9996344 & 1.0000000 & 1.0000000 & 0.9961338 & PDE \\
         &100 & $(60,30)$ & 0.9997338 & 0.9995949 & 0.9999884 & 1.0000000 & 0.9963194 & 0.9997569 \\
         &150 & $(200,200)$& 0.9999019 & 0.9998506 & 0.9999990 & 0.9999995& TO & TO\\
           &150 & $(300,300)$ & 0.9999779 & 0.9999199  & TO & 0.9998398  & TO & TO \\
           \midrule
           \multirow{ 2}{*}{$(p/30,q/5)$} & 100 & $(200,200)$& 0.9998888 & 0.9998159  & 0.9999970 & 0.9999975  & TO & TO\\
           &200 & $(200,200)$ & 0.9999059 & 0.9998667  & 0.9999972 & 0.9999977  & TO & TO \\
        \bottomrule
    \end{tabular}
     }%
\end{table}

\begin{table}[H]
    \centering
    \caption{Sensitivity values for sparsity selection in $B$ averaged over 200 replicates of $\hat{B}$. `TO' is short for `Timeout' which implies that the method could not complete the required number of iterations in 4 days. `PDE' refers to an error caused by intermediate matrices not being positive definite (PD).  }
    \resizebox{1.0\textwidth}{!}{%
    \begin{tabular}{*{9}{c}}
        \toprule
          Sparsity & \multicolumn{2}{c}{Cases}& Joint &  Stepwise & DPE & DCPE &  BANS & HS-GHS\\
        \cmidrule(lr){2-3} 
       & $n$ & $(p,q)$ & & &  &  & & \\ 
         \midrule
         \multirow{4}{*}{$(p/5,q/5)$}& 100 & $(30,60)$ & 1.0000000 & 1.0000000 & 1.0000000 & 1.0000000 & 0.9991667 & PDE \\
         &100 & $(60,30)$ & 1.0000000 & 1.0000000 & 1.0000000 & 1.0000000 & 0.9975000 & 1.000000\\
         &150 & $(200,200)$& 1.000000 & 1.000000 & 1.0000000 & 0.9999992  & TO & TO\\
           &150 & $(300,300)$ & 1.0000000 & 1.0000000 & TO & 1.0000000  & TO & TO \\
           \midrule
           \multirow{ 2}{*}{$(p/30,q/5)$} & 100 & $(200,200)$& 1.0000000 & 1.0000000 & 1.0000000 & 1.0000000  & TO & TO\\
           &200 & $(200,200)$ & 1.0000000 & 1.0000000 & 1.0000000 & 1.0000000  & TO & TO\\
        \bottomrule
    \end{tabular}
     }
\end{table}

\begin{table}[H]
    \centering
    \caption{Specificity values for sparsity selection in $B$ averaged over 200 replicates of $\hat{B}$. `TO' is short for `Timeout' which implies that the method could not complete the required number of iterations in 4 days. `PDE' refers to an error caused by intermediate matrices not being positive definite (PD).  }
    \resizebox{1.0\textwidth}{!}{%
    \begin{tabular}{*{9}{c}}
        \toprule
          Sparsity & \multicolumn{2}{c}{Cases}& Joint &  Stepwise & DPE & DCPE &  BANS & HS-GHS \\
        \cmidrule(lr){2-3} 
       & $n$ & $(p,q)$ & & &  &  & & \\ 
         \midrule
         \multirow{4}{*}{$(p/5,q/5)$}& 100 & $(30,60)$ & 0.9999972 & 1.0000000  & 1.0000000 & 1.0000000 & 0.9941304 & PDE\\
         &100 & $(60,30)$ & 1.0000000 & 1.0000000 & 1.0000000 & 1.0000000 & 0.9987360 & 0.9997791 \\
         &150 & $(200,200)$& 0.9999977 & 0.9999932 & 1.0000000 & 0.9999992  & TO & TO\\
           &150 & $(300,300)$ & 0.9999977  & 0.9996601 & TO & 0.9999030   & TO & TO \\
           \midrule
           \multirow{ 2}{*}{$(p/30,q/5)$} & 100 & $(200,200)$& 0.9999970 & 0.9999735 & 0.9999557 & 0.9999825 & TO & TO\\
           &200 & $(200,200)$ & 0.9999987 & 0.9999627 & 0.9999990 & 0.9999575    & TO & TO\\
        \bottomrule
    \end{tabular}
     }
\end{table}

\subsection{Hyperparameter selection}
\label{suppsec:hpselection}
The important issue of selection of hyperparameters $q_1,q_2,\tau_1^2,\tau_2^2$ was briefly discussed in Section 2 of the main paper. As mentioned there, the theoretical results of our paper and also those of \cite{cao2019posterior} and \cite{Narisetty:He:2014} motivated us to take $q_1 = 1/p$ and $q_2 = 1/q$. 
In order to see how sensitive our results are with respect to changes in the values of the hyperparameters around our choices, we performed simulation experiments for different choices of $q_1$ and $q_2$ for the setting where $(p,q) = (200,200)$ and the number of non-zero entries in $B$ and among the off-diagonal entries of $\Omega$ are $p/5$ and $q/5$ respectively. The values of MCC for sparsity selection of $B$ and $\Omega$ and relative estimation error of $B$ for different values of $q_1$ and $q_2$  together with our choices of $q_1 = 1/p$ and $q_2 = 1/q$ as well as different values of the sample size $n$ are presented in Tables \ref{tabOmega:MCCq1q2} - \ref{tabB:rltnormq1q2}. The results in these tables reaffirm the intuition that especially as the sample size grows, there are no significant changes in the performance of the estimators as we vary the values of the hyperparameters $q_1$ and $q_2$.

\begin{table}[H]
    \centering
    \caption{MCC values for the sparsity selection in $\Omega$ averaged over 200 replicates for JRNS and Stepwise methods for different $q_1$ and $ q_2 $. }
      \resizebox{0.8\textwidth}{!}{%
     \begin{tabular}{cccccccc}
    \toprule
    Sample size & & \multicolumn{2}{c}{$q_1 = 0.2$} & \multicolumn{2}{c}{$q_1 = 0.1$ }& \multicolumn{2}{c}{$q_1 = 1/200$ } \\
    \cmidrule(lr){3-4} \cmidrule(lr){5-6} \cmidrule(lr){7-8}
    $n$& & JRNS & Stepwise & JRNS & Stepwise & JRNS & Stepwise \\
    \midrule
     \multirow{3}{*}{50} & $q_2 = 0.2$ & 0.456611 & 0.195824 & 0.443757 & 0.191426 & 0.413763 & 0.185339 \\ 
    & $q_2 = 0.1$ & 0.522697 & 0.407968 & 0.510742 & 0.396833 & 0.462008 & 0.356453 \\ 
    & $q_2 = 1/200$ & 0.630093 & 0.626606 & 0.626373 & 0.628004 & 0.578496 & 0.605354 \\
    \midrule
     \multirow{3}{*}{100} & $q_2 = 0.2$ &  0.824727 & 0.788363 & 0.816666 & 0.801548 & 0.758186 & 0.787223 \\ 
    & $q_2 = 0.1$ & 0.861820 & 0.829216 & 0.854171 & 0.847995 & 0.809472 & 0.845547 \\ 
    & $q_2 = 1/200$ & 0.852824 & 0.844891 & 0.864933 & 0.859695 & 0.868123 & 0.871109  \\
    \midrule
     \multirow{3}{*}{150} & $q_2 = 0.2$ & 0.944332 & 0.891047 & 0.944369 & 0.913404 & 0.926108 & 0.937398 \\ 
    & $q_2 = 0.1$ & 0.936696 & 0.890740 & 0.943541 & 0.912652 & 0.942260 & 0.945645 \\ 
    & $q_2 = 1/200$ &  0.897044 & 0.891543 & 0.922895 & 0.918023 & 0.948174 & 0.948327 \\
    \midrule
     \multirow{3}{*}{200} & $q_2 = 0.2$ &  0.958262 & 0.906148 & 0.967857 & 0.929079 & 0.969988 & 0.968598 \\ 
    & $q_2 = 0.1$ &  0.944422 & 0.908759 & 0.958733 & 0.932000 & 0.971172 & 0.969177   \\ 
    & $q_2 = 1/200$ & 0.911446 & 0.908679 & 0.937298 & 0.930596 & 0.969767 & 0.969714 \\
    \bottomrule
    \end{tabular}\label{tabOmega:MCCq1q2}
    }%
    \end{table}


\begin{table}[H]
    \centering
    \caption{MCC values for the sparsity selection in $B$ averaged over 200 replicates for JRNS and Stepwise methods for different $q_1$ and $ q_2 $. }
      \resizebox{0.8\textwidth}{!}{%
     \begin{tabular}{cccccccc}
    \toprule
    Sample size & & \multicolumn{2}{c}{$q_1 = 0.2$} & \multicolumn{2}{c}{$q_1 = 0.1$ }& \multicolumn{2}{c}{$q_1 = 1/200$ } \\
    \cmidrule(lr){3-4} \cmidrule(lr){5-6} \cmidrule(lr){7-8}
    $n$& & JRNS & Stepwise & JRNS & Stepwise & JRNS & Stepwise \\
    \midrule
     \multirow{3}{*}{50} & $q_2 = 0.2$ & 0.549165 & 0.088672 & 0.542708 & 0.088634 & 0.543439 & 0.088686  \\ 
    & $q_2 = 0.1$ & 0.592646 & 0.161716 & 0.590881 & 0.161653 & 0.590786 & 0.162460  \\ 
    & $q_2 = 1/200$ & 0.713135 & 0.819988 & 0.714307 & 0.820428 & 0.715335 & 0.818886 \\
    \midrule
     \multirow{3}{*}{100} & $q_2 = 0.2$ & 0.488290 & 0.147177 & 0.487633 & 0.147078 & 0.489134 & 0.147308  \\ 
    & $q_2 = 0.1$ & 0.635307 & 0.282160 & 0.632513 & 0.281971 & 0.632951 & 0.282221 \\ 
    & $q_2 = 1/200$ & 0.924696 & 0.870234 & 0.924777 & 0.868219 & 0.922867 & 0.868982 \\
    \midrule
     \multirow{3}{*}{150} & $q_2 = 0.2$ & 0.412558 & 0.205877 & 0.412212 & 0.205819 & 0.412529 & 0.205947 \\ 
    & $q_2 = 0.1$ & 0.565750 & 0.330461 & 0.566835 & 0.330390 & 0.565658 & 0.330595  \\ 
    & $q_2 = 1/200$ & 0.958730 & 0.892018 & 0.957689 & 0.893823 & 0.958287 & 0.894166\\
    \midrule
     \multirow{3}{*}{200} & $q_2 = 0.2$ & 0.383624 & 0.248326 & 0.384346 & 0.248388 & 0.383731 & 0.248295  \\ 
    & $q_2 = 0.1$ & 0.520121 & 0.374873 & 0.521217 & 0.374745 & 0.521213 & 0.374880 \\ 
    & $q_2 = 1/200$ & 0.934772 & 0.888740 & 0.934523 & 0.889443 & 0.936059 & 0.888556\\
    \bottomrule
    \end{tabular}\label{tabB:MCCq1q2}
    }%
    \end{table}

\begin{table}[H]
    \centering
    \caption{Relative estimation error of $B$ averaged over 200 replicates for different $q_1$ and $ q_2 $. }
      \resizebox{0.8\textwidth}{!}{%
     \begin{tabular}{cccccccc}
    \toprule
    Sample size & & \multicolumn{2}{c}{$q_1 = 0.2$} & \multicolumn{2}{c}{$q_1 = 0.1$ }& \multicolumn{2}{c}{$q_1 = 1/200$ } \\
    \cmidrule(lr){3-4} \cmidrule(lr){5-6} \cmidrule(lr){7-8}
    $n$& & Joint & Stepwise & Joint & Stepwise & Joint & Stepwise \\
    \midrule
     \multirow{3}{*}{50} & $q_2 = 0.2$ & 0.079190 & 0.204494 & 0.082342 & 0.204656 & 0.084490 & 0.204276\\ 
    & $q_2 = 0.1$ & 0.041919 & 0.145363 & 0.041923 & 0.146736 & 0.042020 & 0.145176   \\ 
    & $q_2 = 1/200$ & 0.031765 & 0.021795 & 0.031596 & 0.021790 & 0.031658 & 0.021832 \\
    \midrule
     \multirow{3}{*}{100} & $q_2 = 0.2$ &  0.029408 & 0.104857 & 0.029514 & 0.105318 & 0.029476 & 0.104686 \\ 
    & $q_2 = 0.1$ & 0.021651 & 0.055278 & 0.021727 & 0.055390 & 0.021742 & 0.055232  \\ 
    & $q_2 = 1/200$ &  0.011085 & 0.012678 & 0.011081 & 0.012737 & 0.011161 & 0.012734  \\
    \midrule
     \multirow{3}{*}{150} & $q_2 = 0.2$ &  0.028091 & 0.060263 & 0.028124 & 0.060331 & 0.028103 & 0.060318 \\ 
    & $q_2 = 0.1$ & 0.019697 & 0.039141 & 0.019654 & 0.039091 & 0.019711 & 0.039098 \\ 
    & $q_2 = 1/200$ & 0.007320 & 0.009555 & 0.007345 & 0.009527 & 0.007328 & 0.009517 \\
    \midrule
     \multirow{3}{*}{200} & $q_2 = 0.2$ &   0.026405 & 0.044825 & 0.026352 & 0.044801 & 0.026420 & 0.044862  \\ 
    & $q_2 = 0.1$ & 0.018945 & 0.030533 & 0.018890 & 0.030560 & 0.018888 & 0.030558  \\ 
    & $q_2 = 1/200$ & 0.006642 & 0.008472 & 0.006657 & 0.008449 & 0.006610 & 0.008475  \\
    \bottomrule
    \end{tabular}\label{tabB:rltnormq1q2}
    }%
    \end{table}
    
    \noindent
    For $q_1$ and $q_2$ we have also suggested taking Beta priors. The Beta prior in particular is attractive due to conditional conjugacy and the resulting computational simplicity of the conditional updates for $q_1$ and $q_2$. Below in Tables \ref{tabBhyp} and \ref{tabomegahyp} we present the sparsity selection performance in $B$ and $\Omega$ based on the MCC metric using Beta(1,1), i.e., uniform hyper-priors on $q_1$ and $q_2$ for all simulation settings for both the JRNS and Stepwise algorithms.
    
    \begin{table}[htbp]
    \centering
    \caption{Comparison of MCC values for sparsity selection in $B$ (averaged over 200 replicates) using fixed values for $q_1, q_2$ vs. 
    using a uniform hyper-prior for $q_1, q_2$.}
    \label{tab:priorB}
    \resizebox{1.0\textwidth}{!}{%
    \begin{tabular}{*{9}{c}}
        \toprule
          Sparsity & \multicolumn{2}{c}{Cases}& Joint   &  Joint  &  Stepwise    & Stepwise  \\
        \cmidrule(lr){2-3} 
       & $n$ & $(p,q)$ & fixed $q1$, $q2$ & Beta(1,1) hyperprior & fixed $q1$, $q2$ & Beta(1,1) hyperprior & \\ 
         \midrule
         \multirow{4}{*}{$(p/5,q/5)$}& 100 & $(30,60)$ &1.000  & 1.000 & 1.000  & 1.000 \\
         &100 & $(60,30)$ &1.000 & 1.000 & 1.000 & 1.000 \\
         &150 & $(200,200)$ & 1.000 & 1.000 & 0.997 & 1.000 \\
           &150 & $(300,300)$ & 0.998 & 1.000 & 0.770 & 0.982 \\   \midrule
           \multirow{ 2}{*}{$(p/30,q/5)$} & 100 & $(200,200)$ & 0.991 & 1.000 & 0.961 & 0.996 \\
           &200 & $(200,200)$ & 1.000 & 1.000 & 0.956 & 0.997   \\
        \bottomrule
    \end{tabular}\label{tabBhyp}
     }%
     
\end{table}

\begin{table}[htbp]
    \centering
    \caption{Comparison of MCC values for sparsity selection in $\Omega$ (averaged over 200 replicates) using fixed values for $q_1, q_2$ vs. 
    using a uniform hyper-prior for $q_1, q_2$.}
    \label{tab:prioromega}
    \resizebox{1.0\textwidth}{!}{%
    \begin{tabular}{*{9}{c}}
        \toprule
          Sparsity & \multicolumn{2}{c}{Cases}& Joint   &  Joint  &  Stepwise    & Stepwise  \\
        \cmidrule(lr){2-3} 
       & $n$ & $(p,q)$ & fixed $q1$, $q2$ & Beta(1,1) hyperprior & fixed $q1$, $q2$ & Beta(1,1) hyperprior & \\ 
         \midrule
         \multirow{4}{*}{$(p/5,q/5)$}& 100 & $(30,60)$ & 0.783 & 0.749 & 0.778 & 0.763 \\
         &100 & $(60,30)$ & 0.821 & 0.748 & 0.820 & 0.770 \\
         &150 & $(200,200)$ & 0.918 & 0.939 & 0.899 & 0.939 \\
           &150 & $(300,300)$ & 0.912 & 0.945 & 0.831 & 0.930\\  
           \midrule
           \multirow{ 2}{*}{$(p/30,q/5)$} & 100 & $(200,200)$& 0.867 & 0.856 & 0.846 & 0.859\\
           &200 & $(200,200)$ &0.969  &  0.972 & 0.968  & 0.971  \\
        \bottomrule
    \end{tabular}\label{tabomegahyp}
     }%
     
\end{table}

    \section{Additional network plots and corresponding inclusion probability tables}
    \label{suppsec:TCGAplots}
    
    Here we present the network plots for all the other cancer types apart from LUAD, the plots for which are provided in the main paper. Tables \ref{tab:COAD:geneIndices}, \ref{tab:OV:geneIndices}, \ref{tab:READ:geneIndices}, \ref{tab:SKCM:geneIndices} and \ref{tab:UCEC:geneIndices} provide the indices for all the genes and the proteins included in the dataset for these cancer types. The left panel in Figures \ref{fig:COADnetwork05} - \ref{fig:UCECnetwork05} indicates the associations between mRNA and proteins while the right panel in each of these figures indicate the associations among different proteins considered. The different colors of each node represent the pathway membership of the corresponding gene or protein, which is also provided in the form of a legend in each of these figures. The corresponding tables listing the inclusion probability of each included edge for both types of network plots and for each cancer type are also provided in this section. The network plots for the cancer type LUSC could not be provided here since the pathway membership information is missing for some of the genes and proteins in that dataset. Also, Table \ref{LUAD:Omega:incprob} here lists the inclusion probabilities of all the edges included in the network plot indicating associations among proteins for LUAD cancer given in the right panel of Figure 5 in the main paper.

     \begin{table}
     \centering
    \caption{Inclusion probability of each edge for the LUAD network graph indicating associations among proteins provided in the right panel of Figure 7 in the main paper.}
    \resizebox{0.8\textwidth}{!}{%
    \begin{minipage}{.4\linewidth}
      
      \centering
        \begin{tabular}{cccc}
            \hline
 & Protein  & Protein  & Inclusion   \\ 
 & & & Probability \\
  \hline
  1 &   1 &   3 & 1.00 \\ 
  2 &   2 &   3 & 1.00 \\ 
  3 &   2 &   7 & 1.00 \\ 
  4 &   4 &   7 & 1.00 \\ 
  5 &   2 &   8 & 0.74 \\ 
  6 &   6 &   8 & 1.00 \\ 
  7 &  11 &  12 & 1.00 \\ 
  8 &   3 &  13 & 1.00 \\ 
  9 &   8 &  13 & 1.00 \\ 
  10 &   1 &  14 & 1.00 \\ 
  11 &   5 &  14 & 1.00 \\ 
  12 &  13 &  14 & 1.00 \\ 
  13 &   5 &  15 & 1.00 \\ 
  14 &   7 &  15 & 0.61 \\ 
  15 &  14 &  15 & 1.00 \\ 
  16 &   4 &  16 & 1.00 \\ 
  17 &  11 &  16 & 1.00 \\ 
  18 &  11 &  17 & 1.00 \\ 
  19 &   6 &  18 & 1.00 \\ 
  20 &  14 &  20 & 1.00 \\ 
  21 &  17 &  21 & 1.00 \\ 
  22 &  18 &  22 & 1.00 \\ 
  23 &  19 &  22 & 1.00 \\ 
  24 &   3 &  23 & 1.00 \\ 
  25 &  20 &  23 & 1.00 \\ 
  26 &  21 &  23 & 1.00 \\ 
  27 &   2 &  24 & 1.00 \\ 
  28 &   8 &  24 & 1.00 \\ 
  29 &  14 &  24 & 1.00 \\ 
  30 &  17 &  24 & 1.00 \\ 
  31 &   9 &  25 & 1.00 \\ 
  32 &  19 &  25 & 1.00 \\ 
  33 &  13 &  26 & 1.00 \\ 
  34 &  23 &  26 & 1.00 \\ 
  35 &   2 &  27 & 1.00 \\ 
  36 &  23 &  29 & 1.00 \\ 
  37 &  26 &  29 & 1.00 \\
 
   \hline
        \end{tabular}
    \end{minipage}%
    \quad \quad
    \begin{minipage}{.4\linewidth}
      \centering
         \begin{tabular}{cccc}
            \hline
 & Protein  & Protein  & Inclusion    \\ 
 & & & Probability\\
  \hline
 38 &  19 &  31 & 1.00 \\ 
  39 &  28 &  32 & 1.00 \\ 
  40 &  31 &  32 & 1.00 \\ 
  41 &   2 &  33 & 1.00 \\ 
  42 &  25 &  33 & 1.00 \\ 
  43 &  26 &  33 & 1.00 \\ 
  44 &  32 &  33 & 1.00 \\ 
  45 &  28 &  34 & 1.00 \\ 
  46 &  29 &  35 & 1.00 \\ 
  47 &  13 &  36 & 1.00 \\ 
  48 &  35 &  36 & 1.00 \\ 
  49 &  21 &  37 & 1.00 \\ 
  50 &   2 &  38 & 1.00 \\ 
  51 &  25 &  39 & 1.00 \\ 
  52 &  36 &  40 & 1.00 \\ 
  53 &  19 &  41 & 1.00 \\ 
  54 &  27 &  41 & 1.00 \\ 
  55 &  10 &  42 & 1.00 \\ 
  56 &  41 &  42 & 1.00 \\ 
  57 &  30 &  43 & 1.00 \\ 
  58 &   2 &  44 & 1.00 \\ 
  59 &  41 &  44 & 1.00 \\ 
  60 &  43 &  44 & 1.00 \\ 
  61 &   6 &  45 & 1.00 \\ 
  62 &  25 &  45 & 1.00 \\ 
  63 &  42 &  45 & 1.00 \\ 
  64 &   5 &  46 & 1.00 \\ 
  65 &  41 &  46 & 1.00 \\ 
  66 &  43 &  46 & 1.00 \\ 
  67 &  45 &  46 & 1.00 \\ 
  68 &  33 &  47 & 1.00 \\ 
  69 &  14 &  49 & 1.00 \\ 
  70 &  41 &  49 & 1.00 \\ 
  71 &  21 &  50 & 1.00 \\ 
  72 &  22 &  50 & 1.00 \\ 
  73 &  23 &  50 & 1.00 \\ 
  74 &  10 &  51 & 1.00 \\ 
  \hline
  \end{tabular}
    \end{minipage}%
    \quad \quad
    \begin{minipage}{.4\linewidth}
      \centering
         \begin{tabular}{cccc}
            \hline
 & Protein  & Protein  & Inclusion  \\ 
 & & & Probability\\
  \hline
  75 &  11 &  51 & 1.00 \\ 
  76 &  18 &  51 & 1.00 \\ 
  77 &  30 &  51 & 1.00 \\ 
  78 &  41 &  52 & 1.00 \\ 
  79 &  41 &  53 & 1.00 \\ 
  80 &  50 &  53 & 1.00 \\ 
  81 &  52 &  53 & 1.00 \\ 
  82 &   6 &  54 & 1.00 \\ 
  83 &  51 &  54 & 1.00 \\ 
  84 &  53 &  54 & 0.50 \\ 
  85 &  26 &  55 & 1.00 \\ 
  86 &  45 &  55 & 1.00 \\ 
  87 &   6 &  56 & 1.00 \\ 
  88 &  19 &  56 & 1.00 \\ 
  89 &  55 &  56 & 1.00 \\ 
  90 &  33 &  57 & 1.00 \\ 
  91 &  39 &  57 & 0.95 \\ 
  92 &   3 &  58 & 1.00 \\ 
  93 &  13 &  58 & 1.00 \\ 
  94 &  50 &  58 & 1.00 \\ 
  95 &  57 &  59 & 1.00 \\ 
  96 &  58 &  59 & 1.00 \\ 
  97 &   3 &  60 & 1.00 \\ 
  98 &   8 &  60 & 1.00 \\ 
  99 &  21 &  60 & 1.00 \\ 
  100 &  29 &  60 & 1.00 \\ 
  101 &   3 &  61 & 1.00 \\ 
  102 &  49 &  61 & 1.00 \\ 
  103 &  57 &  61 & 1.00 \\ 
  104 &  40 &  62 & 1.00 \\ 
  105 &  57 &  62 & 1.00 \\ 
  106 &  54 &  63 & 1.00 \\ 
  107 &  59 &  63 & 1.00 \\ 
  108 &  62 &  63 & 1.00 \\ 
  109 &  27 &  64 & 1.00 \\ 
  110 &   6 &  65 & 1.00 \\ 
  111 &  64 &  65 & 1.00 \\
  \hline
  \end{tabular}
    \end{minipage}%
    \quad \quad
    \begin{minipage}{.4\linewidth}
      \centering
         \begin{tabular}{cccc}
            \hline
 & Protein  & Protein  & Inclusion   \\ 
 & & & Probability\\
  \hline
  112 &  21 &  66 & 1.00 \\ 
  113 &  33 &  66 & 0.56 \\ 
  114 &  34 &  66 & 0.59 \\ 
  115 &  64 &  66 & 1.00 \\ 
  116 &  29 &  67 & 1.00 \\ 
  117 &  37 &  67 & 1.00 \\ 
  118 &  41 &  67 & 1.00 \\ 
  119 &  51 &  67 & 1.00 \\ 
  120 &   3 &  68 & 0.97 \\ 
  121 &  23 &  68 & 1.00 \\ 
  122 &  52 &  68 & 1.00 \\ 
  123 &  55 &  68 & 1.00 \\ 
  124 &  63 &  68 & 1.00 \\ 
  125 &  56 &  69 & 1.00 \\ 
  126 &  65 &  70 & 1.00 \\ 
  127 &  69 &  70 & 1.00 \\ 
  128 &  11 &  71 & 1.00 \\ 
  129 &  49 &  71 & 1.00 \\ 
  130 &  65 &  71 & 1.00 \\ 
  131 &  15 &  72 & 1.00 \\ 
  132 &  30 &  72 & 1.00 \\ 
  133 &  41 &  72 & 1.00 \\ 
  134 &  66 &  72 & 1.00 \\ 
  135 &   5 &  73 & 1.00 \\ 
  136 &  72 &  73 & 1.00 \\ 
  137 &  26 &  74 & 1.00 \\ 
  138 &  58 &  74 & 1.00 \\ 
  139 &  16 &  75 & 1.00 \\ 
  140 &  25 &  75 & 1.00 \\ 
  141 &  30 &  75 & 1.00 \\ 
  142 &  38 &  75 & 1.00 \\ 
  143 &  18 &  76 & 1.00 \\ 
  144 &  30 &  76 & 1.00 \\ 
  145 &  47 &  76 & 1.00 \\ 
  & & & \\
  & & & \\
  & & & \\
   \hline
        \end{tabular}
        \end{minipage} 
    }
    \label{LUAD:Omega:incprob}
\end{table}
    
\begin{table}
\centering
    \caption{Indices of genes and proteins for COAD cancer data. The first column lists the components of the dataset mRNA(genes) and the second column lists the components of the dataset RPPA(proteins).}
    \resizebox{0.8\textwidth}{!}{%
    \begin{minipage}{.5\linewidth}
      
      \centering
        \begin{tabular}{rll}
            \hline
 & Gene & Protein \\ 
  \hline
1 & BAK1 & BAK \\ 
  2 & BAX & BAX \\ 
  3 & BID & BID \\ 
  4 & BCL2L11 & BIM \\ 
  5 & CASP7 & CASPASE7CLEAVEDD198 \\ 
  6 & BAD & BADPS112 \\ 
  7 & BCL2 & BCL2 \\ 
  8 & BCL2L1 & BCLXL \\ 
  9 & BIRC2 & CIAP \\ 
  10 & CDK1 & CDK1 \\ 
  11 & CCNB1 & CYCLINB1 \\ 
  12 & CCNE1 & CYCLINE1 \\ 
  13 & CCNE2 & CYCLINE2 \\ 
  14 & CDKN1B & P27PT157 \\ 
  15 & PCNA & P27PT198 \\ 
  16 & FOXM1 & PCNA \\ 
  17 & TP53BP1 & FOXM1 \\ 
  18 & ATM & 53BP1 \\ 
  19 & BRCA2 & ATM \\ 
  20 & CHEK1 & CHK1PS345 \\ 
  21 & CHEK2 & CHK2PT68 \\ 
  22 & XRCC5 & KU80 \\ 
  23 & MRE11A & MRE11 \\ 
  24 & TP53 & P53 \\ 
  25 & RAD50 & RAD50 \\ 
  26 & RAD51 & RAD51 \\ 
  27 & XRCC1 & XRCC1 \\ 
  28 & FN1 & FIBRONECTIN \\ 
  29 & CDH2 & NCADHERIN \\ 
  30 & COL6A1 & COLLAGENVI \\ 
  31 & CLDN7 & CLAUDIN7 \\ 
  32 & CDH1 & ECADHERIN \\ 
  33 & CTNNB1 & BETACATENIN \\ 
  34 & SERPINE1 & PAI1 \\ 
  35 & ESR1 & ERALPHA \\ 
  36 & PGR & ERALPHAPS118 \\ 
  37 & AR & PR \\ 
  38 & INPP4B & AR \\
   \hline
        \end{tabular}
    \end{minipage}%
    \quad \quad
    \begin{minipage}{.5\linewidth}
      \centering
         \begin{tabular}{rll}
            \hline
 & Gene & Protein \\ 
  \hline
 39 & GATA3 & INPP4B \\ 
  40 & AKT1 & GATA3 \\ 
  41 & AKT2 & AKTPS473 \\ 
  42 & AKT3 & AKTPT308 \\ 
  43 & GSK3A & GSK3ALPHABETAPS21S9 \\ 
  44 & GSK3B & GSK3PS9 \\ 
  45 & AKT1S1 & PRAS40PT246 \\ 
  46 & TSC2 & TUBERINPT1462 \\ 
  47 & PTEN & PTEN \\ 
  48 & ARAF & ARAFPS299 \\ 
  49 & JUN & CJUNPS73 \\ 
  50 & RAF1 & CRAFPS338 \\ 
  51 & MAPK8 & JNKPT183Y185 \\ 
  52 & MAPK1 & MAPKPT202Y204 \\ 
  53 & MAPK3 & MEK1PS217S221 \\ 
  54 & MAP2K1 & P38PT180Y182 \\ 
  55 & MAPK14 & P90RSKPT359S363 \\ 
  56 & RPS6KA1 & YB1PS102 \\ 
  57 & YBX1 & EGFRPY1068 \\ 
  58 & EGFR & EGFRPY1173 \\ 
  59 & ERBB2 & HER2PY1248 \\ 
  60 & ERBB3 & HER3PY1298 \\ 
  61 & SHC1 & SHCPY317 \\ 
  62 & SRC & SRCPY416 \\ 
  63 & EIF4EBP1 & SRCPY527 \\ 
  64 & RPS6KB1 & 4EBP1PS65 \\ 
  65 & MTOR & 4EBP1PT37T46 \\ 
  66 & RPS6 & 4EBP1PT70 \\ 
  67 & RB1 & P70S6KPT389 \\ 
  68 & CAV1 & MTORPS2448 \\ 
  69 & MYH11 & S6PS235S236 \\ 
  70 & RAB11A & S6PS240S244 \\ 
  71 & RAB11B & RBPS807S811 \\ 
  72 & GAPDH & CAVEOLIN1 \\ 
  73 & RBM15 & MYH11 \\ 
  74 &  & RAB11 \\ 
  75 &  & GAPDH \\ 
  76 &  & RBM15 \\  
  \hline
        \end{tabular}
        \end{minipage} 
    }
    \label{tab:COAD:geneIndices}
\end{table}    
    
      \begin{figure}[H]
    \centering
    \includegraphics[width = \linewidth]{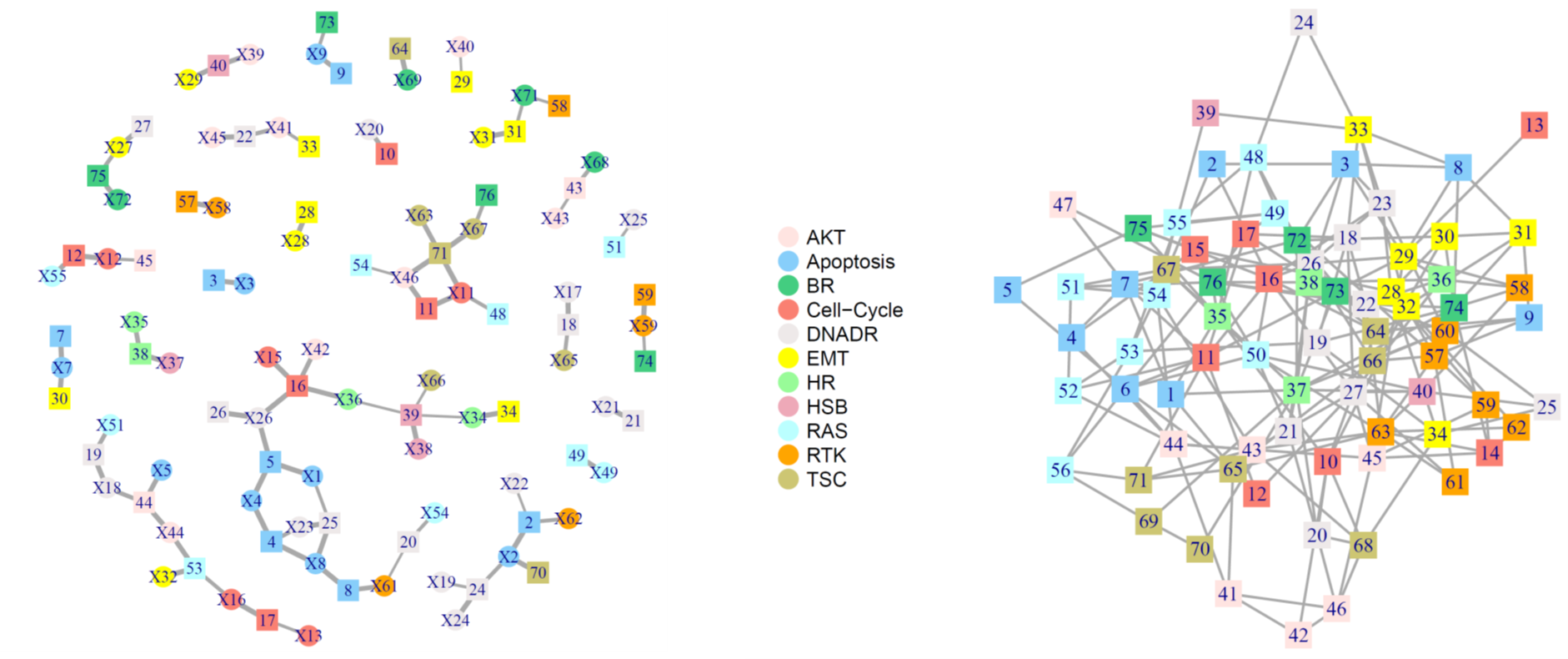}
    \caption{COAD networks with 0.5 as the inclusion probability cutoff. The circles represent genes and the squares represent proteins. The different colors represent the different pathways listed in Table 14 in Appendix. Left : Network graph indicating associations between mRNA and protein. Right : Network graph indicating associations among proteins. The inclusion probabilities are listed in Tables \ref{COAD:B:incprob} and \ref{COAD:Omega:incprob}. \textit{All the edge widths are proportional to the corresponding inclusion probabilities.}}
    \label{fig:COADnetwork05}
\end{figure}
    
    \begin{table}[H]
\centering
    \caption{Inclusion probability of each edge for the COAD network graph indicating associations between mRNA and protein provided in the left panel of Figure \ref{fig:COADnetwork05}.}
    \resizebox{0.9\textwidth}{!}{%
    \begin{minipage}{.4\linewidth}
      
      \centering
        \begin{tabular}{cccc}
            \hline
 & Gene  & Protein  & Inclusion  \\
 & & & Probability\\
  \hline

 1 & X2 &   2 & 1.00 \\ 
  2 & X22 &   2 & 0.50 \\ 
  3 & X62 &   2 & 0.74 \\ 
  4 & X3 &   3 & 0.98 \\ 
  5 & X4 &   4 & 1.00 \\ 
  6 & X8 &   4 & 1.00 \\ 
  7 & X23 &   4 & 0.80 \\ 
  8 & X1 &   5 & 1.00 \\ 
  9 & X4 &   5 & 1.00 \\ 
  10 & X26 &   5 & 0.98 \\ 
  11 & X7 &   7 & 1.00 \\ 
  12 & X8 &   8 & 1.00 \\ 
  13 & X61 &   8 & 0.92 \\ 
  14 & X9 &   9 & 0.93 \\ 
  15 & X20 &  10 & 0.89 \\ 
  16 & X11 &  11 & 1.00 \\ 
  17 & X46 &  11 & 0.99 \\ 
  18 & X12 &  12 & 1.00 \\ 
  19 & X55 &  12 & 1.00 \\ 
  20 & X15 &  16 & 1.00 \\ 
  21 & X26 &  16 & 1.00 \\ 
  
   \hline
        \end{tabular}
    \end{minipage}%
    \quad \quad
    \begin{minipage}{.4\linewidth}
      \centering
         \begin{tabular}{cccc}
            \hline
 & Gene  & Protein  & Inclusion   \\ 
 & & & Probability\\
  \hline
 22 & X36 &  16 & 0.96 \\ 
  23 & X42 &  16 & 0.89 \\ 
  24 & X13 &  17 & 0.60 \\ 
  25 & X16 &  17 & 1.00 \\ 
  26 & X17 &  18 & 1.00 \\ 
  27 & X65 &  18 & 0.76 \\ 
  28 & X18 &  19 & 0.98 \\ 
  29 & X51 &  19 & 0.86 \\ 
  30 & X54 &  20 & 0.69 \\ 
  31 & X61 &  20 & 0.50 \\ 
  32 & X21 &  21 & 1.00 \\ 
  33 & X41 &  22 & 0.82 \\ 
  34 & X45 &  22 & 1.00 \\ 
  35 & X2 &  24 & 0.77 \\ 
  36 & X19 &  24 & 0.57 \\ 
  37 & X24 &  24 & 1.00 \\ 
  38 & X1 &  25 & 0.63 \\ 
  39 & X8 &  25 & 0.98 \\ 
  40 & X23 &  25 & 1.00 \\ 
  41 & X26 &  26 & 1.00 \\ 
  42 & X27 &  27 & 1.00 \\
  \hline
  \end{tabular}
    \end{minipage}%
    \quad \quad
    \begin{minipage}{.4\linewidth}
      \centering
         \begin{tabular}{cccc}
            \hline
 & Gene  & Protein  & Inclusion   \\ 
 & & & Probability\\
  \hline
  43 & X28 &  28 & 1.00 \\ 
  44 & X40 &  29 & 0.51 \\ 
  45 & X7 &  30 & 1.00 \\ 
  46 & X31 &  31 & 1.00 \\ 
  47 & X71 &  31 & 0.66 \\ 
  48 & X41 &  33 & 0.64 \\ 
  49 & X34 &  34 & 1.00 \\ 
  50 & X35 &  38 & 0.92 \\ 
  51 & X37 &  38 & 1.00 \\ 
  52 & X34 &  39 & 0.51 \\ 
  53 & X36 &  39 & 0.52 \\ 
  54 & X38 &  39 & 1.00 \\ 
  55 & X66 &  39 & 0.52 \\ 
  56 & X29 &  40 & 0.98 \\ 
  57 & X39 &  40 & 0.97 \\ 
  58 & X43 &  43 & 0.58 \\ 
  59 & X68 &  43 & 0.80 \\ 
  60 & X5 &  44 & 0.94 \\ 
  61 & X18 &  44 & 0.84 \\ 
  62 & X44 &  44 & 0.97 \\ 
  63 & X12 &  45 & 0.62 \\
  \hline
  \end{tabular}
    \end{minipage}%
    \quad \quad
    \begin{minipage}{.4\linewidth}
      \centering
         \begin{tabular}{cccc}
            \hline
 & Gene  & Protein  & Inclusion \\ 
 & & & Probability\\
  \hline
  64 & X11 &  48 & 0.72 \\ 
  65 & X49 &  49 & 0.99 \\ 
  66 & X25 &  51 & 0.55 \\ 
  67 & X16 &  53 & 0.73 \\ 
  68 & X32 &  53 & 1.00 \\ 
  69 & X44 &  53 & 0.96 \\ 
  70 & X46 &  54 & 0.52 \\ 
  71 & X58 &  57 & 1.00 \\ 
  72 & X71 &  58 & 0.56 \\ 
  73 & X59 &  59 & 1.00 \\ 
  74 & X69 &  64 & 1.00 \\ 
  75 & X2 &  70 & 1.00 \\ 
  76 & X11 &  71 & 1.00 \\ 
  77 & X46 &  71 & 1.00 \\ 
  78 & X63 &  71 & 0.93 \\ 
  79 & X67 &  71 & 1.00 \\ 
  80 & X9 &  73 & 0.68 \\ 
  81 & X59 &  74 & 0.58 \\ 
  82 & X27 &  75 & 0.79 \\ 
  83 & X72 &  75 & 1.00 \\ 
  84 & X67 &  76 & 0.76 \\ 
   \hline
        \end{tabular}
        \end{minipage} 
    }
    \label{COAD:B:incprob}
\end{table}
    
        \begin{table}[H]
\centering
    \caption{Inclusion probability of each edge for the COAD network graph indicating associations among proteins provided in the right panel of Figure \ref{fig:COADnetwork05}.}
    \resizebox{0.8\textwidth}{!}{%
    \begin{minipage}{.4\linewidth}
      
      \centering
        \begin{tabular}{cccc}
            \hline
 & Protein  & Protein  & Inclusion  \\ 
 & & & Probability \\
  \hline
  1 &   2 &   3 & 1.00 \\ 
  2 &   5 &   6 & 1.00 \\ 
  3 &   1 &   7 & 1.00 \\ 
  4 &   4 &   7 & 1.00 \\ 
  5 &   3 &   8 & 1.00 \\ 
  6 &   6 &  11 & 1.00 \\ 
  7 &   7 &  11 & 1.00 \\ 
  8 &  11 &  12 & 1.00 \\ 
  9 &   4 &  15 & 1.00 \\ 
  10 &   2 &  16 & 1.00 \\ 
  11 &  11 &  16 & 1.00 \\ 
  12 &  11 &  17 & 1.00 \\ 
  13 &  17 &  18 & 1.00 \\ 
  14 &  16 &  19 & 1.00 \\ 
  15 &  18 &  19 & 1.00 \\ 
  16 &  10 &  20 & 1.00 \\ 
  17 &  20 &  21 & 1.00 \\ 
  18 &  18 &  22 & 1.00 \\ 
  19 &   3 &  23 & 1.00 \\ 
  20 &  16 &  23 & 0.97 \\ 
  21 &  16 &  26 & 1.00 \\ 
  22 &  23 &  26 & 1.00 \\ 
  23 &  12 &  27 & 1.00 \\ 
  24 &  16 &  27 & 1.00 \\ 
  25 &  18 &  27 & 1.00 \\ 
  26 &  15 &  28 & 1.00 \\ 
  27 &  22 &  28 & 0.67 \\ 
  28 &  23 &  28 & 1.00 \\ 
  29 &   3 &  29 & 1.00 \\ 
  30 &   8 &  29 & 1.00 \\ 
  31 &  11 &  29 & 1.00 \\ 
  32 &  13 &  29 & 1.00 \\ 
  33 &  26 &  29 & 1.00 \\ 
  34 &  28 &  30 & 1.00 \\ 
  35 &   8 &  31 & 1.00 \\ 
  36 &   9 &  31 & 1.00 \\ 
  37 &  30 &  31 & 1.00 \\ 
  38 &  10 &  32 & 1.00 \\ 
  39 &  22 &  32 & 1.00 \\ 
  40 &  31 &  32 & 1.00 \\ 
  41 &   8 &  33 & 1.00 \\ 
  42 &  18 &  33 & 1.00 \\ 
  43 &  23 &  33 & 1.00 \\ 
  44 &  24 &  33 & 1.00 \\ 
  45 &  32 &  33 & 1.00 \\ 
  46 &  19 &  34 & 1.00 \\ 
  47 &  25 &  34 & 1.00 \\ 
  48 &  28 &  34 & 1.00 \\ 
  49 &   1 &  35 & 1.00 \\ 
  50 &  10 &  35 & 1.00 \\
 
   \hline
        \end{tabular}
    \end{minipage}%
    \quad \quad
    \begin{minipage}{.4\linewidth}
      \centering
         \begin{tabular}{cccc}
            \hline
 & Protein  & Protein  & Inclusion   \\ 
 & & & Probability\\
  \hline
 51 &  22 &  36 & 1.00 \\ 
  52 &  27 &  36 & 1.00 \\ 
  53 &  29 &  36 & 1.00 \\ 
  54 &   1 &  37 & 1.00 \\ 
  55 &  20 &  37 & 1.00 \\ 
  56 &  21 &  37 & 1.00 \\ 
  57 &  26 &  37 & 1.00 \\ 
  58 &  35 &  37 & 1.00 \\ 
  59 &   3 &  38 & 1.00 \\ 
  60 &  21 &  38 & 1.00 \\ 
  61 &  26 &  38 & 1.00 \\ 
  62 &  36 &  38 & 1.00 \\ 
  63 &  37 &  38 & 1.00 \\ 
  64 &  33 &  39 & 1.00 \\ 
  65 &  25 &  40 & 0.85 \\ 
  66 &  27 &  40 & 1.00 \\ 
  67 &  37 &  40 & 1.00 \\ 
  68 &  10 &  42 & 1.00 \\ 
  69 &  41 &  42 & 1.00 \\ 
  70 &  27 &  43 & 1.00 \\ 
  71 &   4 &  44 & 1.00 \\ 
  72 &   6 &  44 & 1.00 \\ 
  73 &  41 &  44 & 1.00 \\ 
  74 &  43 &  44 & 1.00 \\ 
  75 &  14 &  45 & 1.00 \\ 
  76 &  27 &  45 & 1.00 \\ 
  77 &  29 &  45 & 1.00 \\ 
  78 &  44 &  45 & 1.00 \\ 
  79 &  20 &  46 & 1.00 \\ 
  80 &  41 &  46 & 1.00 \\ 
  81 &  42 &  46 & 1.00 \\ 
  82 &  45 &  46 & 1.00 \\ 
  83 &   7 &  47 & 1.00 \\ 
  84 &  24 &  48 & 1.00 \\ 
  85 &  35 &  48 & 1.00 \\ 
  86 &  38 &  48 & 1.00 \\ 
  87 &  17 &  49 & 1.00 \\ 
  88 &   6 &  50 & 1.00 \\ 
  89 &  16 &  50 & 1.00 \\ 
  90 &  21 &  50 & 1.00 \\ 
  91 &  40 &  50 & 1.00 \\ 
  92 &  11 &  52 & 1.00 \\ 
  93 &  44 &  52 & 1.00 \\ 
  94 &  51 &  52 & 1.00 \\ 
  95 &  19 &  53 & 1.00 \\ 
  96 &  52 &  53 & 1.00 \\ 
  97 &   2 &  54 & 1.00 \\ 
  98 &   7 &  54 & 1.00 \\ 
  99 &  39 &  54 & 1.00 \\ 
  100 &  43 &  54 & 1.00 \\ 
  \hline
  \end{tabular}
    \end{minipage}%
    \quad \quad
    \begin{minipage}{.4\linewidth}
      \centering
         \begin{tabular}{cccc}
            \hline
 & Protein  & Protein  & Inclusion  \\ 
 & & & Probability\\
  \hline
  101 &  44 &  54 & 1.00 \\ 
  102 &  51 &  54 & 1.00 \\ 
  103 &  53 &  54 & 1.00 \\ 
  104 &  48 &  55 & 1.00 \\ 
  105 &  49 &  55 & 1.00 \\ 
  106 &  50 &  55 & 1.00 \\ 
  107 &  53 &  55 & 1.00 \\ 
  108 &   6 &  56 & 1.00 \\ 
  109 &  53 &  56 & 1.00 \\ 
  110 &   9 &  57 & 1.00 \\ 
  111 &  32 &  57 & 1.00 \\ 
  112 &  49 &  57 & 1.00 \\ 
  113 &  29 &  58 & 1.00 \\ 
  114 &  36 &  58 & 1.00 \\ 
  115 &  40 &  58 & 1.00 \\ 
  116 &  14 &  59 & 1.00 \\ 
  117 &  28 &  59 & 1.00 \\ 
  118 &  45 &  59 & 1.00 \\ 
  119 &  57 &  59 & 1.00 \\ 
  120 &   9 &  60 & 1.00 \\ 
  121 &  23 &  60 & 1.00 \\ 
  122 &  25 &  60 & 1.00 \\ 
  123 &  37 &  60 & 1.00 \\ 
  124 &  57 &  60 & 1.00 \\ 
  125 &  27 &  61 & 1.00 \\ 
  126 &  57 &  61 & 1.00 \\ 
  127 &  14 &  62 & 1.00 \\ 
  128 &  32 &  62 & 1.00 \\ 
  129 &  14 &  63 & 1.00 \\ 
  130 &  22 &  63 & 1.00 \\ 
  131 &  40 &  63 & 1.00 \\ 
  132 &  43 &  63 & 1.00 \\ 
  133 &  50 &  63 & 1.00 \\ 
  134 &  57 &  63 & 1.00 \\ 
  135 &  61 &  63 & 1.00 \\ 
  136 &  62 &  63 & 1.00 \\ 
  137 &   9 &  64 & 1.00 \\ 
  138 &  10 &  64 & 1.00 \\ 
  139 &  16 &  64 & 1.00 \\ 
  140 &  38 &  64 & 1.00 \\ 
  141 &   1 &  65 & 1.00 \\ 
  142 &   6 &  65 & 1.00 \\ 
  143 &  41 &  65 & 1.00 \\ 
  144 &  63 &  65 & 0.92 \\ 
  145 &  26 &  66 & 1.00 \\ 
  146 &  30 &  66 & 0.91 \\ 
  147 &  57 &  66 & 1.00 \\ 
  148 &  60 &  66 & 1.00 \\ 
  149 &  65 &  66 & 1.00 \\ 
  150 &  26 &  67 & 1.00 \\ 
  \hline
  \end{tabular}
    \end{minipage}%
    \quad \quad
    \begin{minipage}{.4\linewidth}
      \centering
         \begin{tabular}{cccc}
            \hline
 & Protein  & Protein  & Inclusion   \\ 
 & & & Probability\\
  \hline
  151 &  35 &  67 & 1.00 \\ 
  152 &  50 &  67 & 1.00 \\ 
  153 &  51 &  67 & 1.00 \\ 
  154 &  10 &  68 & 1.00 \\ 
  155 &  20 &  68 & 1.00 \\ 
  156 &  34 &  68 & 1.00 \\ 
  157 &  43 &  68 & 0.80 \\ 
  158 &  46 &  68 & 1.00 \\ 
  159 &  37 &  69 & 1.00 \\ 
  160 &  56 &  69 & 1.00 \\ 
  161 &  37 &  70 & 1.00 \\ 
  162 &  69 &  70 & 1.00 \\ 
  163 &  11 &  71 & 1.00 \\ 
  164 &  43 &  71 & 1.00 \\ 
  165 &  56 &  71 & 1.00 \\ 
  166 &  65 &  71 & 1.00 \\ 
  167 &   3 &  72 & 1.00 \\ 
  168 &   7 &  72 & 1.00 \\ 
  169 &  30 &  72 & 1.00 \\ 
  170 &  37 &  72 & 1.00 \\ 
  171 &  48 &  72 & 1.00 \\ 
  172 &  64 &  72 & 1.00 \\ 
  173 &  17 &  73 & 1.00 \\ 
  174 &  37 &  73 & 1.00 \\ 
  175 &  60 &  73 & 1.00 \\ 
  176 &  72 &  73 & 1.00 \\ 
  177 &   8 &  74 & 1.00 \\ 
  178 &  14 &  74 & 1.00 \\ 
  179 &  19 &  74 & 1.00 \\ 
  180 &  22 &  74 & 1.00 \\ 
  181 &  30 &  74 & 1.00 \\ 
  182 &  31 &  74 & 1.00 \\ 
  183 &  37 &  74 & 1.00 \\ 
  184 &  58 &  74 & 1.00 \\ 
  185 &   5 &  75 & 1.00 \\ 
  186 &  15 &  75 & 1.00 \\ 
  187 &  16 &  75 & 1.00 \\ 
  188 &  35 &  75 & 1.00 \\ 
  189 &  49 &  75 & 1.00 \\ 
  190 &   6 &  76 & 1.00 \\ 
  191 &  17 &  76 & 1.00 \\ 
  192 &  22 &  76 & 1.00 \\ 
  193 &  32 &  76 & 1.00 \\ 
  194 &  35 &  76 & 1.00 \\ 
  195 &  38 &  76 & 1.00 \\ 
  196 &  47 &  76 & 1.00 \\ 
  197 &  51 &  76 & 1.00 \\ 
  & & & \\
  & & & \\
  & & & \\
   \hline
        \end{tabular}
        \end{minipage} 
    }
    \label{COAD:Omega:incprob}
\end{table}

\begin{table}[H]
\centering
    \caption{Indices of genes and proteins for OV cancer data. The first column lists the components of the dataset mRNA(genes) and the second column lists the components of the dataset RPPA(proteins).}
    \resizebox{0.8\textwidth}{!}{%
    \begin{minipage}{.5\linewidth}
      
      \centering
        \begin{tabular}{rll}
            \hline
 & Gene & Protein \\ 
  \hline
1 & BAK1 & BAK \\ 
  2 & BAX & BAX \\ 
  3 & BID & BID \\ 
  4 & BCL2L11 & BIM \\ 
  5 & CASP7 & CASPASE7CLEAVEDD198 \\ 
  6 & BAD & BADPS112 \\ 
  7 & BCL2 & BCL2 \\ 
  8 & BCL2L1 & BCLXL \\ 
  9 & BIRC2 & CIAP \\ 
  10 & CDK1 & CDK1 \\ 
  11 & CCNB1 & CYCLINB1 \\ 
  12 & CCNE1 & CYCLINE1 \\ 
  13 & CCNE2 & CYCLINE2 \\ 
  14 & CDKN1B & P27PT157 \\ 
  15 & PCNA & P27PT198 \\ 
  16 & FOXM1 & PCNA \\ 
  17 & TP53BP1 & FOXM1 \\ 
  18 & ATM & 53BP1 \\ 
  19 & BRCA2 & ATM \\ 
  20 & CHEK1 & BRCA2 \\ 
  21 & CHEK2 & CHK1PS345 \\ 
  22 & XRCC5 & CHK2PT68 \\ 
  23 & MRE11A & KU80 \\ 
  24 & TP53 & MRE11 \\ 
  25 & RAD50 & P53 \\ 
  26 & RAD51 & RAD50 \\ 
  27 & XRCC1 & RAD51 \\ 
  28 & FN1 & XRCC1 \\ 
  29 & CDH2 & FIBRONECTIN \\ 
  30 & COL6A1 & NCADHERIN \\ 
  31 & CLDN7 & COLLAGENVI \\ 
  32 & CDH1 & CLAUDIN7 \\ 
  33 & CTNNB1 & ECADHERIN \\ 
  34 & SERPINE1 & BETACATENIN \\ 
  35 & ESR1 & PAI1 \\ 
  36 & PGR & ERALPHA \\ 
  37 & AR & ERALPHAPS118 \\ 
  38 & INPP4B & PR \\ 
  39 & GATA3 & AR \\
   \hline
        \end{tabular}
    \end{minipage}%
    \quad \quad
    \begin{minipage}{.5\linewidth}
      \centering
         \begin{tabular}{rll}
            \hline
 & Gene & Protein \\ 
  \hline
  40 & AKT1 & INPP4B \\ 
  41 & AKT2 & GATA3 \\ 
  42 & AKT3 & AKTPS473 \\ 
  43 & GSK3A & AKTPT308 \\ 
  44 & GSK3B & GSK3ALPHABETAPS21S9 \\ 
  45 & AKT1S1 & GSK3PS9 \\ 
  46 & TSC2 & PRAS40PT246 \\ 
  47 & PTEN & TUBERINPT1462 \\ 
  48 & ARAF & PTEN \\ 
  49 & JUN & ARAFPS299 \\ 
  50 & RAF1 & CJUNPS73 \\ 
  51 & MAPK8 & CRAFPS338 \\ 
  52 & MAPK1 & JNKPT183Y185 \\ 
  53 & MAPK3 & MAPKPT202Y204 \\ 
  54 & MAP2K1 & MEK1PS217S221 \\ 
  55 & MAPK14 & P38PT180Y182 \\ 
  56 & RPS6KA1 & P90RSKPT359S363 \\ 
  57 & YBX1 & YB1PS102 \\ 
  58 & EGFR & EGFRPY1068 \\ 
  59 & ERBB2 & EGFRPY1173 \\ 
  60 & ERBB3 & HER2PY1248 \\ 
  61 & SHC1 & HER3PY1298 \\ 
  62 & SRC & SHCPY317 \\ 
  63 & EIF4EBP1 & SRCPY416 \\ 
  64 & RPS6KB1 & SRCPY527 \\ 
  65 & MTOR & 4EBP1PS65 \\ 
  66 & RPS6 & 4EBP1PT37T46 \\ 
  67 & RB1 & 4EBP1PT70 \\ 
  68 & CAV1 & P70S6KPT389 \\ 
  69 & MYH11 & MTORPS2448 \\ 
  70 & RAB11A & S6PS235S236 \\ 
  71 & RAB11B & S6PS240S244 \\ 
  72 & GAPDH & RBPS807S811 \\ 
  73 & RBM15 & CAVEOLIN1 \\ 
  74 &  & MYH11 \\ 
  75 &  & RAB11 \\ 
  76 &  & GAPDH \\ 
  77 &  & RBM15 \\ 
  & & \\
  \hline
        \end{tabular}
        \end{minipage} 
    }
    \label{tab:OV:geneIndices}
\end{table}

\begin{figure}[H]
    \centering
    \includegraphics[width = \linewidth]{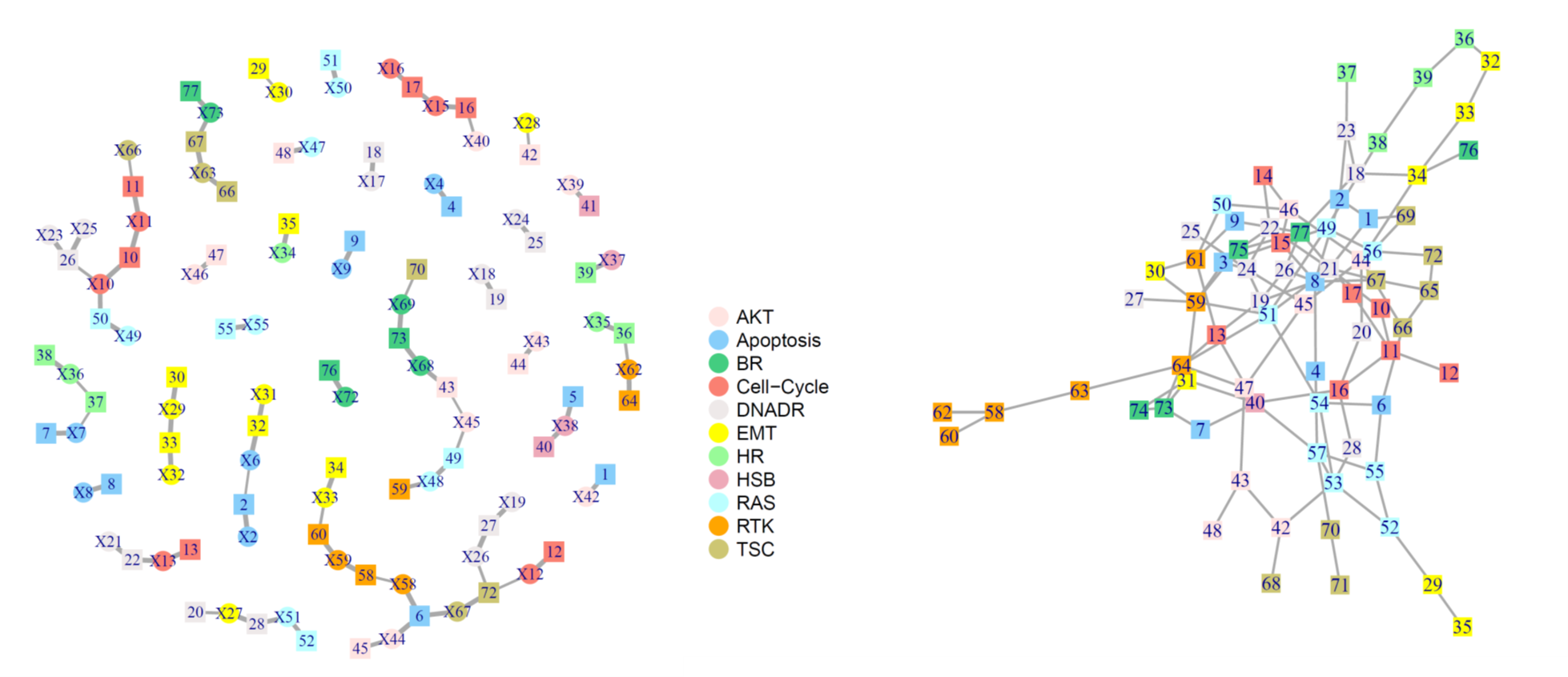}
    \caption{OV networks with 0.5 as the inclusion probability cutoff. The circles represent genes and the squares represent proteins. The different colors represent the different pathways listed in Table 14 in Appendix. Left : Network graph indicating associations between mRNA and protein. Right : Network graph indicating associations among proteins. The inclusion probabilities are listed in Tables \ref{OV:B:incprob} and \ref{OV:Omega:incprob} . \textit{All the edge widths are proportional to the corresponding inclusion probabilities.}}
    \label{fig:OVnetwork05}
\end{figure}

 \begin{table}[H]
\centering
    \caption{Inclusion probability of each edge for the OV network graph indicating associations between mRNA and proteins provided in the left panel of Figure \ref{fig:OVnetwork05}.}
    \resizebox{0.9\textwidth}{!}{%
    \begin{minipage}{.4\linewidth}
      
      \centering
        \begin{tabular}{cccc}
            \hline
 & Gene  & Protein  & Inclusion \\ 
 & & & Probability \\
  \hline
  1 & X42 &   1 & 0.94 \\ 
  2 & X2 &   2 & 1.00 \\ 
  3 & X6 &   2 & 0.50 \\ 
  4 & X4 &   4 & 1.00 \\ 
  5 & X38 &   5 & 0.98 \\ 
  6 & X44 &   6 & 1.00 \\ 
  7 & X58 &   6 & 1.00 \\ 
  8 & X67 &   6 & 0.95 \\ 
  9 & X7 &   7 & 1.00 \\ 
  10 & X8 &   8 & 1.00 \\ 
  11 & X9 &   9 & 1.00 \\ 
  12 & X10 &  10 & 1.00 \\ 
  13 & X11 &  10 & 1.00 \\ 
  14 & X11 &  11 & 1.00 \\ 
  15 & X66 &  11 & 0.54 \\ 
  16 & X12 &  12 & 1.00 \\ 
  17 & X13 &  13 & 1.00 \\ 
  18 & X15 &  16 & 1.00 \\ 
  19 & X40 &  16 & 0.60 \\ 
  20 & X15 &  17 & 0.95 \\ 
  21 & X16 &  17 & 1.00 \\ 
 
   \hline
        \end{tabular}
    \end{minipage}%
    \quad \quad
    \begin{minipage}{.4\linewidth}
      \centering
         \begin{tabular}{cccc}
            \hline
 & Gene  & Protein  & Inclusion   \\ 
 & & & Probability\\
  \hline
  22 & X17 &  18 & 1.00 \\ 
  23 & X18 &  19 & 1.00 \\ 
  24 & X27 &  20 & 0.61 \\ 
  25 & X13 &  22 & 1.00 \\ 
  26 & X21 &  22 & 1.00 \\ 
  27 & X24 &  25 & 1.00 \\ 
  28 & X10 &  26 & 0.60 \\ 
  29 & X23 &  26 & 1.00 \\ 
  30 & X25 &  26 & 1.00 \\ 
  31 & X19 &  27 & 0.99 \\ 
  32 & X26 &  27 & 1.00 \\ 
  33 & X27 &  28 & 1.00 \\ 
  34 & X51 &  28 & 0.93 \\ 
  35 & X30 &  29 & 0.55 \\ 
  36 & X29 &  30 & 1.00 \\ 
  37 & X6 &  32 & 1.00 \\ 
  38 & X31 &  32 & 1.00 \\ 
  39 & X29 &  33 & 1.00 \\ 
  40 & X32 &  33 & 1.00 \\ 
  41 & X33 &  34 & 1.00 \\ 
  42 & X34 &  35 & 1.00 \\
  \hline
  \end{tabular}
    \end{minipage}%
    \quad \quad
    \begin{minipage}{.4\linewidth}
      \centering
         \begin{tabular}{cccc}
            \hline
 & Gene  & Protein  & Inclusion   \\ 
 & & & Probability\\
  \hline
  43 & X35 &  36 & 1.00 \\ 
  44 & X62 &  36 & 0.62 \\ 
  45 & X7 &  37 & 0.50 \\ 
  46 & X36 &  37 & 0.57 \\ 
  47 & X36 &  38 & 1.00 \\ 
  48 & X37 &  39 & 1.00 \\ 
  49 & X38 &  40 & 1.00 \\ 
  50 & X39 &  41 & 1.00 \\ 
  51 & X28 &  42 & 0.57 \\ 
  52 & X45 &  43 & 0.50 \\ 
  53 & X68 &  43 & 0.91 \\ 
  54 & X43 &  44 & 1.00 \\ 
  55 & X44 &  45 & 0.84 \\ 
  56 & X46 &  47 & 1.00 \\ 
  57 & X47 &  48 & 1.00 \\ 
  58 & X45 &  49 & 0.59 \\ 
  59 & X48 &  49 & 1.00 \\ 
  60 & X10 &  50 & 0.94 \\ 
  61 & X49 &  50 & 1.00 \\ 
  62 & X50 &  51 & 1.00 \\ 
  63 & X51 &  52 & 0.90 \\
  \hline
  \end{tabular}
    \end{minipage}%
    \quad \quad
    \begin{minipage}{.4\linewidth}
      \centering
         \begin{tabular}{cccc}
            \hline
 & Gene  & Protein  & Inclusion  \\ 
 & & & Probability\\
  \hline
  64 & X55 &  55 & 1.00 \\ 
  65 & X58 &  58 & 0.52 \\ 
  66 & X59 &  58 & 1.00 \\ 
  67 & X48 &  59 & 0.83 \\ 
  68 & X33 &  60 & 0.60 \\ 
  69 & X59 &  60 & 1.00 \\ 
  70 & X62 &  64 & 0.99 \\ 
  71 & X63 &  66 & 1.00 \\ 
  72 & X63 &  67 & 1.00 \\ 
  73 & X73 &  67 & 0.88 \\ 
  74 & X69 &  70 & 0.53 \\ 
  75 & X12 &  72 & 0.52 \\ 
  76 & X26 &  72 & 0.59 \\ 
  77 & X67 &  72 & 0.98 \\ 
  78 & X68 &  73 & 1.00 \\ 
  79 & X69 &  73 & 1.00 \\ 
  80 & X72 &  76 & 1.00 \\ 
  81 & X73 &  77 & 1.00 \\ 
  & & & \\
  & & & \\
  & & & \\
   \hline
        \end{tabular}
        \end{minipage} 
    }
    \label{OV:B:incprob}
\end{table}
    
     \begin{table}[H]
\centering
    \caption{Inclusion probability of each edge for the OV network graph indicating associations among proteins provided in the right panel of Figure \ref{fig:OVnetwork05}.}
    \resizebox{0.9\textwidth}{!}{%
    \begin{minipage}{.4\linewidth}
      
      \centering
        \begin{tabular}{cccc}
            \hline
 & Gene  & Protein  & Inclusion  \\ 
 & & & Probability \\
  \hline
  1 &   1 &   2 & 1.00 \\ 
  2 &   2 &   8 & 1.00 \\ 
  3 &   4 &   8 & 1.00 \\ 
  4 &  10 &  11 & 1.00 \\ 
  5 &  11 &  12 & 1.00 \\ 
  6 &   3 &  15 & 1.00 \\ 
  7 &   8 &  15 & 1.00 \\ 
  8 &  11 &  16 & 1.00 \\ 
  9 &  11 &  17 & 1.00 \\ 
  10 &   8 &  19 & 1.00 \\ 
  11 &  13 &  19 & 1.00 \\ 
  12 &  16 &  20 & 1.00 \\ 
  13 &   1 &  21 & 1.00 \\ 
  14 &  10 &  21 & 1.00 \\ 
  15 &  14 &  22 & 1.00 \\ 
  16 &  21 &  22 & 1.00 \\ 
  17 &   2 &  23 & 1.00 \\ 
  18 &  18 &  23 & 1.00 \\ 
  19 &   3 &  24 & 1.00 \\ 
  20 &  22 &  24 & 1.00 \\ 
  21 &  24 &  25 & 1.00 \\ 
  22 &   8 &  26 & 1.00 \\ 
  23 &  19 &  26 & 1.00 \\ 
  24 &  16 &  28 & 1.00 \\ 
  25 &  32 &  33 & 1.00 \\ 
  26 &  18 &  34 & 1.00 \\ 
  27 &  33 &  34 & 1.00 \\ 
  28 &  29 &  35 & 1.00 \\ 
 
   \hline
        \end{tabular}
    \end{minipage}%
    \quad \quad
    \begin{minipage}{.4\linewidth}
      \centering
         \begin{tabular}{cccc}
            \hline
 & Gene  & Protein  & Inclusion   \\ 
 & & & Probability\\
  \hline
 29 &  32 &  36 & 1.00 \\ 
  30 &  23 &  37 & 1.00 \\ 
  31 &  36 &  39 & 1.00 \\ 
  32 &  38 &  39 & 1.00 \\ 
  33 &   7 &  40 & 1.00 \\ 
  34 &  13 &  40 & 1.00 \\ 
  35 &  16 &  40 & 1.00 \\ 
  36 &  31 &  40 & 1.00 \\ 
  37 &  42 &  43 & 1.00 \\ 
  38 &  20 &  44 & 1.00 \\ 
  39 &  34 &  44 & 1.00 \\ 
  40 &  24 &  45 & 1.00 \\ 
  41 &  44 &  45 & 1.00 \\ 
  42 &  14 &  46 & 1.00 \\ 
  43 &  15 &  46 & 1.00 \\ 
  44 &  44 &  46 & 1.00 \\ 
  45 &  43 &  47 & 1.00 \\ 
  46 &  45 &  47 & 1.00 \\ 
  47 &  43 &  48 & 1.00 \\ 
  48 &  15 &  49 & 1.00 \\ 
  49 &  38 &  49 & 1.00 \\ 
  50 &  46 &  50 & 1.00 \\ 
  51 &  21 &  51 & 1.00 \\ 
  52 &  24 &  51 & 1.00 \\ 
  53 &  49 &  51 & 1.00 \\ 
  54 &  29 &  52 & 1.00 \\ 
  55 &  28 &  53 & 1.00 \\ 
  56 &  42 &  53 & 1.00 \\
  \hline
  \end{tabular}
    \end{minipage}%
    \quad \quad
    \begin{minipage}{.4\linewidth}
      \centering
         \begin{tabular}{cccc}
            \hline
 & Gene  & Protein  & Inclusion   \\ 
 & & & Probability\\
  \hline
  57 &  52 &  53 & 1.00 \\ 
  58 &   6 &  54 & 1.00 \\ 
  59 &  51 &  54 & 1.00 \\ 
  60 &  53 &  54 & 1.00 \\ 
  61 &   6 &  55 & 1.00 \\ 
  62 &  52 &  55 & 1.00 \\ 
  63 &  45 &  56 & 1.00 \\ 
  64 &  49 &  56 & 1.00 \\ 
  65 &   4 &  57 & 1.00 \\ 
  66 &  40 &  57 & 1.00 \\ 
  67 &  53 &  57 & 1.00 \\ 
  68 &  55 &  57 & 1.00 \\ 
  69 &   3 &  59 & 1.00 \\ 
  70 &  27 &  59 & 1.00 \\ 
  71 &  30 &  59 & 1.00 \\ 
  72 &  31 &  59 & 1.00 \\ 
  73 &  51 &  59 & 1.00 \\ 
  74 &  58 &  60 & 1.00 \\ 
  75 &   9 &  61 & 1.00 \\ 
  76 &  13 &  61 & 1.00 \\ 
  77 &  30 &  61 & 1.00 \\ 
  78 &  50 &  61 & 1.00 \\ 
  79 &  58 &  62 & 1.00 \\ 
  80 &  58 &  63 & 1.00 \\ 
  81 &  13 &  64 & 1.00 \\ 
  82 &  47 &  64 & 1.00 \\ 
  83 &  51 &  64 & 1.00 \\ 
  84 &  63 &  64 & 1.00 \\
  \hline
  \end{tabular}
    \end{minipage}%
    \quad \quad
    \begin{minipage}{.4\linewidth}
      \centering
         \begin{tabular}{cccc}
            \hline
 & Gene  & Protein  & Inclusion  \\ 
 & & & Probability\\
  \hline
  85 &   6 &  66 & 1.00 \\ 
  86 &  44 &  66 & 1.00 \\ 
  87 &  65 &  66 & 1.00 \\ 
  88 &   8 &  67 & 1.00 \\ 
  89 &  21 &  67 & 1.00 \\ 
  90 &  65 &  67 & 1.00 \\ 
  91 &  42 &  68 & 1.00 \\ 
  92 &   1 &  69 & 1.00 \\ 
  93 &  56 &  69 & 1.00 \\ 
  94 &  57 &  70 & 1.00 \\ 
  95 &  70 &  71 & 1.00 \\ 
  96 &  56 &  72 & 1.00 \\ 
  97 &  65 &  72 & 1.00 \\ 
  98 &   7 &  73 & 1.00 \\ 
  99 &  31 &  73 & 1.00 \\ 
  100 &  64 &  73 & 1.00 \\ 
  101 &  31 &  74 & 1.00 \\ 
  102 &  73 &  74 & 1.00 \\ 
  103 &   3 &  75 & 1.00 \\ 
  104 &  22 &  75 & 1.00 \\ 
  105 &  59 &  75 & 0.78 \\ 
  106 &  34 &  76 & 1.00 \\ 
  107 &   9 &  77 & 1.00 \\ 
  108 &  17 &  77 & 1.00 \\ 
  109 &  18 &  77 & 1.00 \\ 
  110 &  26 &  77 & 1.00 \\ 
  111 &  75 &  77 & 1.00 \\
  & & & \\
   \hline
        \end{tabular}
        \end{minipage} 
    }
    \label{OV:Omega:incprob}
\end{table}

\begin{table}[H]
\centering
    \caption{Indices of genes and proteins for READ cancer data. The first column lists the components of the dataset mRNA(genes) and the second column lists the components of the dataset RPPA(proteins).}
    \resizebox{0.8\textwidth}{!}{%
    \begin{minipage}{.5\linewidth}
      
      \centering
        \begin{tabular}{rll}
            \hline
 & Gene & Protein \\ 
  \hline
1 & BAK1 & BAK \\ 
  2 & BAX & BAX \\ 
  3 & BID & BID \\ 
  4 & BCL2L11 & BIM \\ 
  5 & CASP7 & CASPASE7CLEAVEDD198 \\ 
  6 & BAD & BADPS112 \\ 
  7 & BCL2 & BCL2 \\ 
  8 & BCL2L1 & BCLXL \\ 
  9 & BIRC2 & CIAP \\ 
  10 & CDK1 & CDK1 \\ 
  11 & CCNB1 & CYCLINB1 \\ 
  12 & CCNE1 & CYCLINE1 \\ 
  13 & CCNE2 & CYCLINE2 \\ 
  14 & CDKN1B & P27PT157 \\ 
  15 & PCNA & P27PT198 \\ 
  16 & FOXM1 & PCNA \\ 
  17 & TP53BP1 & FOXM1 \\ 
  18 & ATM & 53BP1 \\ 
  19 & BRCA2 & ATM \\ 
  20 & CHEK1 & CHK1PS345 \\ 
  21 & CHEK2 & CHK2PT68 \\ 
  22 & XRCC5 & KU80 \\ 
  23 & MRE11A & MRE11 \\ 
  24 & TP53 & P53 \\ 
  25 & RAD50 & RAD50 \\ 
  26 & RAD51 & RAD51 \\ 
  27 & XRCC1 & XRCC1 \\ 
  28 & FN1 & FIBRONECTIN \\ 
  29 & CDH2 & NCADHERIN \\ 
  30 & COL6A1 & COLLAGENVI \\ 
  31 & CLDN7 & CLAUDIN7 \\ 
  32 & CDH1 & ECADHERIN \\ 
  33 & CTNNB1 & BETACATENIN \\ 
  34 & SERPINE1 & PAI1 \\ 
  35 & ESR1 & ERALPHA \\ 
  36 & PGR & ERALPHAPS118 \\ 
  37 & AR & PR \\ 
  38 & INPP4B & AR \\
   \hline
        \end{tabular}
    \end{minipage}%
    \quad \quad
    \begin{minipage}{.5\linewidth}
      \centering
         \begin{tabular}{rll}
            \hline
 & Gene & Protein \\ 
  \hline
 39 & GATA3 & INPP4B \\ 
  40 & AKT1 & GATA3 \\ 
  41 & AKT2 & AKTPS473 \\ 
  42 & AKT3 & AKTPT308 \\ 
  43 & GSK3A & GSK3ALPHABETAPS21S9 \\ 
  44 & GSK3B & GSK3PS9 \\ 
  45 & AKT1S1 & PRAS40PT246 \\ 
  46 & TSC2 & TUBERINPT1462 \\ 
  47 & PTEN & PTEN \\ 
  48 & ARAF & ARAFPS299 \\ 
  49 & JUN & CJUNPS73 \\ 
  50 & RAF1 & CRAFPS338 \\ 
  51 & MAPK8 & JNKPT183Y185 \\ 
  52 & MAPK1 & MAPKPT202Y204 \\ 
  53 & MAPK3 & MEK1PS217S221 \\ 
  54 & MAP2K1 & P38PT180Y182 \\ 
  55 & MAPK14 & P90RSKPT359S363 \\ 
  56 & RPS6KA1 & YB1PS102 \\ 
  57 & YBX1 & EGFRPY1068 \\ 
  58 & EGFR & EGFRPY1173 \\ 
  59 & ERBB2 & HER2PY1248 \\ 
  60 & ERBB3 & HER3PY1298 \\ 
  61 & SHC1 & SHCPY317 \\ 
  62 & SRC & SRCPY416 \\ 
  63 & EIF4EBP1 & SRCPY527 \\ 
  64 & RPS6KB1 & 4EBP1PS65 \\ 
  65 & MTOR & 4EBP1PT37T46 \\ 
  66 & RPS6 & 4EBP1PT70 \\ 
  67 & RB1 & P70S6KPT389 \\ 
  68 & CAV1 & MTORPS2448 \\ 
  69 & MYH11 & S6PS235S236 \\ 
  70 & RAB11A & S6PS240S244 \\ 
  71 & RAB11B & RBPS807S811 \\ 
  72 & GAPDH & CAVEOLIN1 \\ 
  73 & RBM15 & MYH11 \\ 
  74 &  & RAB11 \\ 
  75 &  & GAPDH \\ 
  76 &  & RBM15 \\ 
  \hline
        \end{tabular}
        \end{minipage} 
    }
    \label{tab:READ:geneIndices}
\end{table}

\begin{figure}[H]
    \centering
    \includegraphics[width = \linewidth]{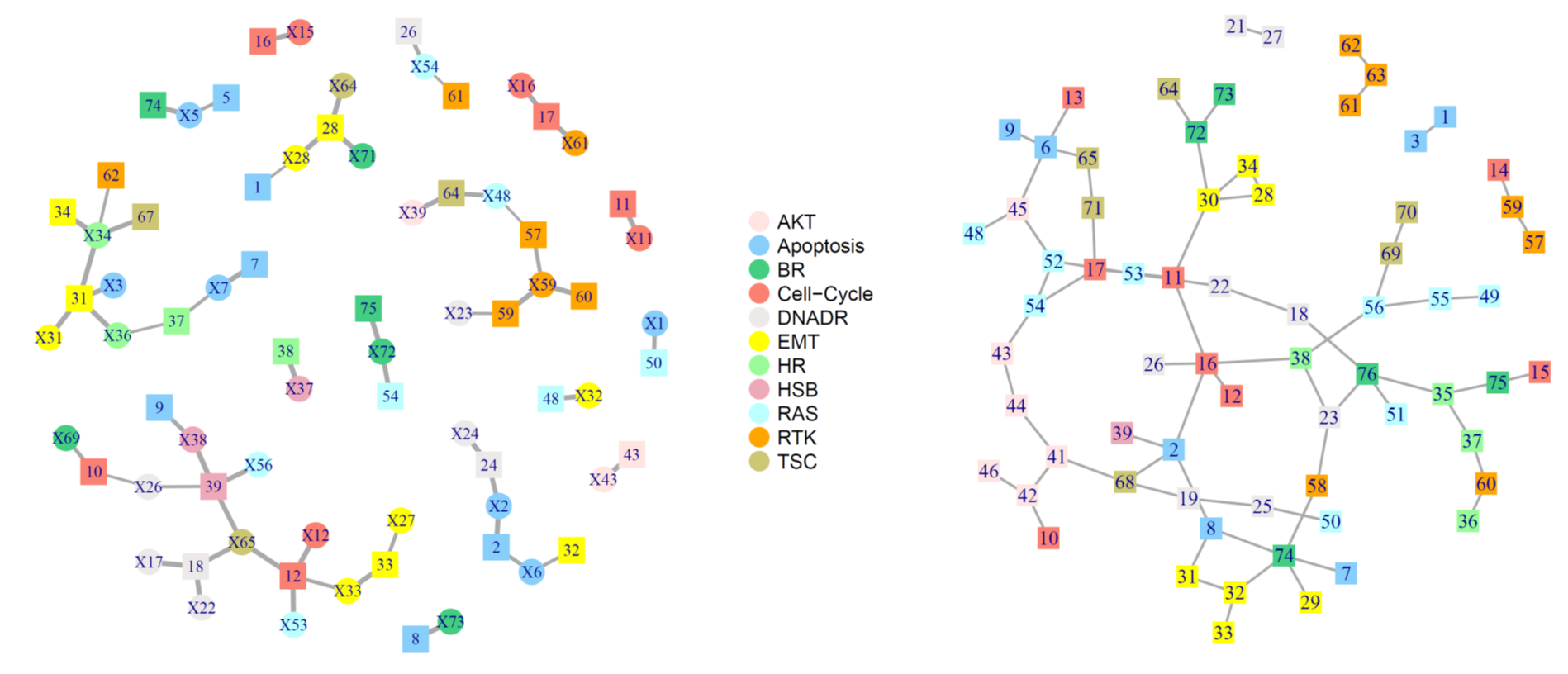}
    \caption{READ networks with 0.5 as the inclusion probability cutoff. The circles represent genes and the squares represent proteins. The different colors represent the different pathways listed in the Table 14 in Appendix. Left : Network graph indicating associations between mRNA and protein. Right : Network graph indicating associations among proteins. The inclusion probabilities are listed in Tables \ref{READ:B:incprob} and \ref{READ:Omega:incprob} . \textit{All the edge widths are proportional to the corresponding inclusion probabilities.}}
    \label{fig:READnetwork05}
\end{figure}

 \begin{table}[H]
\centering
    \caption{Inclusion probability of each edge for the READ network graph indicating associations between mRNA and proteins provided in the left panel of Figure \ref{fig:READnetwork05}.}
    \resizebox{0.9\textwidth}{!}{%
    \begin{minipage}{.4\linewidth}
      
      \centering
        \begin{tabular}{cccc}
            \hline
 & Gene  & Protein  & Inclusion   \\ 
 & & & Probability \\
  \hline
   1 & X28 &   1 & 0.54 \\ 
  2 & X2 &   2 & 1.00 \\ 
  3 & X6 &   2 & 0.76 \\ 
  4 & X5 &   5 & 0.68 \\ 
  5 & X7 &   7 & 1.00 \\ 
  6 & X73 &   8 & 1.00 \\ 
  7 & X38 &   9 & 0.83 \\ 
  8 & X26 &  10 & 0.54 \\ 
  9 & X69 &  10 & 0.75 \\ 
  10 & X11 &  11 & 1.00 \\ 
  11 & X12 &  12 & 1.00 \\ 
  12 & X33 &  12 & 0.72 \\ 
  13 & X53 &  12 & 1.00 \\ 
  14 & X65 &  12 & 0.96 \\ 
  15 & X15 &  16 & 1.00 \\
   \hline
        \end{tabular}
    \end{minipage}%
    \quad \quad
    \begin{minipage}{.4\linewidth}
      \centering
         \begin{tabular}{cccc}
            \hline
 & Gene  & Protein  & Inclusion   \\ 
 & & & Probability\\
  \hline
 16 & X16 &  17 & 0.97 \\ 
  17 & X61 &  17 & 0.98 \\ 
  18 & X17 &  18 & 1.00 \\ 
  19 & X22 &  18 & 0.98 \\ 
  20 & X65 &  18 & 0.94 \\ 
  21 & X2 &  24 & 0.91 \\ 
  22 & X24 &  24 & 1.00 \\ 
  23 & X54 &  26 & 0.82 \\ 
  24 & X28 &  28 & 1.00 \\ 
  25 & X64 &  28 & 1.00 \\ 
  26 & X71 &  28 & 1.00 \\ 
  27 & X3 &  31 & 0.97 \\ 
  28 & X31 &  31 & 1.00 \\ 
  29 & X34 &  31 & 0.99 \\ 
  30 & X36 &  31 & 0.99 \\ 
  \hline
  \end{tabular}
    \end{minipage}%
    \quad \quad
    \begin{minipage}{.4\linewidth}
      \centering
         \begin{tabular}{cccc}
            \hline
 & Gene  & Protein  & Inclusion   \\ 
 & & & Probability\\
  \hline
  31 & X6 &  32 & 0.67 \\ 
  32 & X27 &  33 & 0.73 \\ 
  33 & X33 &  33 & 1.00 \\ 
  34 & X34 &  34 & 1.00 \\ 
  35 & X7 &  37 & 0.65 \\ 
  36 & X36 &  37 & 0.57 \\ 
  37 & X37 &  38 & 1.00 \\ 
  38 & X26 &  39 & 0.63 \\ 
  39 & X38 &  39 & 1.00 \\ 
  40 & X56 &  39 & 0.94 \\ 
  41 & X65 &  39 & 0.93 \\ 
  42 & X43 &  43 & 1.00 \\ 
  43 & X32 &  48 & 0.97 \\ 
  44 & X1 &  50 & 0.84 \\ 
  45 & X72 &  54 & 0.94 \\ 
  \hline
  \end{tabular}
    \end{minipage}%
    \quad \quad
    \begin{minipage}{.4\linewidth}
      \centering
         \begin{tabular}{cccc}
            \hline
 & Gene  & Protein  & Inclusion  \\ 
 & & & Probability\\
  \hline
  46 & X48 &  57 & 0.52 \\ 
  47 & X59 &  57 & 0.78 \\ 
  48 & X23 &  59 & 0.60 \\ 
  49 & X59 &  59 & 1.00 \\ 
  50 & X59 &  60 & 1.00 \\ 
  51 & X54 &  61 & 0.52 \\ 
  52 & X34 &  62 & 0.61 \\ 
  53 & X39 &  64 & 0.97 \\ 
  54 & X48 &  64 & 0.88 \\ 
  55 & X34 &  67 & 0.92 \\ 
  56 & X5 &  74 & 0.97 \\ 
  57 & X72 &  75 & 1.00 \\
   & & & \\
   & & & \\
   & & & \\
   \hline
        \end{tabular}
        \end{minipage} 
    }
    \label{READ:B:incprob}
\end{table}
    
     \begin{table}[H]
\centering
    \caption{Inclusion probability of each edge for the READ network graph indicating associations among proteins provided in the right panel of Figure \ref{fig:READnetwork05}.}
    \resizebox{0.9\textwidth}{!}{%
    \begin{minipage}{.4\linewidth}
      
      \centering
        \begin{tabular}{cccc}
            \hline
 & Protein  & Protein  & Inclusion  \\ 
 & & & Probability \\
  \hline
   1 &   1 &   3 & 1.00 \\ 
  2 &   2 &   8 & 1.00 \\ 
  3 &   6 &   9 & 1.00 \\ 
  4 &   6 &  13 & 1.00 \\ 
  5 &   2 &  16 & 1.00 \\ 
  6 &  11 &  16 & 1.00 \\ 
  7 &  12 &  16 & 1.00 \\ 
  8 &  11 &  17 & 1.00 \\ 
  9 &  18 &  22 & 1.00 \\ 
  10 &  19 &  25 & 1.00 \\ 
  11 &  16 &  26 & 1.00 \\ 
  12 &  21 &  27 & 1.00 \\ 
  13 &  11 &  30 & 1.00 \\ 
  14 &  28 &  30 & 1.00 \\ 
  15 &   8 &  31 & 1.00 \\ 
  16 &  31 &  32 & 1.00 \\ 
  17 &  32 &  33 & 1.00 \\ 
  18 &  28 &  34 & 1.00 \\
   \hline
        \end{tabular}
    \end{minipage}%
    \quad \quad
    \begin{minipage}{.4\linewidth}
      \centering
         \begin{tabular}{cccc}
            \hline
 & Protein  & Protein  & Inclusion \\ 
 & & & Probability\\
  \hline
 19 &  30 &  34 & 1.00 \\ 
  20 &  35 &  37 & 1.00 \\ 
  21 &  16 &  38 & 1.00 \\ 
  22 &  23 &  38 & 1.00 \\ 
  23 &   2 &  39 & 1.00 \\ 
  24 &  10 &  42 & 1.00 \\ 
  25 &  41 &  42 & 1.00 \\ 
  26 &  41 &  44 & 1.00 \\ 
  27 &  43 &  44 & 1.00 \\ 
  28 &   6 &  45 & 1.00 \\ 
  29 &  42 &  46 & 1.00 \\ 
  30 &  45 &  48 & 1.00 \\ 
  31 &  25 &  50 & 1.00 \\ 
  32 &  45 &  52 & 1.00 \\ 
  33 &  22 &  53 & 1.00 \\ 
  34 &  52 &  53 & 1.00 \\ 
  35 &  17 &  54 & 1.00 \\ 
  36 &  43 &  54 & 1.00 \\ 
  \hline
  \end{tabular}
    \end{minipage}%
    \quad \quad
    \begin{minipage}{.4\linewidth}
      \centering
         \begin{tabular}{cccc}
            \hline
 & Protein  & Protein  & Inclusion  \\ 
 & & & Probability\\
  \hline
  37 &  52 &  54 & 1.00 \\ 
  38 &  49 &  55 & 1.00 \\ 
  39 &  38 &  56 & 1.00 \\ 
  40 &  55 &  56 & 1.00 \\ 
  41 &  23 &  58 & 1.00 \\ 
  42 &  14 &  59 & 1.00 \\ 
  43 &  57 &  59 & 1.00 \\ 
  44 &  36 &  60 & 1.00 \\ 
  45 &  37 &  60 & 1.00 \\ 
  46 &  61 &  63 & 1.00 \\ 
  47 &  62 &  63 & 1.00 \\ 
  48 &   6 &  65 & 1.00 \\ 
  49 &   2 &  68 & 1.00 \\ 
  50 &  19 &  68 & 1.00 \\ 
  51 &  41 &  68 & 1.00 \\ 
  52 &  56 &  69 & 1.00 \\ 
  53 &  69 &  70 & 1.00 \\ 
  54 &  17 &  71 & 1.00 \\
  \hline
  \end{tabular}
    \end{minipage}%
    \quad \quad
    \begin{minipage}{.4\linewidth}
      \centering
         \begin{tabular}{cccc}
            \hline
 & Protein  & Protein  & Inclusion \\ 
 & & & Probability\\
  \hline
  55 &  65 &  71 & 1.00 \\ 
  56 &  30 &  72 & 1.00 \\ 
  57 &  64 &  72 & 1.00 \\ 
  58 &  72 &  73 & 1.00 \\ 
  59 &   7 &  74 & 1.00 \\ 
  60 &   8 &  74 & 1.00 \\ 
  61 &  29 &  74 & 1.00 \\ 
  62 &  32 &  74 & 1.00 \\ 
  63 &  58 &  74 & 0.95 \\ 
  64 &  15 &  75 & 1.00 \\ 
  65 &  35 &  75 & 1.00 \\ 
  66 &  18 &  76 & 1.00 \\ 
  67 &  23 &  76 & 1.00 \\ 
  68 &  35 &  76 & 1.00 \\ 
  69 &  51 &  76 & 1.00 \\ 
  & & & \\
  & & & \\
  & & & \\
     \hline
        \end{tabular}
        \end{minipage} 
    }
    \label{READ:Omega:incprob}
\end{table}
    
    \begin{table}[H]
\centering
    \caption{Indices of genes and proteins for SKCM cancer data. The first column lists the components of the dataset mRNA(genes) and the second column lists the components of the dataset RPPA(proteins).}
    \resizebox{0.8\textwidth}{!}{%
    \begin{minipage}{.5\linewidth}
      
      \centering
        \begin{tabular}{rll}
            \hline
 & Gene & Protein \\ 
  \hline
   1 & BAK1 & BAK \\ 
  2 & BAX & BAX \\ 
  3 & BID & BID \\ 
  4 & BCL2L11 & BIM \\ 
  5 & CASP7 & CASPASE7CLEAVEDD198 \\ 
  6 & BAD & BADPS112 \\ 
  7 & BCL2 & BCL2 \\ 
  8 & BCL2L1 & BCLXL \\ 
  9 & BIRC2 & CIAP \\ 
  10 & CDK1 & CDK1 \\ 
  11 & CCNB1 & CYCLINB1 \\ 
  12 & CCNE1 & CYCLINE1 \\ 
  13 & CCNE2 & CYCLINE2 \\ 
  14 & CDKN1B & P27PT157 \\ 
  15 & PCNA & P27PT198 \\ 
  16 & FOXM1 & PCNA \\ 
  17 & TP53BP1 & FOXM1 \\ 
  18 & ATM & 53BP1 \\ 
  19 & BRCA2 & ATM \\ 
  20 & CHEK1 & CHK1PS345 \\ 
  21 & CHEK2 & CHK2PT68 \\ 
  22 & XRCC5 & KU80 \\ 
  23 & MRE11A & MRE11 \\ 
  24 & TP53 & P53 \\ 
  25 & RAD50 & RAD50 \\ 
  26 & RAD51 & RAD51 \\ 
  27 & XRCC1 & XRCC1 \\ 
  28 & FN1 & FIBRONECTIN \\ 
  29 & CDH2 & NCADHERIN \\ 
  30 & COL6A1 & COLLAGENVI \\ 
  31 & CLDN7 & CLAUDIN7 \\ 
  32 & CDH1 & ECADHERIN \\ 
  33 & CTNNB1 & BETACATENIN \\ 
  34 & SERPINE1 & PAI1 \\ 
  35 & ESR1 & ERALPHA \\ 
  36 & PGR & ERALPHAPS118 \\ 
  37 & AR & PR \\ 
  38 & INPP4B & AR \\
   \hline
        \end{tabular}
    \end{minipage}%
    \quad \quad
    \begin{minipage}{.5\linewidth}
      \centering
         \begin{tabular}{rll}
            \hline
 & Gene & Protein \\ 
  \hline
 39 & GATA3 & INPP4B \\ 
  40 & AKT1 & GATA3 \\ 
  41 & AKT2 & AKTPS473 \\ 
  42 & AKT3 & AKTPT308 \\ 
  43 & GSK3A & GSK3ALPHABETAPS21S9 \\ 
  44 & GSK3B & GSK3PS9 \\ 
  45 & AKT1S1 & PRAS40PT246 \\ 
  46 & TSC2 & TUBERINPT1462 \\ 
  47 & PTEN & PTEN \\ 
  48 & ARAF & ARAFPS299 \\ 
  49 & JUN & CJUNPS73 \\ 
  50 & RAF1 & CRAFPS338 \\ 
  51 & MAPK8 & JNKPT183Y185 \\ 
  52 & MAPK1 & MAPKPT202Y204 \\ 
  53 & MAPK3 & MEK1PS217S221 \\ 
  54 & MAP2K1 & P38PT180Y182 \\ 
  55 & MAPK14 & P90RSKPT359S363 \\ 
  56 & RPS6KA1 & YB1PS102 \\ 
  57 & YBX1 & EGFRPY1068 \\ 
  58 & EGFR & EGFRPY1173 \\ 
  59 & ERBB2 & HER2PY1248 \\ 
  60 & ERBB3 & HER3PY1298 \\ 
  61 & SHC1 & SHCPY317 \\ 
  62 & SRC & SRCPY416 \\ 
  63 & EIF4EBP1 & SRCPY527 \\ 
  64 & RPS6KB1 & 4EBP1PS65 \\ 
  65 & MTOR & 4EBP1PT37T46 \\ 
  66 & RPS6 & 4EBP1PT70 \\ 
  67 & RB1 & P70S6KPT389 \\ 
  68 & CAV1 & MTORPS2448 \\ 
  69 & MYH11 & S6PS235S236 \\ 
  70 & RAB11A & S6PS240S244 \\ 
  71 & RAB11B & RBPS807S811 \\ 
  72 & GAPDH & CAVEOLIN1 \\ 
  73 & RBM15 & MYH11 \\ 
  74 & 1 & RAB11 \\ 
  75 & 2 & GAPDH \\ 
  76 & 3 & RBM15 \\ 
  \hline
        \end{tabular}
        \end{minipage} 
    }
    \label{tab:SKCM:geneIndices}
\end{table}
    
    \begin{figure}[H]
    \centering
    \includegraphics[width = \linewidth]{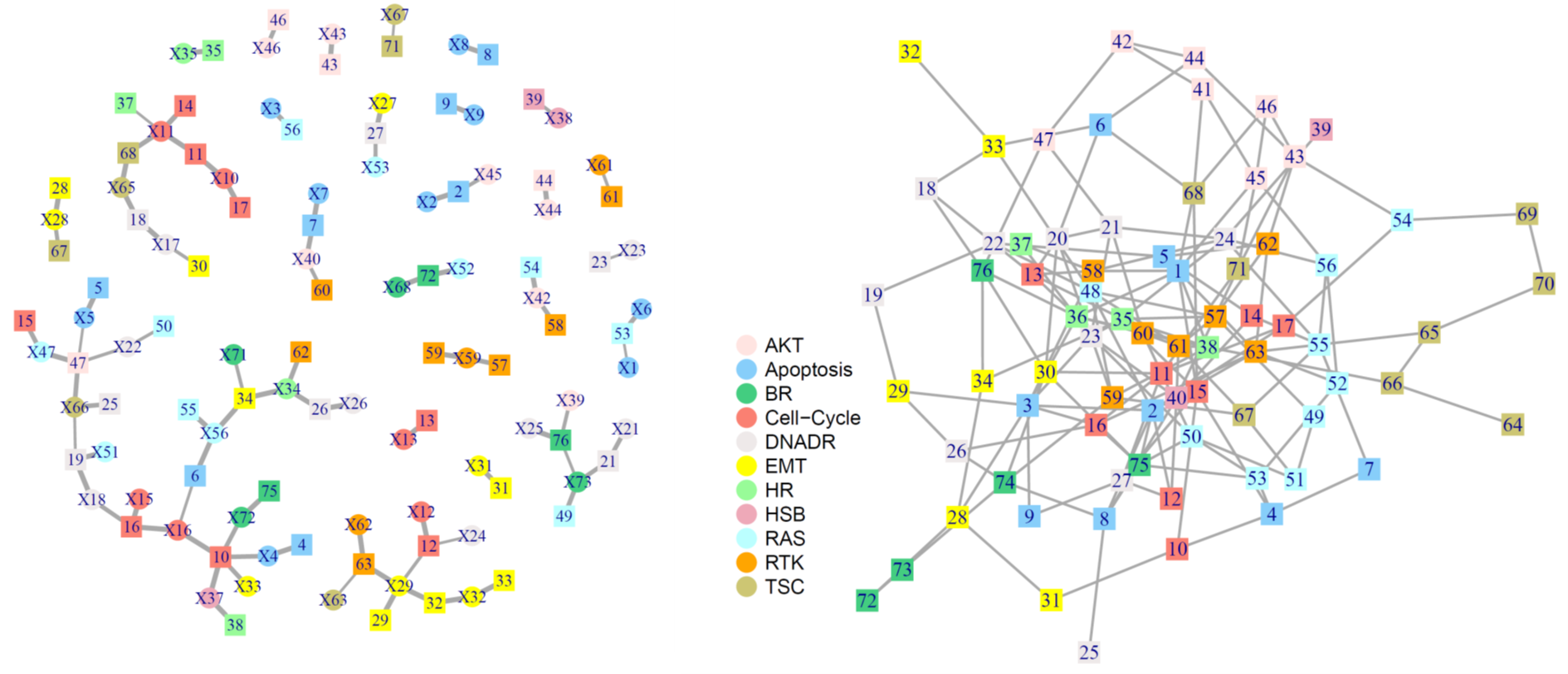}
    \caption{SKCM networks with 0.5 as the inclusion probability cutoff. The circles represent genes and the squares represent proteins. The different colors represent the different pathways listed in Table 14 in Appendix. Left : Network graph indicating associations between mRNA and protein. Right : Network graph indicating associations among proteins. The inclusion probabilities are listed in Table . \textit{All the edge widths are proportional to the corresponding inclusion probabilities.}}
    \label{fig:SKCMnetwork05}
\end{figure}

 \begin{table}[H]
\centering
    \caption{Inclusion probability of each edge for the SKCM network graph indicating associations between mRNA and proteins provided in the left panel of Figure \ref{fig:SKCMnetwork05}.}
    \resizebox{0.9\textwidth}{!}{%
    \begin{minipage}{.4\linewidth}
      
      \centering
        \begin{tabular}{cccc}
            \hline
 & Gene  & Protein  & Inclusion   \\ 
 & & & Probability \\
  \hline
1 & X2 &   2 & 1.00 \\ 
  2 & X45 &   2 & 0.73 \\ 
  3 & X4 &   4 & 1.00 \\ 
  4 & X5 &   5 & 1.00 \\ 
  5 & X16 &   6 & 0.56 \\ 
  6 & X56 &   6 & 0.97 \\ 
  7 & X7 &   7 & 1.00 \\ 
  8 & X40 &   7 & 0.96 \\ 
  9 & X8 &   8 & 1.00 \\ 
  10 & X9 &   9 & 1.00 \\ 
  11 & X4 &  10 & 0.89 \\ 
  12 & X16 &  10 & 0.98 \\ 
  13 & X33 &  10 & 0.88 \\ 
  14 & X37 &  10 & 1.00 \\ 
  15 & X72 &  10 & 0.86 \\ 
  16 & X10 &  11 & 1.00 \\ 
  17 & X11 &  11 & 1.00 \\ 
  18 & X12 &  12 & 1.00 \\ 
  19 & X24 &  12 & 0.53 \\ 
  20 & X29 &  12 & 0.65 \\ 
  21 & X13 &  13 & 1.00 \\ 
  22 & X11 &  14 & 1.00 \\ 
 
   \hline
        \end{tabular}
    \end{minipage}%
    \quad \quad
    \begin{minipage}{.4\linewidth}
      \centering
         \begin{tabular}{cccc}
            \hline
 & Gene  & Protein  & Inclusion   \\ 
 & & & Probability\\
  \hline
 23 & X47 &  15 & 0.90 \\ 
  24 & X15 &  16 & 1.00 \\ 
  25 & X16 &  16 & 1.00 \\ 
  26 & X18 &  16 & 0.84 \\ 
  27 & X10 &  17 & 1.00 \\ 
  28 & X17 &  18 & 1.00 \\ 
  29 & X65 &  18 & 0.98 \\ 
  30 & X18 &  19 & 1.00 \\ 
  31 & X51 &  19 & 0.99 \\ 
  32 & X66 &  19 & 0.55 \\ 
  33 & X21 &  21 & 0.88 \\ 
  34 & X73 &  21 & 0.86 \\ 
  35 & X23 &  23 & 0.53 \\ 
  36 & X66 &  25 & 0.90 \\ 
  37 & X26 &  26 & 1.00 \\ 
  38 & X34 &  26 & 0.98 \\ 
  39 & X27 &  27 & 1.00 \\ 
  40 & X53 &  27 & 0.79 \\ 
  41 & X28 &  28 & 1.00 \\ 
  42 & X29 &  29 & 1.00 \\ 
  43 & X17 &  30 & 0.66 \\ 
  44 & X31 &  31 & 1.00 \\
  \hline
  \end{tabular}
    \end{minipage}%
    \quad \quad
    \begin{minipage}{.4\linewidth}
      \centering
         \begin{tabular}{cccc}
            \hline
 & Gene  & Protein  & Inclusion  \\ 
 & & & Probability\\
  \hline
  45 & X29 &  32 & 1.00 \\ 
  46 & X32 &  32 & 1.00 \\ 
  47 & X32 &  33 & 1.00 \\ 
  48 & X34 &  34 & 1.00 \\ 
  49 & X56 &  34 & 0.99 \\ 
  50 & X71 &  34 & 0.69 \\ 
  51 & X35 &  35 & 1.00 \\ 
  52 & X11 &  37 & 0.52 \\ 
  53 & X37 &  38 & 1.00 \\ 
  54 & X38 &  39 & 0.68 \\ 
  55 & X43 &  43 & 1.00 \\ 
  56 & X44 &  44 & 1.00 \\ 
  57 & X46 &  46 & 0.87 \\ 
  58 & X5 &  47 & 0.53 \\ 
  59 & X22 &  47 & 0.64 \\ 
  60 & X47 &  47 & 1.00 \\ 
  61 & X66 &  47 & 0.91 \\ 
  62 & X73 &  49 & 0.91 \\ 
  63 & X22 &  50 & 0.62 \\ 
  64 & X1 &  53 & 0.52 \\ 
  65 & X6 &  53 & 0.93 \\ 
  66 & X42 &  54 & 0.98 \\ 
  \hline
  \end{tabular}
    \end{minipage}%
    \quad \quad
    \begin{minipage}{.4\linewidth}
      \centering
         \begin{tabular}{cccc}
            \hline
 & Gene  & Protein  & Inclusion   \\ 
 & & & Probability\\
  \hline
  67 & X56 &  55 & 1.00 \\ 
  68 & X3 &  56 & 0.98 \\ 
  69 & X59 &  57 & 0.92 \\ 
  70 & X42 &  58 & 0.90 \\ 
  71 & X59 &  59 & 1.00 \\ 
  72 & X40 &  60 & 0.73 \\ 
  73 & X61 &  61 & 0.63 \\ 
  74 & X34 &  62 & 0.91 \\ 
  75 & X29 &  63 & 0.99 \\ 
  76 & X62 &  63 & 1.00 \\ 
  77 & X63 &  63 & 0.52 \\ 
  78 & X28 &  67 & 0.79 \\ 
  79 & X11 &  68 & 0.93 \\ 
  80 & X65 &  68 & 1.00 \\ 
  81 & X67 &  71 & 0.51 \\ 
  82 & X52 &  72 & 1.00 \\ 
  83 & X68 &  72 & 1.00 \\ 
  84 & X72 &  75 & 1.00 \\ 
  85 & X25 &  76 & 1.00 \\ 
  86 & X39 &  76 & 0.72 \\ 
  87 & X73 &  76 & 0.58 \\ 
  & & & \\
   \hline
        \end{tabular}
        \end{minipage} 
    }
    \label{SKCM:B:incprob}
\end{table}
    
     \begin{table}[H]
\centering
    \caption{Inclusion probability of each edge for the SKCM network graph indicating associations among proteins provided in the right panel of Figure \ref{fig:SKCMnetwork05}.}
    \resizebox{0.9\textwidth}{!}{%
    \begin{minipage}{.4\linewidth}
      
      \centering
        \begin{tabular}{cccc}
            \hline
 & Protein  & Protein  & Inclusion   \\ 
 & & & Probability \\
  \hline
1 &   2 &   3 & 1.00 \\ 
  2 &   4 &   7 & 1.00 \\ 
  3 &   2 &   8 & 1.00 \\ 
  4 &   3 &   9 & 1.00 \\ 
  5 &   4 &  10 & 1.00 \\ 
  6 &   2 &  12 & 1.00 \\ 
  7 &   5 &  14 & 1.00 \\ 
  8 &   1 &  15 & 1.00 \\ 
  9 &   4 &  15 & 1.00 \\ 
  10 &  10 &  15 & 1.00 \\ 
  11 &   3 &  16 & 1.00 \\ 
  12 &  11 &  16 & 1.00 \\ 
  13 &   2 &  17 & 1.00 \\ 
  14 &  11 &  17 & 1.00 \\ 
  15 &   6 &  20 & 1.00 \\ 
  16 &  13 &  20 & 1.00 \\ 
  17 &  20 &  21 & 1.00 \\ 
  18 &   5 &  22 & 1.00 \\ 
  19 &  13 &  22 & 1.00 \\ 
  20 &  18 &  22 & 1.00 \\ 
  21 &  19 &  22 & 1.00 \\ 
  22 &   3 &  23 & 1.00 \\ 
  23 &  20 &  23 & 1.00 \\ 
  24 &  21 &  23 & 1.00 \\ 
  25 &  21 &  24 & 1.00 \\ 
  26 &   8 &  25 & 1.00 \\ 
  27 &  16 &  26 & 1.00 \\ 
  28 &   9 &  27 & 1.00 \\ 
  29 &  11 &  27 & 1.00 \\ 
  30 &  12 &  27 & 1.00 \\ 
  31 &  15 &  27 & 1.00 \\ 
  32 &   3 &  28 & 1.00 \\ 
  33 &   3 &  29 & 1.00 \\ 
  34 &  19 &  29 & 1.00 \\ 
  35 &  26 &  29 & 1.00 \\ 
  36 &  11 &  30 & 1.00 \\ 
  37 &  20 &  30 & 1.00 \\ 
  38 &  10 &  31 & 1.00 \\ 
  39 &  28 &  31 & 1.00 \\
 
   \hline
        \end{tabular}
    \end{minipage}%
    \quad \quad
    \begin{minipage}{.4\linewidth}
      \centering
         \begin{tabular}{cccc}
            \hline
 & Protein  & Protein  & Inclusion  \\ 
 & & & Probability\\
  \hline
 40 &   6 &  33 & 1.00 \\ 
  41 &  18 &  33 & 1.00 \\ 
  42 &  32 &  33 & 1.00 \\ 
  43 &  28 &  34 & 1.00 \\ 
  44 &  13 &  35 & 1.00 \\ 
  45 &  15 &  35 & 1.00 \\ 
  46 &  23 &  35 & 1.00 \\ 
  47 &  24 &  35 & 1.00 \\ 
  48 &  34 &  35 & 0.76 \\ 
  49 &   1 &  36 & 1.00 \\ 
  50 &   3 &  36 & 1.00 \\ 
  51 &  13 &  36 & 1.00 \\ 
  52 &  23 &  37 & 1.00 \\ 
  53 &   1 &  38 & 1.00 \\ 
  54 &  35 &  38 & 1.00 \\ 
  55 &  36 &  38 & 1.00 \\ 
  56 &  24 &  39 & 1.00 \\ 
  57 &   8 &  40 & 1.00 \\ 
  58 &  12 &  40 & 1.00 \\ 
  59 &  14 &  40 & 1.00 \\ 
  60 &  23 &  40 & 1.00 \\ 
  61 &  41 &  42 & 1.00 \\ 
  62 &   6 &  44 & 1.00 \\ 
  63 &  42 &  44 & 1.00 \\ 
  64 &  43 &  44 & 1.00 \\ 
  65 &   1 &  45 & 1.00 \\ 
  66 &  41 &  45 & 1.00 \\ 
  67 &  43 &  46 & 1.00 \\ 
  68 &  45 &  46 & 1.00 \\ 
  69 &  21 &  47 & 1.00 \\ 
  70 &  42 &  47 & 1.00 \\ 
  71 &   2 &  48 & 1.00 \\ 
  72 &  22 &  48 & 1.00 \\ 
  73 &  33 &  48 & 1.00 \\ 
  74 &  23 &  50 & 1.00 \\ 
  75 &  49 &  51 & 1.00 \\ 
  76 &  50 &  51 & 1.00 \\ 
  77 &   4 &  53 & 1.00 \\ 
  78 &  50 &  53 & 1.00 \\ 
  \hline
  \end{tabular}
    \end{minipage}%
    \quad \quad
    \begin{minipage}{.4\linewidth}
      \centering
         \begin{tabular}{cccc}
            \hline
 & Protein  & Protein  & Inclusion   \\ 
 & & & Probability\\
  \hline
  79 &  52 &  53 & 1.00 \\ 
  80 &  17 &  54 & 1.00 \\ 
  81 &  43 &  54 & 1.00 \\ 
  82 &   7 &  55 & 1.00 \\ 
  83 &  14 &  55 & 1.00 \\ 
  84 &  24 &  55 & 1.00 \\ 
  85 &  24 &  56 & 1.00 \\ 
  86 &  45 &  56 & 1.00 \\ 
  87 &  52 &  56 & 1.00 \\ 
  88 &  55 &  56 & 1.00 \\ 
  89 &  35 &  57 & 1.00 \\ 
  90 &  48 &  57 & 1.00 \\ 
  91 &  49 &  57 & 1.00 \\ 
  92 &   1 &  58 & 1.00 \\ 
  93 &   3 &  58 & 1.00 \\ 
  94 &  13 &  58 & 1.00 \\ 
  95 &  24 &  58 & 1.00 \\ 
  96 &  37 &  58 & 1.00 \\ 
  97 &  23 &  59 & 1.00 \\ 
  98 &  38 &  59 & 1.00 \\ 
  99 &  48 &  59 & 1.00 \\ 
  100 &  21 &  60 & 1.00 \\ 
  101 &  36 &  60 & 1.00 \\ 
  102 &  38 &  60 & 1.00 \\ 
  103 &  50 &  60 & 1.00 \\ 
  104 &  38 &  61 & 1.00 \\ 
  105 &  40 &  61 & 1.00 \\ 
  106 &  48 &  61 & 1.00 \\ 
  107 &  57 &  61 & 1.00 \\ 
  108 &  60 &  61 & 1.00 \\ 
  109 &   5 &  62 & 1.00 \\ 
  110 &  43 &  62 & 1.00 \\ 
  111 &  57 &  62 & 1.00 \\ 
  112 &   2 &  63 & 1.00 \\ 
  113 &   5 &  63 & 1.00 \\ 
  114 &  16 &  63 & 1.00 \\ 
  115 &  52 &  63 & 1.00 \\ 
  116 &  62 &  63 & 1.00 \\ 
  117 &  63 &  65 & 1.00 \\ 
  \hline
  \end{tabular}
    \end{minipage}%
    \quad \quad
    \begin{minipage}{.4\linewidth}
      \centering
         \begin{tabular}{cccc}
            \hline
 & Protein  & Protein  & Inclusion \\ 
 & & & Probability\\
  \hline
  118 &  38 &  66 & 0.70 \\ 
  119 &  64 &  66 & 1.00 \\ 
  120 &  65 &  66 & 1.00 \\ 
  121 &  51 &  67 & 1.00 \\ 
  122 &  55 &  67 & 1.00 \\ 
  123 &  59 &  67 & 1.00 \\ 
  124 &  61 &  67 & 1.00 \\ 
  125 &   1 &  68 & 1.00 \\ 
  126 &   6 &  68 & 1.00 \\ 
  127 &  15 &  68 & 1.00 \\ 
  128 &  41 &  68 & 1.00 \\ 
  129 &  46 &  68 & 1.00 \\ 
  130 &  54 &  69 & 1.00 \\ 
  131 &  65 &  70 & 1.00 \\ 
  132 &  69 &  70 & 1.00 \\ 
  133 &  11 &  71 & 1.00 \\ 
  134 &  40 &  71 & 1.00 \\ 
  135 &  43 &  71 & 1.00 \\ 
  136 &  45 &  71 & 1.00 \\ 
  137 &  28 &  72 & 1.00 \\ 
  138 &  72 &  73 & 1.00 \\ 
  139 &   8 &  74 & 1.00 \\ 
  140 &  26 &  74 & 1.00 \\ 
  141 &  30 &  74 & 1.00 \\ 
  142 &  59 &  74 & 1.00 \\ 
  143 &  73 &  74 & 1.00 \\ 
  144 &   2 &  75 & 1.00 \\ 
  145 &   8 &  75 & 1.00 \\ 
  146 &  16 &  75 & 1.00 \\ 
  147 &  30 &  75 & 1.00 \\ 
  148 &  50 &  75 & 1.00 \\ 
  149 &  53 &  75 & 1.00 \\ 
  150 &  18 &  76 & 1.00 \\ 
  151 &  30 &  76 & 1.00 \\ 
  152 &  34 &  76 & 1.00 \\ 
  153 &  36 &  76 & 1.00 \\ 
  154 &  47 &  76 & 1.00 \\ 
  & & & \\
  & & & \\
   \hline
        \end{tabular}
        \end{minipage} 
    }
    \label{SKCM:Omega:incprob}
\end{table}

    \begin{table}[H]
\centering
    \caption{Indices of genes and proteins for UCEC cancer data. The first column lists the components of the dataset mRNA(genes) and the second column lists the components of the dataset RPPA(proteins).}
    \resizebox{0.8\textwidth}{!}{%
    \begin{minipage}{.5\linewidth}
      
      \centering
        \begin{tabular}{rll}
            \hline
 & Gene & Protein \\ 
  \hline
   1 & BAK1 & BAK \\ 
  2 & BAX & BAX \\ 
  3 & BID & BID \\ 
  4 & BCL2L11 & BIM \\ 
  5 & CASP7 & CASPASE7CLEAVEDD198 \\ 
  6 & BAD & BADPS112 \\ 
  7 & BCL2 & BCL2 \\ 
  8 & BCL2L1 & BCLXL \\ 
  9 & BIRC2 & CIAP \\ 
  10 & CDK1 & CDK1 \\ 
  11 & CCNB1 & CYCLINB1 \\ 
  12 & CCNE1 & CYCLINE1 \\ 
  13 & CCNE2 & CYCLINE2 \\ 
  14 & CDKN1B & P27PT157 \\ 
  15 & PCNA & P27PT198 \\ 
  16 & FOXM1 & PCNA \\ 
  17 & TP53BP1 & FOXM1 \\ 
  18 & ATM & 53BP1 \\ 
  19 & BRCA2 & ATM \\ 
  20 & CHEK1 & BRCA2 \\ 
  21 & CHEK2 & CHK1PS345 \\ 
  22 & XRCC5 & CHK2PT68 \\ 
  23 & MRE11A & KU80 \\ 
  24 & TP53 & MRE11 \\ 
  25 & RAD50 & P53 \\ 
  26 & RAD51 & RAD50 \\ 
  27 & XRCC1 & RAD51 \\ 
  28 & FN1 & XRCC1 \\ 
  29 & CDH2 & FIBRONECTIN \\ 
  30 & COL6A1 & NCADHERIN \\ 
  31 & CLDN7 & COLLAGENVI \\ 
  32 & CDH1 & CLAUDIN7 \\ 
  33 & CTNNB1 & ECADHERIN \\ 
  34 & SERPINE1 & BETACATENIN \\ 
  35 & ESR1 & PAI1 \\ 
  36 & PGR & ERALPHA \\ 
  37 & AR & ERALPHAPS118 \\ 
  38 & INPP4B & PR \\ 
  39 & GATA3 & AR \\ 
   \hline
        \end{tabular}
    \end{minipage}%
    \quad \quad
    \begin{minipage}{.5\linewidth}
      \centering
         \begin{tabular}{rll}
            \hline
 & Gene & Protein \\ 
  \hline
 40 & AKT1 & INPP4B \\ 
  41 & AKT2 & GATA3 \\ 
  42 & AKT3 & AKTPS473 \\ 
  43 & GSK3A & AKTPT308 \\ 
  44 & GSK3B & GSK3ALPHABETAPS21S9 \\ 
  45 & AKT1S1 & GSK3PS9 \\ 
  46 & TSC2 & PRAS40PT246 \\ 
  47 & PTEN & TUBERINPT1462 \\ 
  48 & ARAF & PTEN \\ 
  49 & JUN & ARAFPS299 \\ 
  50 & RAF1 & CJUNPS73 \\ 
  51 & MAPK8 & CRAFPS338 \\ 
  52 & MAPK1 & JNKPT183Y185 \\ 
  53 & MAPK3 & MAPKPT202Y204 \\ 
  54 & MAP2K1 & MEK1PS217S221 \\ 
  55 & MAPK14 & P38PT180Y182 \\ 
  56 & RPS6KA1 & P90RSKPT359S363 \\ 
  57 & YBX1 & YB1PS102 \\ 
  58 & EGFR & EGFRPY1068 \\ 
  59 & ERBB2 & EGFRPY1173 \\ 
  60 & ERBB3 & HER2PY1248 \\ 
  61 & SHC1 & HER3PY1298 \\ 
  62 & SRC & SHCPY317 \\ 
  63 & EIF4EBP1 & SRCPY416 \\ 
  64 & RPS6KB1 & SRCPY527 \\ 
  65 & MTOR & 4EBP1PS65 \\ 
  66 & RPS6 & 4EBP1PT37T46 \\ 
  67 & RB1 & 4EBP1PT70 \\ 
  68 & CAV1 & P70S6KPT389 \\ 
  69 & MYH11 & MTORPS2448 \\ 
  70 & RAB11A & S6PS235S236 \\ 
  71 & RAB11B & S6PS240S244 \\ 
  72 & GAPDH & RBPS807S811 \\ 
  73 & RBM15 & CAVEOLIN1 \\ 
  74 &  & MYH11 \\ 
  75 &  & RAB11 \\ 
  76 &  & GAPDH \\ 
  77 &  & RBM15 \\ 
  \hline
        \end{tabular}
        \end{minipage} 
    }
    \label{tab:UCEC:geneIndices}
\end{table}

\begin{figure}[H]
    \centering
    \includegraphics[width = \linewidth]{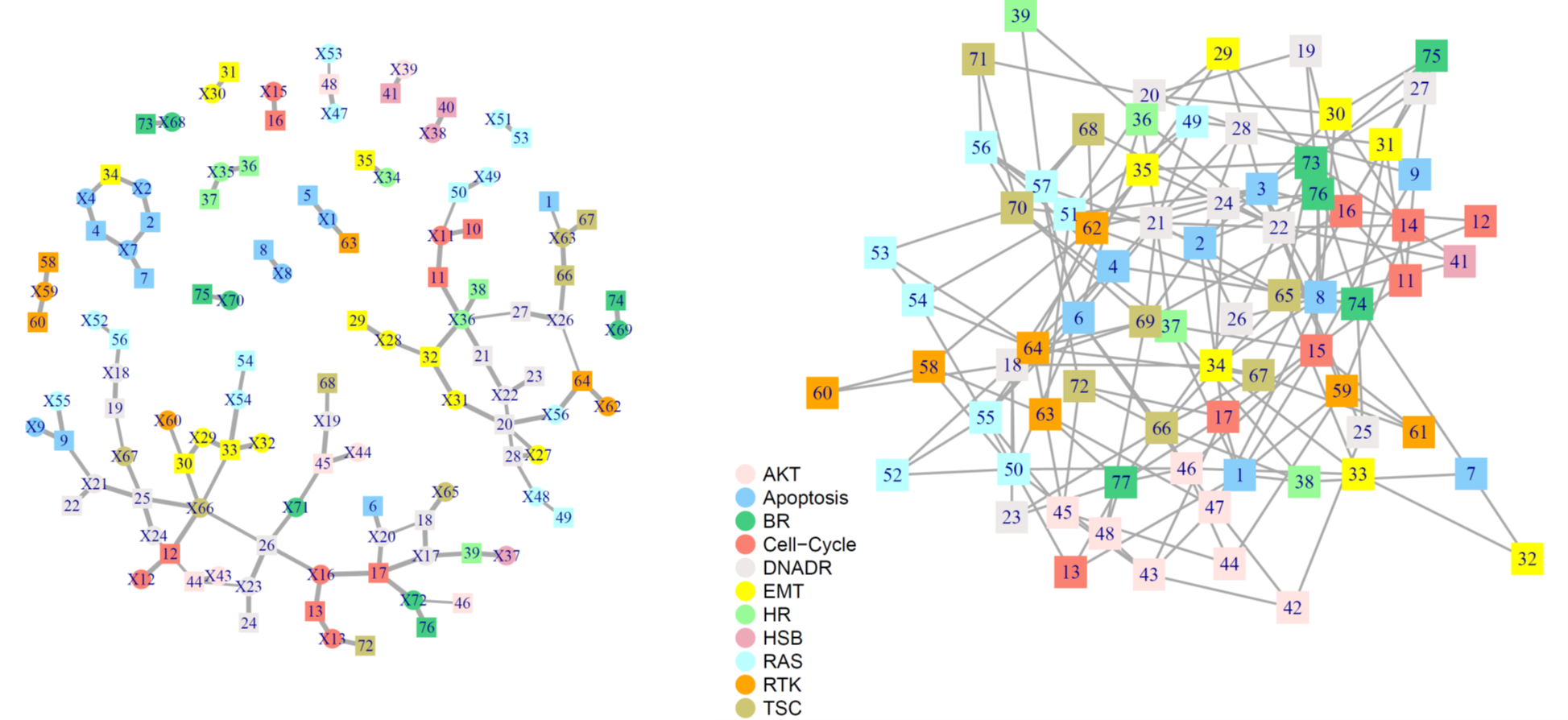}
    \caption{UCEC networks with 0.5 as the inclusion probability cutoff. The circles represent genes and the squares represent proteins. The different colors represent the different pathways listed in Table 14 in Appendix. Left : Network graph indicating associations between mRNA and protein. Right : Network graph indicating associations among proteins. The inclusion probabilities are listed in Table . \textit{All the edge widths are proportional to the corresponding inclusion probabilities.}}
    \label{fig:UCECnetwork05}
\end{figure}

\begin{table}[H]
\centering
    \caption{Inclusion probability of each edge for the UCEC network graph indicating associations between mRNA and proteins provided in the left panel of Figure \ref{fig:UCECnetwork05}.}
    \resizebox{0.9\textwidth}{!}{%
    \begin{minipage}{.4\linewidth}
      
      \centering
        \begin{tabular}{cccc}
            \hline
 & Gene  & Protein  & Inclusion   \\ 
 & & & Probability \\
  \hline
1 & X63 &   1 & 0.63 \\ 
  2 & X2 &   2 & 1.00 \\ 
  3 & X7 &   2 & 0.96 \\ 
  4 & X4 &   4 & 1.00 \\ 
  5 & X7 &   4 & 1.00 \\ 
  6 & X1 &   5 & 1.00 \\ 
  7 & X20 &   6 & 0.94 \\ 
  8 & X7 &   7 & 1.00 \\ 
  9 & X8 &   8 & 1.00 \\ 
  10 & X9 &   9 & 0.98 \\ 
  11 & X21 &   9 & 0.77 \\ 
  12 & X55 &   9 & 0.96 \\ 
  13 & X11 &  10 & 0.96 \\ 
  14 & X11 &  11 & 1.00 \\ 
  15 & X36 &  11 & 0.96 \\ 
  16 & X12 &  12 & 1.00 \\ 
  17 & X66 &  12 & 0.98 \\ 
  18 & X13 &  13 & 1.00 \\ 
  19 & X16 &  13 & 1.00 \\ 
  20 & X15 &  16 & 1.00 \\ 
  21 & X16 &  17 & 1.00 \\ 
  22 & X17 &  17 & 1.00 \\
  23 & X20 &  17 & 0.95 \\ 
  24 & X72 &  17 & 0.98 \\ 
  25 & X17 &  18 & 1.00 \\ 
  26 & X20 &  18 & 0.58 \\
    \hline
        \end{tabular}
    \end{minipage}%
    \quad \quad
    \begin{minipage}{.4\linewidth}
      \centering
         \begin{tabular}{cccc}
            \hline
 & Gene  & Protein  & Inclusion   \\ 
 & & & Probability\\
  \hline
 27 & X65 &  18 & 1.00 \\ 
  28 & X18 &  19 & 1.00 \\ 
  29 & X67 &  19 & 0.65 \\ 
  30 & X27 &  20 & 0.80 \\ 
  31 & X31 &  20 & 1.00 \\ 
  32 & X56 &  20 & 0.97 \\ 
  33 & X22 &  21 & 0.95 \\ 
  34 & X36 &  21 & 0.97 \\ 
  35 & X21 &  22 & 1.00 \\ 
  36 & X22 &  23 & 0.96 \\ 
  37 & X23 &  24 & 1.00 \\ 
  38 & X21 &  25 & 0.97 \\ 
  39 & X24 &  25 & 1.00 \\ 
  40 & X66 &  25 & 0.92 \\ 
  41 & X67 &  25 & 0.84 \\ 
  42 & X16 &  26 & 0.81 \\ 
  43 & X23 &  26 & 1.00 \\ 
  44 & X66 &  26 & 0.67 \\ 
  45 & X71 &  26 & 0.85 \\ 
  46 & X26 &  27 & 1.00 \\ 
  47 & X36 &  27 & 0.55 \\ 
  48 & X22 &  28 & 0.72 \\ 
  49 & X27 &  28 & 1.00 \\ 
  50 & X48 &  28 & 0.69 \\ 
  51 & X28 &  29 & 1.00 \\ 
  52 & X29 &  30 & 1.00 \\
  \hline
  \end{tabular}
    \end{minipage}%
    \quad \quad
    \begin{minipage}{.4\linewidth}
      \centering
         \begin{tabular}{cccc}
            \hline
 & Gene  & Protein  & Inclusion  \\ 
 & & & Probability\\
  \hline
  53 & X60 &  30 & 0.76 \\ 
  54 & X66 &  30 & 0.82 \\ 
  55 & X30 &  31 & 1.00 \\ 
  56 & X28 &  32 & 0.85 \\ 
  57 & X31 &  32 & 1.00 \\ 
  58 & X36 &  32 & 0.99 \\ 
  59 & X29 &  33 & 1.00 \\ 
  60 & X32 &  33 & 1.00 \\ 
  61 & X54 &  33 & 1.00 \\ 
  62 & X66 &  33 & 0.72 \\ 
  63 & X2 &  34 & 0.92 \\ 
  64 & X4 &  34 & 0.92 \\ 
  65 & X34 &  35 & 1.00 \\ 
  66 & X35 &  36 & 1.00 \\ 
  67 & X35 &  37 & 1.00 \\ 
  68 & X36 &  38 & 1.00 \\ 
  69 & X17 &  39 & 0.91 \\ 
  70 & X37 &  39 & 1.00 \\ 
  71 & X38 &  40 & 1.00 \\ 
  72 & X39 &  41 & 1.00 \\ 
  73 & X23 &  44 & 0.52 \\ 
  74 & X24 &  44 & 0.58 \\ 
  75 & X43 &  44 & 1.00 \\ 
  76 & X19 &  45 & 1.00 \\ 
  77 & X44 &  45 & 1.00 \\ 
  78 & X71 &  45 & 0.82 \\
  \hline
  \end{tabular}
    \end{minipage}%
    \quad \quad
    \begin{minipage}{.4\linewidth}
      \centering
         \begin{tabular}{cccc}
            \hline
 & Gene  & Protein  & Inclusion   \\ 
 & & & Probability\\
  \hline
   79 & X72 &  46 & 0.52 \\ 
  80 & X47 &  48 & 1.00 \\ 
  81 & X53 &  48 & 0.56 \\ 
  82 & X48 &  49 & 0.65 \\ 
  83 & X11 &  50 & 0.59 \\ 
  84 & X49 &  50 & 1.00 \\ 
  85 & X51 &  53 & 0.94 \\ 
  86 & X54 &  54 & 0.80 \\ 
  87 & X18 &  56 & 0.81 \\ 
  88 & X52 &  56 & 0.72 \\ 
  89 & X59 &  58 & 1.00 \\ 
  90 & X59 &  60 & 1.00 \\ 
  91 & X1 &  63 & 0.98 \\ 
  92 & X26 &  64 & 0.52 \\ 
  93 & X56 &  64 & 1.00 \\ 
  94 & X62 &  64 & 1.00 \\ 
  95 & X26 &  66 & 0.91 \\ 
  96 & X63 &  66 & 1.00 \\ 
  97 & X63 &  67 & 1.00 \\ 
  98 & X19 &  68 & 0.80 \\ 
  99 & X13 &  72 & 0.95 \\ 
  100 & X68 &  73 & 1.00 \\ 
  101 & X69 &  74 & 1.00 \\ 
  102 & X70 &  75 & 0.90 \\ 
  103 & X72 &  76 & 1.00 \\ 
  & & & \\
   \hline
        \end{tabular}
        \end{minipage} 
    }
    \label{UCEC:B:incprob}
\end{table}
    
     \begin{table}[H]
\centering
    \caption{Inclusion probability of each edge for the UCEC network graph indicating associations among proteins provided in the right panel of Figure \ref{fig:UCECnetwork05}.}
    \resizebox{0.9\textwidth}{!}{%
    \begin{minipage}{.4\linewidth}
      
      \centering
        \begin{tabular}{cccc}
            \hline
 & Protein  & Protein  & Inclusion  \\ 
 & & & Probability \\
  \hline
1 &   2 &   3 & 1.00 \\ 
  2 &   3 &   8 & 1.00 \\ 
  3 &   9 &  11 & 1.00 \\ 
  4 &   1 &  13 & 1.00 \\ 
  5 &   1 &  15 & 1.00 \\ 
  6 &   9 &  15 & 1.00 \\ 
  7 &   3 &  16 & 1.00 \\ 
  8 &  11 &  16 & 1.00 \\ 
  9 &  12 &  16 & 1.00 \\ 
  10 &   1 &  17 & 1.00 \\ 
  11 &   2 &  17 & 1.00 \\ 
  12 &  11 &  17 & 1.00 \\ 
  13 &   4 &  18 & 1.00 \\ 
  14 &  16 &  19 & 1.00 \\ 
  15 &   6 &  20 & 1.00 \\ 
  16 &  19 &  20 & 1.00 \\ 
  17 &   6 &  21 & 1.00 \\ 
  18 &  14 &  22 & 1.00 \\ 
  19 &  21 &  22 & 1.00 \\ 
  20 &  18 &  23 & 1.00 \\ 
  21 &   3 &  24 & 1.00 \\ 
  22 &  21 &  24 & 1.00 \\ 
  23 &  22 &  24 & 1.00 \\ 
  24 &  22 &  25 & 1.00 \\ 
  25 &   2 &  26 & 1.00 \\ 
  26 &   8 &  26 & 1.00 \\ 
  27 &   3 &  27 & 1.00 \\ 
  28 &  16 &  27 & 1.00 \\ 
  29 &   2 &  28 & 1.00 \\ 
  30 &   4 &  28 & 1.00 \\ 
  31 &   9 &  28 & 1.00 \\ 
  32 &  20 &  28 & 1.00 \\ 
  33 &   3 &  29 & 0.54 \\ 
  34 &   3 &  30 & 1.00 \\ 
  35 &   8 &  30 & 1.00 \\ 
  36 &  14 &  30 & 1.00 \\ 
  37 &  20 &  30 & 1.00 \\ 
  38 &  11 &  31 & 1.00 \\ 
  39 &  16 &  31 & 1.00 \\ 
  40 &   7 &  32 & 1.00 \\ 
  41 &  32 &  33 & 1.00 \\ 
  42 &   8 &  34 & 1.00 \\ 
  43 &  18 &  34 & 1.00 \\ 
  44 &  24 &  34 & 1.00 \\ 
  45 &  33 &  34 & 1.00 \\ 
  46 &  18 &  35 & 1.00 \\ 
  47 &  24 &  35 & 1.00 \\ 
   \hline
        \end{tabular}
    \end{minipage}%
    \quad \quad
    \begin{minipage}{.4\linewidth}
      \centering
         \begin{tabular}{cccc}
            \hline
 & Protein  & Protein  & Inclusion  \\ 
 & & & Probability\\
  \hline
 48 &  29 &  35 & 1.00 \\ 
  49 &  22 &  36 & 1.00 \\ 
  50 &  29 &  36 & 1.00 \\ 
  51 &   1 &  37 & 1.00 \\ 
  52 &  36 &  37 & 1.00 \\ 
  53 &   1 &  38 & 1.00 \\ 
  54 &   7 &  38 & 1.00 \\ 
  55 &  37 &  38 & 1.00 \\ 
  56 &  36 &  39 & 1.00 \\ 
  57 &   8 &  41 & 1.00 \\ 
  58 &  14 &  41 & 1.00 \\ 
  59 &  22 &  41 & 1.00 \\ 
  60 &  25 &  42 & 1.00 \\ 
  61 &  17 &  43 & 1.00 \\ 
  62 &  42 &  43 & 1.00 \\ 
  63 &  43 &  45 & 1.00 \\ 
  64 &  44 &  45 & 1.00 \\ 
  65 &   6 &  46 & 0.76 \\ 
  66 &  15 &  46 & 1.00 \\ 
  67 &  33 &  46 & 1.00 \\ 
  68 &  44 &  46 & 1.00 \\ 
  69 &   8 &  47 & 1.00 \\ 
  70 &  42 &  47 & 1.00 \\ 
  71 &  43 &  47 & 1.00 \\ 
  72 &  46 &  47 & 1.00 \\ 
  73 &  43 &  48 & 1.00 \\ 
  74 &  21 &  49 & 1.00 \\ 
  75 &  31 &  49 & 1.00 \\ 
  76 &   6 &  50 & 1.00 \\ 
  77 &  18 &  50 & 1.00 \\ 
  78 &  43 &  50 & 1.00 \\ 
  79 &  46 &  50 & 1.00 \\ 
  80 &   3 &  51 & 1.00 \\ 
  81 &  20 &  51 & 1.00 \\ 
  82 &  49 &  51 & 1.00 \\ 
  83 &  18 &  52 & 1.00 \\ 
  84 &  50 &  52 & 1.00 \\ 
  85 &  51 &  54 & 1.00 \\ 
  86 &  53 &  54 & 1.00 \\ 
  87 &   4 &  55 & 1.00 \\ 
  88 &   6 &  55 & 1.00 \\ 
  89 &  13 &  55 & 1.00 \\ 
  90 &  45 &  55 & 1.00 \\ 
  91 &  52 &  55 & 1.00 \\ 
  92 &  54 &  55 & 1.00 \\ 
  93 &   4 &  56 & 1.00 \\ 
  94 &   2 &  57 & 1.00 \\ 
  \hline
  \end{tabular}
    \end{minipage}%
    \quad \quad
    \begin{minipage}{.4\linewidth}
      \centering
         \begin{tabular}{cccc}
            \hline
 & Protein  & Protein  & Inclusion  \\ 
 & & & Probability\\
  \hline
  95 &   6 &  57 & 1.00 \\ 
  96 &  39 &  57 & 0.53 \\ 
  97 &  56 &  57 & 1.00 \\ 
  98 &  50 &  58 & 1.00 \\ 
  99 &   1 &  59 & 1.00 \\ 
  100 &   3 &  59 & 1.00 \\ 
  101 &  14 &  59 & 1.00 \\ 
  102 &  44 &  59 & 1.00 \\ 
  103 &  18 &  60 & 1.00 \\ 
  104 &  58 &  60 & 1.00 \\ 
  105 &  15 &  61 & 1.00 \\ 
  106 &  59 &  61 & 1.00 \\ 
  107 &   3 &  62 & 1.00 \\ 
  108 &  49 &  62 & 1.00 \\ 
  109 &  56 &  62 & 1.00 \\ 
  110 &  57 &  62 & 1.00 \\ 
  111 &  58 &  62 & 1.00 \\ 
  112 &  47 &  63 & 1.00 \\ 
  113 &  51 &  63 & 1.00 \\ 
  114 &  58 &  63 & 1.00 \\ 
  115 &  45 &  64 & 1.00 \\ 
  116 &  53 &  64 & 1.00 \\ 
  117 &  54 &  64 & 0.75 \\ 
  118 &  55 &  64 & 1.00 \\ 
  119 &  62 &  64 & 1.00 \\ 
  120 &  63 &  64 & 1.00 \\ 
  121 &  12 &  65 & 1.00 \\ 
  122 &  21 &  65 & 1.00 \\ 
  123 &  33 &  65 & 1.00 \\ 
  124 &  62 &  65 & 1.00 \\ 
  125 &   6 &  66 & 1.00 \\ 
  126 &  23 &  66 & 1.00 \\ 
  127 &  38 &  66 & 1.00 \\ 
  128 &  46 &  66 & 1.00 \\ 
  129 &  65 &  66 & 1.00 \\ 
  130 &  16 &  67 & 1.00 \\ 
  131 &  34 &  67 & 0.95 \\ 
  132 &  61 &  67 & 1.00 \\ 
  133 &  64 &  67 & 1.00 \\ 
  134 &  66 &  67 & 1.00 \\ 
  135 &  22 &  68 & 1.00 \\ 
  136 &  57 &  68 & 1.00 \\ 
  137 &  62 &  68 & 1.00 \\ 
  138 &   2 &  69 & 1.00 \\ 
  139 &  15 &  69 & 1.00 \\ 
  140 &  16 &  69 & 1.00 \\ 
  141 &  18 &  69 & 1.00 \\ 
  \hline
  \end{tabular}
    \end{minipage}%
    \quad \quad
    \begin{minipage}{.4\linewidth}
      \centering
         \begin{tabular}{cccc}
            \hline
 & Protein  & Protein  & Inclusion  \\ 
 & & & Probability\\
  \hline
  142 &  21 &  69 & 1.00 \\ 
  143 &  48 &  69 & 1.00 \\ 
  144 &  64 &  69 & 1.00 \\ 
  145 &  21 &  70 & 0.65 \\ 
  146 &  37 &  70 & 1.00 \\ 
  147 &  53 &  70 & 1.00 \\ 
  148 &  56 &  70 & 1.00 \\ 
  149 &  57 &  70 & 1.00 \\ 
  150 &  68 &  70 & 1.00 \\ 
  151 &  69 &  70 & 1.00 \\ 
  152 &  20 &  71 & 1.00 \\ 
  153 &  56 &  71 & 0.79 \\ 
  154 &  70 &  71 & 1.00 \\ 
  155 &  17 &  72 & 1.00 \\ 
  156 &  45 &  72 & 1.00 \\ 
  157 &  57 &  72 & 1.00 \\ 
  158 &  66 &  72 & 1.00 \\ 
  159 &  14 &  73 & 1.00 \\ 
  160 &  26 &  73 & 0.94 \\ 
  161 &  29 &  73 & 1.00 \\ 
  162 &  35 &  73 & 0.90 \\ 
  163 &  65 &  73 & 1.00 \\ 
  164 &   4 &  74 & 1.00 \\ 
  165 &   7 &  74 & 1.00 \\ 
  166 &  25 &  74 & 1.00 \\ 
  167 &  31 &  74 & 1.00 \\ 
  168 &  65 &  74 & 1.00 \\ 
  169 &  73 &  74 & 1.00 \\ 
  170 &   9 &  75 & 1.00 \\ 
  171 &  27 &  75 & 1.00 \\ 
  172 &  30 &  75 & 1.00 \\ 
  173 &  31 &  75 & 1.00 \\ 
  174 &  73 &  75 & 1.00 \\ 
  175 &   2 &  76 & 1.00 \\ 
  176 &   8 &  76 & 1.00 \\ 
  177 &  15 &  76 & 1.00 \\ 
  178 &  19 &  76 & 1.00 \\ 
  179 &  24 &  76 & 1.00 \\ 
  180 &  13 &  77 & 1.00 \\ 
  181 &  17 &  77 & 1.00 \\ 
  182 &  23 &  77 & 1.00 \\ 
  183 &  26 &  77 & 1.00 \\ 
  184 &  46 &  77 & 1.00 \\ 
  185 &  48 &  77 & 1.00 \\ 
  186 &  63 &  77 & 1.00 \\ 
  & & & \\
  & & & \\
   \hline
        \end{tabular}
        \end{minipage} 
    }
    \label{UCEC:Omega:incprob}
\end{table}

\bibliographystyle{chicago}
\bibliography{refs,myrefs}

\end{document}